\documentclass[12pt]{article}

\usepackage{fullpage}
\usepackage{amsmath,amsthm,amssymb,amsfonts,MnSymbol}
\usepackage{mathrsfs}
\usepackage{color}
\usepackage[backref]{hyperref}
\usepackage{enumitem}
\usepackage{booktabs}
\usepackage{lineno}

\hypersetup{colorlinks=true, linkcolor=red,citecolor=blue,urlcolor=blue}

\newtheorem{lemma}{Lemma}
\newtheorem{claim}{Claim}
\newtheorem{corollary}{Corollary}
\newtheorem{theorem}{Theorem}
\newtheorem{proposition}{Proposition}
\newtheorem{definition}{Definition}

\newcounter{example}

\newenvironment{example}{\refstepcounter{example}\par\bigskip
\noindent\textit{Example~\theexample.}\rmfamily}{\hfill$\dashv$\medskip}

\newcommand{\supp}{{\mathrm{Supp}}}
\newcommand{\tuples}{{\mathrm{Tup}}}
\newcommand{\domain}{{\mathrm{Dom}}}

\newcommand{\ftp}{transportation property}
\newcommand{\icp}{inner consistency property}
\newcommand{\ltgc}{local-to-global consistency property}

\newcommand{\norm}[1]{\Vert #1 \Vert}

\newcommand{\free}[1]{\mathbb{F}(#1)}

\newcommand{\transpose}{\mathrm{T}}
\newcommand{\Id}{\mathrm{I}}
\newcommand{\trace}{\mathrm{tr}}

\newcommand{\onto}{\stackrel{\scriptscriptstyle{s}}\rightarrow}

\newcommand{\standardjoin}[1]{\Join_{#1,\mathrm{S}}}
\newcommand{\vorobyevjoin}[1]{\Join_{#1,\mathrm{V}}}
\newcommand{\componentwisejoin}[2]{\Join^{#1}_{#2}}

\newcommand{\commentout}[1]{}

\newcommand{\newatop}[2]{\genfrac{}{}{0pt}{2}{#1}{#2}}

\usepackage{thmtools,thm-restate}

\begin{document}

\title{\bf Consistency of Relations over Monoids}

\author{Albert Atserias\footnote{Universitat Polit\`ecnica de Catalunya and Centre de Recerca Matem\`atica, Barcelona, Spain. }
\and
Phokion G. Kolaitis\footnote{University of California Santa Cruz and IBM Research,
 California, USA}}

\maketitle

\begin{abstract}

The interplay between local consistency and global consistency has
been the object of study in several different areas, including
probability theory, relational databases, and quantum information. For
relational databases, Beeri, Fagin, Maier, and Yannakakis showed that
a database schema is acyclic if and only if it has the local-to-global 
consistency property for relations, which means that every
collection of pairwise consistent relations over the schema is
globally consistent. More recently, the same result has been shown
under bag semantics. In this paper, we carry out a systematic study
of local vs.\ global consistency for relations over positive
commutative monoids, which is a common generalization of ordinary
relations and bags. Let $\mathbb K$ be an arbitrary positive
commutative monoid. We begin by showing that acyclicity of the schema
is a necessary condition for the local-to-global consistency property
for $\mathbb K$-relations to hold. Unlike the case of ordinary
relations and bags, however, we show that acyclicity is not always
sufficient. After this, we characterize the positive commutative
monoids for which acyclicity is both necessary and sufficient for the
local-to-global consistency property to hold; this characterization
involves a combinatorial property of monoids, which we call the
\emph{transportation property}. We then identify several different
classes of monoids that possess the transportation property. As our
final contribution, we introduce a modified notion of local
consistency of $\mathbb{K}$-relations, which we call \emph{pairwise
consistency up to the free cover}. We prove that, for all positive
commutative monoids $\mathbb{K}$, even those without the
transportation property, acyclicity is both necessary and sufficient
for every family of $\mathbb{K}$-relations that is pairwise consistent
up to the free cover to be globally consistent.

\end{abstract}

\newpage

\section{Introduction}  The interplay between local consistency and global consistency has been investigated
in several different settings. In each such setting, the concepts ``local", ``global", and ``consistent" are defined rigorously and  a study is carried out as to when objects that are locally consistent are also globally consistent. In probability theory, Vorob'ev
\cite{vorob1962consistent} studied when, for  a collection of probability distributions on overlapping sets of variables, there is  a global probability distribution whose marginals  coincide with the probability distributions in that  collection. In quantum mechanics, Bell's theorem \cite{bell1964einstein} is about \emph{contextuality} phenomena, where   empirical local measurements  may be locally consistent but there is no global explanation for  these measurements in terms of hidden local variables. In relational databases,
there has been an extensive study of 
 the universal relation problem \cite{DBLP:journals/tods/AhoBU79,DBLP:journals/ipl/HoneymanLY80,DBLP:conf/pods/Ullman82}: given  relations $R_1,\ldots,R_m$, is there a relation
$W$ such that, for each relation $R_i$, the projection of $W$ on the
attributes of $R_i$ is equal to $R_i$? 
If the answer is positive, the relations $R_1,\ldots,R_m$ are said to be \emph{globally consistent} and $W$ is  a \emph{universal relation} for them.
Note that if the relations  $R_1,\ldots,R_m$ are globally consistent, then they are
\emph{pairwise consistent} (i.e., every two of them are globally consistent), but the converse need not hold. 

Beeri, Fagin, Maier, and Yannakakis
\cite{BeeriFaginMaierYannakakis1983} showed that a relational schema is \emph{acyclic} 
if and only if the \emph{local-to-global consistency property for relations}
over that schema holds, which means that every collection of pairwise
consistent relations over the schema is globally consistent. Thus, for
acyclic schemas, pairwise consistency and global consistency coincide. Note that set semantics is used in this result, i.e., the result is about ordinary relations. More recently, in \cite{DBLP:conf/pods/AtseriasK21}  it was shown that an analogous result holds also under bag semantics:  a relational schema is acyclic if and only if the \ltgc~for bags holds, where in the definitions of pairwise consistency and global consistency for bags, the projection operation adds the multiplicities of all tuples in the relation that are projected to the same tuple. 
 It should be pointed out, however, that there are significant differences between set semantics and bag semantics as regards consistency properties. In particular, under set semantics, the relational join of two consistent relations is the largest witness of their consistency, while, under bag semantics, the join of two consistent bags  need not even be a witness of their consistency \cite{DBLP:conf/pods/AtseriasK21}.

During the past two decades and starting with the influential paper \cite{DBLP:conf/pods/GreenKT07}, there has been a growing study of $\mathbb K$-relations, where tuples in $\mathbb K$-relations are annotated with values from the universe of a fixed semiring $\mathbb K$.  Clearly, ordinary relations are $\mathbb B$-relations, where $\mathbb B$ is the Boolean semiring, while bags are $\mathbb N$-relations, where $\mathbb N$ is the semiring of non-negative integers. Originally, $\mathbb K$-relations were studied in the context of provenance in databases \cite{DBLP:conf/pods/GreenKT07}; since that time, the study has been expanded to other fundamental problems in databases, including the query containment problem \cite{DBLP:journals/mst/Green11,DBLP:journals/tods/KostylevRS14}. Note that in the study of both provenance and  query containment, the definitions of the basic concepts  involve  both the addition operation and the multiplication operation of the semiring $\mathbb K$.

Aiming to obtain a common generalization of the results in  \cite{BeeriFaginMaierYannakakis1983} and in \cite{DBLP:conf/pods/AtseriasK21},  we carry out a systematic investigation  of local consistency vs.\ global consistency for relations whose tuples are annotated with values from the universe of some suitable algebraic structure. At first sight, semirings appear to be the most general algebraic structures for this purpose. Upon closer reflection, however,  one realizes that the definition of a projection of $\mathbb K$-relation involves only the addition operation of the semiring (and not the multiplication operation), hence so do the definitions of the notions of local and global consistency for $\mathbb K$-relations. For this reason, we embark on a study of the interplay between local vs.\ global consistency for $\mathbb K$-relations, where ${\mathbb K}=(K,+,0)$ is a commutative monoid. In addition, we  require the monoid $\mathbb K$ to be \emph{positive}, which means that the  sum of non-zero elements from $K$ is non-zero. This condition is needed in  key technical results, but it also ensures that the support of the projection of a $\mathbb K$-relation is equal to the support of that relation.

 Let $\mathbb K$ be an arbitrary positive commutative monoid.
Our first result asserts that 
if a hypergraph $H$ is not acyclic, then there is a collection of pairwise consistent $\mathbb K$-relations over $H$ that  are not globally consistent; in other words, acyclicity is a necessary condition for the \ltgc~for $\mathbb K$-relations to hold. The construction of such $\mathbb K$-relations is similar to the one used for bags in \cite{DBLP:conf/pods/AtseriasK21}, which, in turn, was inspired from
 an earlier
construction of hard-to-prove tautologies in propositional logic by
Tseitin \cite{Tseitin1968}.

Unlike the Boolean monoid $\mathbb B$ (case of ordinary relations) and the monoid $\mathbb N$ of non-negative integers (case of bags), however, we show that there are positive commutative monoids $\mathbb K$ 
%(actually, both finite and infinite ones \albert{do we need this remark?}) 
for which acyclicity is not a sufficent condition  for the \ltgc~for $\mathbb K$-relations to hold.
%Actually, it turns out that there are both finite and infinite  positive commutative monoids for which acyclicity is not enough. 
We then go on to
characterize the positive commutative monoids for which acyclicity is both necessary and sufficient for the \ltgc~to hold.  In fact, we obtain two different characterizations, a semantic one, which we call the \emph{inner consistency property}, and a combinatorial one, which we call the 
\emph{transportation property}. The inner consistency property asserts that if two $\mathbb K$-relations have the same projection on the set of their common attributes, then they are consistent (note that the converse is always true). The transportation property asserts that every balanced instance of the transportation problem with values from $\mathbb K$ has a solution in $\mathbb K$; these concepts and the terminology are as in the well-studied transportation problem in linear programming.

We then identify several different classes of monoids that possess the transportation property. 
Special cases include the Boolean  monoid $\mathbb B$, the  monoid $\mathbb N$ of non-negative integers,  the monoid ${\mathbb R}^{\geq 0}$ of the non-negative real numbers with addition, the monoids obtained by restricting tropical semirings to their additive structure, 
%product monoids 
 various monoids of provenance polynomials, and the free commutative monoid on a set of indeterminates.  Furthermore, for each such class of monoids, we give either an explicit construction or a procedure for computing a witness to the consistency of two consistent $\mathbb K$-relations.

After this extended investigation of classes of positive commutative monoids with the transportation property, we revisit the broader
question of characterizing the \ltgc~for
collections of $\mathbb K$-relations on acyclic schemas for \emph{arbitrary} positive
commutative monoids $\mathbb K$. By the ``no-go examples" in the first part of the
paper, we know that any such characterization that applies to all
positive commutative monoids must  either require  more than just pairwise
consistency or settle for less than global consistency. 

In \cite{AtseriasK23}, 
the second scenario was explored. Specifically, by relaxing the notion of
consistency to what was called there \emph{consistency up to normalization}, it was shown
that the \ltgc~up to normalization holds precisely for the acyclic
schemas. While this result is a common generalization of the theorems  by  Vorob'ev \cite{vorob1962consistent} and by Beeri et al.~\cite{BeeriFaginMaierYannakakis1983} (because for ordinary relations and for probability distributions
the relaxed concept of consistency up to normalization agrees with the
standard one), it fails to generalize the \ltgc~for bags from \cite{DBLP:conf/pods/AtseriasK21}.
 Furthermore, the definition of this relaxed  notion of consistency
required $\mathbb K$ to come equipped with a multiplication operation making it
into a positive semiring, hence the result in \cite{AtseriasK23} does not apply to arbitrary positive commutative monoids.

Here, we explore the first scenario by
introducing a stronger notion of consistency, which we call 
 \emph{consistency up to the free cover} (the term reflects the role that the free commutative monoid plays in the definition of this notion). 
 First, we prove
that the \ltgc~with consistency strengthened to consistency up to the free
cover holds precisely for the acyclic schemas. Second and perhaps 
unexpectedly, by exploiting the universal property of the free
commutative monoid, we establish  that the  notion of
global consistency up to the free cover is \emph{absolute}, in the sense that
global consistency holds up to the free cover if and only if it holds in
the standard sense. As a consequence,
we have that for every positive commutative monoid $\mathbb K$, 
a schema $H$ is acyclic 
 precisely when
every collection of $\mathbb K$-relations over $H$ that is pairwise consistent up to the
free cover is indeed globally consistent. 
Vice versa,  every
 collection of $\mathbb K$-relations that is globally consistent is pairwise
 consistent up to the free cover.
We view these results as an answer to the
question of characterizing the global consistency of relations for
acyclic schemas in the broader setting of relations over arbitrary
positive commutative monoids.

\section{Preliminaries} \label{sec:prelims}

\paragraph{Positive Commutative Monoids} 
A \emph{commutative monoid} is a structure $\mathbb{K}=(K,+,0)$, where $+$ is a binary operation on the universe $K$ of $\mathbb K$  that is
 associative, commutative, and has $0$ as its neutral  element, i.e., $p+ 0 = p = 0 + p$ holds for all $p\in K$.  A \emph{positive commutative monoid} is a commutative monoid ${\mathbb K}=(K,+,0)$ such that for all elements $p,q\in K$ with
$p+q=0$, we have that  $p=0$ and $q=0$.  To avoid trivialities, we will assume that all commutative monoids considered have at least two elements in their universe.

As an example, the structure $\mathbb{B} = (\{0,1\},\vee,0)$ with disjunction $\vee$ as its operation and~$0$~(false) as its neutral element is a positive commutative monoid. Other examples of positive commutative monoids include the structures $\mathbb{N}=(Z^{\geq 0}, +, 0,)$, $\mathbb{Q}^{\geq 0}=(Q^{\geq 0}, +,  0)$, $\mathbb{R}^{\geq 0}=(R^{\geq 0}, +, 0)$, where $Z^{\geq 0}$ is the set of non-negative integers, $Q^{\geq 0}$ is the set of non-negative rational numbers, $R^{\geq 0}$ is the set of non-negative real numbers, and $+$ is the standard addition operation. In contrast, the structure ${\mathbb Z}=(Z,+,0)$, where $Z$ is the set of integers, is a commutative monoid, but not a positive one. Two examples of positive commutative monoids of different flavor are the structures ${\mathbb T}= (R\cup\{\infty\}, \min, \infty)$ and ${\mathbb V}=([0,1], \max, 0)$, where $R$ is the set of real numbers, and $\min$ and $\max$ are the standard minimum and maximum operations.
Finally, if $A$ is a set and $\mathcal{P}(A)$ is its powerset, then the structure
${\mathbb P}(A)= (\mathcal{P}(A), \cup,\emptyset)$ is a positive commutative monoid, where $\cup$ is the union operation on sets.

\paragraph{Definition of~$\mathbb K$-relations and their Marginals}
An \emph{attribute}~$A$ is a symbol with an associated
set~$\domain(A)$,  called its \emph{domain}. If~$X$ is a finite set of
attributes, then we write~$\tuples(X)$ for the set
of~\emph{$X$-tuples}, i.e.,~$\tuples(X)$ is the set of
functions that take each attribute~$A \in X$ to an element of its
domain~$\domain(A)$. Note that~$\tuples(\emptyset)$ is non-empty as it
contains the \emph{empty tuple}, i.e., the unique function with empty
domain. If~$Y \subseteq X$ is a subset of attributes and~$t$ is
an~$X$-tuple, then the \emph{projection of~$t$ on~$Y$}, denoted
by~$t[Y]$, is the unique~$Y$-tuple that agrees with~$t$ on~$Y$. In
particular,~$t[\emptyset]$ is the empty tuple.

Let~${\mathbb K} = (K,+,0)$ be a positive commutative monoid and let~$X$ be a finite set
of attributes. A~\emph{$\mathbb{K}$-relation over~$X$} is a
function~$R : \tuples(X) \rightarrow K$ that assigns a value~$R(t)$ in~$K$
to every~$X$-tuple~$t$ in~$\tuples(X)$. 
We will often write $R(X)$ to indicate that $R$ is a $\mathbb K$-relation over $X$, and we will refer to $X$ as the set of attributes of $R$.
These notions 
make sense even if~$X$ is the empty set of attributes, in which case a~$\mathbb{K}$-relation 
over~$X$ is simply a single value from~$K$ that
is assigned to the empty tuple. 
Clearly, the $\mathbb B$-relations are just the ordinary relations, while the $\mathbb N$-relations are the \emph{bags} or 
\emph{multisets}, i.e., each tuple has a non-negative integer associated with it that denotes the  \emph{multiplicity} of the tuple.

The \emph{support} of
a~$\mathbb{K}$-relation~$R(X)$, denoted by~$\supp(R)$, is the set
of~$X$-tuples~$t$ that are assigned non-zero value, i.e.,
\begin{equation}
\supp(R) := \{ t \in \tuples(X) : R(t) \not= 0 \}. \label{def:support}
\end{equation}
Whenever this does not lead to confusion, we write~$R'$ to
denote~$\supp(R)$. Note that~$R'$ is an ordinary relation
over~$X$. A~$\mathbb{K}$-relation is \emph{finitely supported} if its support is a
finite set. In this paper, all~$\mathbb{K}$-relations considered will be  finitely supported, 
and we omit the term; thus, from now on, a $\mathbb{K}$-relation is a finitely supported $\mathbb{K}$-relation. 
When~$R'$ is empty, we say that~$R$ is the empty~$\mathbb{K}$-relation over~$X$. 

If~$Y\subseteq X$, then the \emph{marginal $R[Y]$ of $R$ on $Y$}  is the~$\mathbb{K}$-relation 
over~$Y$ such that for every~$Y$-tuple~$t$, we have that
   \begin{equation}
R[Y](t) := \sum_{\newatop{r \in R':}{r[Y] = t}} R(r). \label{eqn:marginal}
\end{equation}
The value $R[Y](t)$ is  the \emph{marginal of $R$ over $t$}. In what follows and for notational simplicity, we will often write $R(t)$ for the marginal of $R$ over $t$, instead of $R[Y](t)$. It will be clear from the context (e.g., from the arity of the tuple $t$) if $R(t)$ is indeed the marginal of $R$ over $t$ (in which case $t$ must be a $Y$-tuple) or $R(t)$ is the actual value of $R$ on $t$ as a mapping from $\tuples(X)$ to $K$ (in which case $t$ must be an $X$-tuple).
   Note that if $R$ is an ordinary relation (i.e., $R$ is a $\mathbb B$-relation), then the marginal $R[Y]$ is the projection of $R$ on $Y$.

\begin{lemma} \label{lem:easyfacts1} Let 
$\mathbb K$ be a positive commutative monoid and let 
$R(X)$ be a~$\mathbb K$-relation. The
  following statements hold:
  \begin{enumerate} \itemsep=0pt
  \item For all~$Y \subseteq X$, we have~$R'[Y] = R[Y]'$.
  \item For all $Z \subseteq Y \subseteq X$, we have $R[Y][Z] = R[Z]$.
\end{enumerate}
\end{lemma}

\begin{proof}
  For the first part, the inclusion~$R[Y]' \subseteq R'[Y]$ is obvious. For the converse, assume that~$t \in R'[Y]$, so
  there exists~$r$ such that~$R(r) \not= 0$ and~$r[Y] =
  t$. By~\eqref{eqn:marginal} and the positivity of~$\mathbb K$, we have
  that~$R(t) \not= 0$. Hence~$t \in R[Y]'$.  
  
  For the second part, we have
  \begin{equation}
  R[Y][Z](u) = \sum_{\newatop{v \in R[Y]':}{v[Z]=u}} R[Y](v) =
  \sum_{\newatop{v \in R'[Y]:}{v[Z]=u}} \sum_{\newatop{w \in R':}{w[Y]=v}} R(w) =
  \sum_{\newatop{w \in R':}{w[Z]=u}} R(w) = R[Z](u)
\end{equation}
where the first equality follows from~\eqref{eqn:marginal}, the second
follows from the first part of this lemma to replace~$R[Y]'$ by~$R'[Y]$, and
again~\eqref{eqn:marginal}, the third follows from partitioning the
tuples in~$R'$ by their projection on~$Y$, together
with~$Z \subseteq Y$, and the fourth follows from~\eqref{eqn:marginal}
again.
\end{proof}

If~$X$ and~$Y$ are sets of attributes, then we write~$XY$ as
shorthand for the union~$X \cup Y$. Accordingly, if~$x$ is
an~$X$-tuple and~$y$ is a~$Y$-tuple with the property
that~$x[X \cap Y] = y[X \cap Y]$, then we write~$xy$ to denote
the~$XY$-tuple that agrees with~$x$ on~$X$ and on~$y$ on~$Y$.  We
say that~\emph{$x$ joins with~$y$}, and that~\emph{$y$ joins
  with~$x$}, to \emph{produce} the tuple~$xy$.

A \emph{schema} is a sequence~$X_1,\ldots,X_m$ of sets of attributes. A  \emph{collection of
$\mathbb{K}$-relations over the schema~$X_1,\ldots,X_m$} is a sequence $R_1(X_1),\ldots,R_m(X_m)$ of~$\mathbb{K}$-relations,
where~$R_i(X_i)$ is a~$\mathbb{K}$-relation over~$X_i$, for $i = 1,\ldots,m$.

\paragraph{Homomorphisms, Subalgebras, Products, and Varieties}  
For later reference, we introduce some basic 
terminology from universal algebra for the particular case of monoids. 

If  $\mathbb{M}_1 = (M_1,+_1,0_1)$ and $\mathbb{M}_2 = (M_2,+_2,0_2)$
are monoids, then a \emph{homomorphism} from $\mathbb{M}_1$ to $\mathbb{M}_2$ is a map $h : M_1 \to M_2$ such
that $h(0_1) = 0_2$ and
\begin{equation*}
    h(a +_1 b) = h(a) +_2 h(b)
\end{equation*}
holds for all $a,b \in M_1$. The homomorphism
is \emph{surjective} if $h$ is surjective, i.e., if for all $b \in M_2$ there exists $a \in M_1$ such that
$h(a) = b$. If $h$ is a surjective homomorphism from $\mathbb{M}_1$ to $\mathbb{M}_2$ then we say
that $\mathbb{M}_2$ is a \emph{homomorphic image} of $\mathbb{M}_1$, and we write $h : \mathbb{M}_1 \onto \mathbb{M}_2$ to denote this fact.
An \emph{isomorphism} is a bijection $h : M_1 \to M_2$ such that both $h$ and its its inverse  $h^{-1}$ are homomorphisms.
We say that $\mathbb{M}_1$ is a \emph{subalgebra} of $\mathbb{M}_2$ if $M_1 \subseteq M_2$ with $0_1 = 0_2$
and $M_1$ is closed under $+_2$, that is, for all $a,b,c \in M_1$, if $a +_2 b = c$, then $c \in M_1$. If $I$ is a finite or infinite set of indices and 
$(\mathbb{M}_i : i \in I)$ is an indexed set of monoids,
then the \emph{product} monoid~$\prod_{i \in I} \mathbb{M}_i$
is defined as follows. The domain of~$\prod_{i \in I} \mathbb{M}_i$ is the product set $\prod_{i \in I} M_i$, where $M_i$ is the domain of $\mathbb{M}_i$, that is,
the elements  of the product monoid are
the maps $f$ with domain $I$ that map each index $i \in I$ to an element $f(i) \in M_i$; 
the operation $+$ of the product monoid is defined pointwise:
for two maps $f$ and $g$ in $\prod_{i \in I} M_i$, the sum $f+g$
is defined by the equation
\begin{equation}
    (f + g)(i) = f(i) +_i g(i) \label{eqn:pointwise}
\end{equation}
for all $i \in I$, where the addition operation $+_i$
on the right-hand side is over $\mathbb{M}_i$; finally, the neutral element $0$ of the product monoid is the map that maps $i \in I$ to $0_i$, where 
$0_i$ is the neutral element of $\mathbb{M}_i$. The special case of a product monoid
in which every factor $\mathbb{M}_i$ is the same monoid $\mathbb{M}$ is called an \emph{$I$-power
of $\mathbb{M}$} and is denoted by $\mathbb{M}^I$; furthermore, its domain is denoted by $M^I$. In
the special case in which the index set $I$ has the form $[k] = \{1,\ldots,k\}$ for some 
natural number $k$, we write $\mathbb{M}^k$ and
$M^k$, instead of $\mathbb{M}^{[k]}$ and $M^{[k]}$, respectively.

A \emph{variety of monoids} is a class of monoids that is closed 
under homomorphic images, subalgebras, and products. 
By Birkhoff's  HSP theorem \cite{birkhoff1935structure}, a class of monoids is a variety if and only if it is the class of all monoids that satisfy a set of identities (for a modern exposition of this classical result, see \cite{DBLP:books/daglib/0067494}).
For example, 
the class of commutative monoids is a variety.
In contrast, the class of positive commutative monoids is not a variety because it is not closed under homomorphic images. Indeed, the map
that sends each non-negative integer $n$ to its residue class mod~$2$ is
a surjective homomorphism from the positive commutative monoid $\mathbb{N} = (Z^{\geq 0},+,0)$
onto the structure $\mathbb{Z}/2\mathbb{Z} = (\{0,1\},\oplus,0)$, where $\oplus$ is 
addition mod~$2$. The latter is a commutative monoid but it 
is not positive because $1 \oplus 1 = 0$.

\section{Consistency over Positive Commutative Monoids} 
\label{sec:semiring-cons}
  
The following definitions are the direct generalizations of the standard notions of consistency
for collections of ordinary relations to collections of $\mathbb{K}$-relations, where $\mathbb{K}$ is
an arbitrary positive commutative monoid.
Recall that a \emph{schema} is a collection $X_1,\ldots,X_m$ of sets of attributes. 

\begin{definition}\label{defn:consistency}
Let $\mathbb{K}$ be a positive commutative monoid, let
$X_1,\ldots,X_m$ be a schema, let
$R_1(X_1),\ldots,R_m(X_m)$ be a collection
of~$\mathbb{K}$-relations over $X_1,\ldots,X_m$, 
and let~$k$ be a positive integer.
We say that the collection $R_1,\ldots,R_m$ is \emph{$k$-wise consistent} if
for all $q \in [k]$ and $i_1, \ldots, i_q  \in [m]$ 
there exists a~$\mathbb{K}$-relation~$W(X_{i_1} \cdots X_{i_q})$ such 
that~$W[X_i] = R_i$ holds for all $i \in [q]$. 
If~$k = 2$, then we
say that the collection~$R_1,\ldots,R_m$ is \emph{pairwise consistent}.
If~$k = m$,
then we say that the collection~$R_1,\ldots,R_m$ is \emph{globally consistent}. In all such cases we say
that~$W(X_{i_1}\cdots X_{i_q})$ \emph{witnesses} the consistency
of~$R_{i_1},\ldots,R_{i_q}$. 
\end{definition}

From Definition \ref{defn:consistency}, 
it follows that if a collection
of~$\mathbb{K}$-relations is~$(k+1)$-wise consistent,
then it is also~$k$-wise consistent.  In particular, if a collection
of~$\mathbb{K}$-relations is globally consistent, then it is also pairwise
consistent. Our goal in this paper is to investigate when the converse is true. In other words, 
we focus on the following question:
under what
conditions on the positive commutative monoid ~$\mathbb K$ and on the schema~$X_1,\ldots,X_m$ is it
the case that every collection of~$\mathbb{K}$-relations of
schema~$X_1,\ldots,X_m$ that is pairwise consistent is also globally
consistent?  Our investigation begins by identifying a very broad necessary condition.

\subsection{Acyclicity is Always Necessary}

To formulate  the necessary condition, we need to introduce some
terminology. A \emph{hypergraph} is a pair~$H = (V,E)$, where~$V$ is a
set of \emph{vertices} and~$E$ is a set of \emph{hyperedges}, each of
which is a non-empty subset of~$V$.  Every collection~$X_1,\ldots,X_m$
of sets of attributes can be identified with a hypergraph~$H=(V,E)$,
where~$V = X_1\cup \cdots \cup X_m$ and~$E
=\{X_1,\ldots,X_m\}$. Conversely, every hypergraph~$H = (V,E)$ gives
rise to a collection~$X_1,\ldots,X_m$ of sets of attributes,
where~$X_1,\dots,X_m$ are the hyperedges of~$H$. Thus, we can move
seamlessly between collections of sets of attributes and hypergraphs.

\paragraph{Acyclic Hypergraphs} The notion of an \emph{acyclic}
hypergraph generalizes the notion of an acyclic graph. Since we will
not work directly with the definition of an acyclic hypergraph, we
refer the reader to~\cite{BeeriFaginMaierYannakakis1983} for the
precise definition. Instead, we focus on other notions that are
equivalent to hypergraph acyclicity and will be of interest to us in
the sequel.

\paragraph{Conformal and Chordal Hypergraphs} The \emph{primal} graph
of a hypergraph~$H = (V,E)$ is the undirected graph that has~$V$ as
its set of vertices and has an edge between any two distinct vertices
that appear together in at least one hyperedge of~$H$. A
hypergraph~$H$ is \emph{conformal} if the set of vertices of every
clique (i.e., complete subgraph) of the primal graph of~$H$ is
contained in some hyperedge of~$H$.  A hypergraph~$H$ is
\emph{chordal} if its primal graph is chordal, that is, if every cycle
of length at least four of the primal graph of~$H$ has a chord.  To
illustrate these concepts, let~$V_n=\{A_1,\ldots,A_n\}$ be a set
of~$n$ vertices and consider the hypergraphs
\begin{eqnarray}
P_n &  = &  (V_n, \{ {A_1,A_2}\}, \ldots, \{A_{n-1},A_n\})  \label{path-hyper}\\
C_n & = & (V_n, \{A_1,A_2\}, \ldots, \{A_{n-1},A_n\}, \{A_n,A_1\}) \label{cycle-hyper}\\
H_n & = & (V_n, \{ V_n\setminus \{A_i\}: 1\leq i\leq n \}) \label{clique-hyper}
\end{eqnarray}
If~$n\geq 2$, then the hypergraph~$P_n$ is both conformal and chordal.
The hypergraph~$C_3 = H_3$ is chordal, but not conformal. For
every~$n\geq 4$, the hypergraph~$C_n$ is conformal, but not chordal,
while the hypergraph~$H_n$ is chordal, but not conformal.

\paragraph{Running Intersection Property} We say that a hypergraph~$H$
has the \emph{running intersection property} if there is a
listing~$X_1,\ldots,X_m$ of all hyperedges of~$H$ such that for
every~$i \in [m]$ with~$i \geq 2$, there exists a~$j \in \{1,\ldots,i-1\}$ such
that~$X_i \cap (X_1 \cup \cdots \cup X_{i-1}) \subseteq X_j$.

\paragraph{Join Tree} A \emph{join tree} for a hypergraph~$H$ is an
undirected tree~$T$ with the set~$E$ of the hyperedges of~$H$ as its
vertices and such that for every vertex~$v$ of~$H$, the set of
vertices of~$T$ containing~$v$ forms a subtree of~$T$, i.e., if~$v$
belongs to two vertices~$X_i$ and~$X_j$ of~$T$, then~$v$ belongs to
every vertex of~$T$ in the unique simple path from~$X_i$ to~$X_j$ in~$T$.

\paragraph{Local-to-Global Consistency Property for Relations}
Let~$H$ be a hypergraph and let~$X_1,\dots,X_m$ be a listing of all
hyperedges of~$H$. We say that~$H$ has the \emph{\ltgc~for relations}
if every collection~$R_1(X_1),\ldots,R_m(X_m)$ of relations of
schema~$X_1,\ldots,X_m$ that is pairwise consistent is also globally
consistent.

We are now ready to state the main result in Beeri et al.~\cite{BeeriFaginMaierYannakakis1983}.

\begin{theorem} [Theorem 3.4 in \cite{BeeriFaginMaierYannakakis1983}] \label{thm:BFMY}
Let~$H$ be a hypergraph. The following statements are equivalent:
\begin{enumerate} \itemsep=0pt
\item[(a)] $H$ is an acyclic hypergraph.
\item[(b)] $H$ is a conformal and chordal hypergraph.
\item[(c)] $H$ has the running intersection property.
\item[(d)] $H$ has a join tree.
\item[(e)] $H$ has the \ltgc~for relations.
\end{enumerate}
\end{theorem}
As an illustration, if~$n\geq 2$, the hypergraph~$P_n$ is acyclic,
hence it has the \ltgc~for relations. In contrast, if~$n\geq 3$, the
hypergraphs~$C_n$ and~$H_n$ are cyclic, hence they do not have the
\ltgc~for relations.

We now generalize the notion of local-to-global consistency from relations to $\mathbb K$-relations.

\begin{definition} \label{defn:local-to-global}
Let $\mathbb K$~be a positive commutative monoid, and let~$X_1,\dots,X_m$ be a listing of all 
the hyperedges of a hypergraph~$H$. We say that~$H$ has the \emph{\ltgc~for~$\mathbb K$-relations} if every
collection~$R_1(X_1),\ldots,R_m(X_m)$ of~$\mathbb K$-relations that is
pairwise consistent is also globally consistent.
\end{definition}

In what follows, we will show that the implication
(e)~$\Rightarrow$ (a) in  Theorem \ref{thm:BFMY} holds
more generally for~$\mathbb K$-relations, where~$\mathbb K$ is an arbitrary positive commutative monoid. To
prove this result, we will need to find a more general construction than the
one devised in \cite{BeeriFaginMaierYannakakis1983} since the
construction given there uses some special properties of ordinary (set-theoretic)
relations that are not always shared by~$\mathbb K$-relations when $\mathbb K$ is an arbitrary positive commutative monoid. We are now ready to state the main result of this section.

\begin{theorem} \label{thm:necessary} Let~$\mathbb K$ be a positive commutative monoid
  and let~$H$ be a hypergraph.  If~$H$ has the
  \ltgc~for~$\mathbb K$-relations, then~$H$ is acyclic.
\end{theorem}

Before embarking on the proof of Theorem~\ref{thm:necessary}, we need
some additional notions about hypergraphs. The hypergraph~$H$ is called~\emph{$k$-uniform}
  if every hyperedge of~$H$ has exactly~$k$ vertices. It is called
  \emph{$d$-regular} if any vertex of~$H$ appears in exactly~$d$
  hyperedges of~$H$.  We show that hypergraphs that have such
  properties with $k \geq 1$ and $d \geq 2$ do not have the \ltgc~for
  any positive commutative monoid. After this is proved, we will show how
  to reduce the general case of an arbitrary acyclic hypergraph $H$ to
  the $k$-uniform and $d$-regular case. If a schema
  $X_1,\ldots,X_m$ is the set of hyperedges of a $k$-uniform or $d$-regular hypergraph,
  then we say that the schema $X_1,\ldots,X_m$ is $k$-uniform or $d$-regular, respectively.

\begin{lemma} \label{lem:newtseitin}
Let~$\mathbb{K}$ be a positive commutative monoid and
let~$X_1,\ldots,X_m$ be a schema that is~$k$-uniform and~$d$-regular
with~$k \geq 1$ and~$d \geq 2$. Then, there exists a collection
of~$\mathbb{K}$-relations over~$X_1,\ldots,X_m$ that is pairwise
consistent but not globally consistent.
\end{lemma}

\begin{proof}
    Let $c$ be an element of the universe $K$ of $\mathbb{K}=(K,+,0)$ 
    such that $c\not =0$ (recall that we have made the blanket assumption that the universes of 
    the positive commutative monoids considered have at least two elements). 
  Let~$a := c + \cdots + c$ with~$c$ appearing~$d^k$ times in the sum.
  Since $c\not =0$, the positivity of
  $\mathbb K$ implies that $a$~is a non-zero element of $K$; i.e..,~$a \not= 0$. 
  The $\mathbb{K}$-relations that we build will have all its attributes
  valued in the set $\{0,\ldots,d-1\}$. Therefore, if $Z$ is a set of
  attributes, then a $Z$-tuple $t$ is a map 
  \begin{equation}
      t : Z \to \{0,\ldots,d-1\}. \label{eqn:tuple}
  \end{equation}
  For each~$i \in [m]$ with~$i \not= m$, let~$R_i(X_i)$ be 
  defined by $R_i(t) = a$ for every $X_i$-tuple $t$ 
  whose
  total sum~$\sum_{C \in X_i} t(C)$ as integers is congruent to~$0$
  mod~$d$, and $R(t) = 0$ for every other $X_i$-tuple $t$. For~$i = m$, let~$R_m(X_m)$ be
  defined by $R_m(t) = a$ for every $X_i$-tuple $t$ whose
  total sum~$\sum_{C \in X_m} t(C)$ as integers is congruent to~$1$ mod~$d$,
  and $R_m(t) = 0$  for every other~$X_m$-tuple~$t$.

  To show that the collection~$R_1,\ldots,R_m$ of~$\mathbb K$-relations is pairwise
  consistent, fix any two indices~$i,j \in [m]$ and
  let~$a_i,a_j \in \{0,1\}$ be such that the supports of
  the~$\mathbb K$-relations~$R_i$ and~$R_j$ are, respectively, the set
  of~$X_i$-tuples~$t$ that satisfy the congruence
  equation~$\sum_{C \in X_i} t(C) \equiv a_i \text{ mod } d$, and the
  set of~$X_j$-tuples~$t$ that satisfy the congruence
  equation~$\sum_{C \in X_j} t(C) \equiv a_j \text{ mod } d$.
  Let~$X = X_i \cup X_j$ and~$Z = X_i \cap X_j$, and
  let~$b := c + \cdots + c$ with~$c$ appearing~$d^{|Z|+1}$ times in
  the sum.  Again,~$b$ is an element of~$\mathbb K$, and~$b \not= 0$ because~$\mathbb K$
  is a positive commutative monoid. Let~$T(X)$ be the~$\mathbb K$-relation defined 
  by $T(t) = b$ for every $X$-tuple $t$ that satisfies the system
  of two congruence
  equations
  \begin{align}
      & \textstyle{\sum_{C \in X_i} t(C) \equiv a_i \text{ mod } d}, \label{eqn:sys1} \\
      & \textstyle{\sum_{C \in X_j} t(C) \equiv a_j \text{ mod } d}, \label{eqn:sys2}
  \end{align} 
  and $T(t) = 0$
  for every other $X$-tuple $t$. We claim that~$T$ witnesses the consistency
  of~$R_i$ and~$R_j$.  Indeed,
  each~$X_i$-tuple~$u$ that
  satisfies the congruence
  equation~$\sum_{C \in X_i} u(C) \equiv a_i \text{ mod }d$ extends in
  exactly~$d^{k-|Z|-1}$ ways to an $X$-tuple $t$ that is a solution
  to the system of two
  congruence
  equations~\eqref{eqn:sys1}--\eqref{eqn:sys2}. Symmetrically,
  each~$X_j$-tuple~$v$ that
  satisfies the congruence
  equation~$\sum_{C \in X_j} v(C) \equiv a_j \text{ mod } d$ extends
  in exactly~$d^{k-|Z|-1}$ ways to an $X$-tuple $t$ that is a solution
  to the same system of
  two congruence equations. The consequence of this is that for
  each~$u \in R'_i$ and each~$v \in R'_j$ we
  have~$T[X_i](u) = T[X_j](v) = b + \cdots + b$ with~$b$ appearing~$d^{k-|Z|-1}$
  times in the sum. Recalling now that~$b = c + \cdots + c$ with~$c$
  appearing~$d^{|Z|+1}$ times in the sum we see
  that~$T[X_i](u) = T[X_j](v) = c + \cdots + c$ with~$c$
  appearing~$d^{k-|Z|-1} d^{|Z|+1} = d^k$ times in the sum, which
  equals~$a = R_i(u) = R_j(v)$.

  To argue that the relations~$R_1,\ldots,R_m$ are not globally
  consistent, we proceed by contradiction. If~$R$ were a~$\mathbb K$-relation
  that witnesses their consistency, then  its
  support would contain a tuple~$t$ such that the projections~$t[X_i]$
  belong to the supports~$R'_i$ of the~$R_i$, for each~$i \in [m]$. In
  turn this means that
  \begin{align}
    & \textstyle{\sum_{C \in X_i} t(C) \;\equiv\; 0 \text{ mod } d}, \;\;\;\;\;\
 \text{ for $i \not = m$ } \label{eqn:restluhere} \\
    & \textstyle{\sum_{C \in X_i} t(C) \;\equiv\; 1 \text{ mod } d}, \;\;\;\;\;\
 \text{ for $i = m$. } \label{eqn:a0luhere}
  \end{align}
  Since by~$d$-regularity each~$C \in V$ belongs to exactly~$d$ 
  sets~$X_i$, adding up all the equations in~\eqref{eqn:restluhere}
  and~\eqref{eqn:a0luhere} gives
\begin{equation}
    \textstyle{\sum_{C \in V} dt(C) \;\equiv\; 1 \text{ mod } d},
  \end{equation}
  which is absurd since the left-hand side is congruent to~$0$
  mod~$d$, the right-hand side is congruent to~$1$ mod~$d$,
  and~$d \geq 2$ by assumption. This completes the proof of
  Theorem~\ref{thm:necessary}.
\end{proof}

Building towards the proof of Theorem~\ref{thm:necessary}, 
in what follows we show how to reduce the general case of an 
arbitrary acyclic schema to a special case
of Lemma~\ref{lem:newtseitin}. We need some more terminology 
about hypergraphs, and two more lemmas.

Let~$H = (V,E)$ be a hypergraph. The \emph{reduction} of~$H$ is the
hypergraph~$R(H)$~whose set of vertices is~$V$ and whose hyperedges
are those hyperedges~$X \in E$ that are not included in any other
hyperedge of~$H$. A hypergraph~$H$ is \emph{reduced} if~$H=R(H)$.
If~$W \subseteq V$, then the \emph{hypergraph induced by~$W$ on~$H$}
is the hypergraph~$H[W]$~whose set of vertices is~$W$ and whose
hyperedges are the non-empty subsets of the form~$X \cap W$,
where~$X \in E$ is a hyperedge of~$H$; in
symbols,
\begin{equation*}
    H[W] = (W, \{X \cap W: X \in E\}\setminus \{\emptyset\}).
\end{equation*}
For a vertex~$u \in V$, we
write~$H\setminus u$ for the hypergraph induced by~$V\setminus\{u\}$
on~$H$. For an edge~$e \in E$, we write~$H\setminus e$ for the
hypergraph with~$V$ as the set of its vertices and
with~$E \setminus \{e\}$ as the set of its edges. 
We say that another hypergraph~$H'$ is obtained from~$H$ by a
\emph{vertex-deletion} if~$H' = H\setminus u$ for some~$u \in V$. We
say that~$H'$ is obtained from~$H$ by a \emph{covered-edge-deletion}
if~~$H' = H\setminus e$ for some~$e \in E$ such that~$e \subseteq f$
for some~$f \in E\setminus \{e\}$. In either case, we say that~$H'$ is
obtained from~$H$ by a \emph{safe-deletion operation}. We say that a
sequence of safe-deletion operations \emph{transforms~$H$ to~$H'$}
if~$H'$ can be obtained from~$H$ by starting with~$H$ and applying the
operations in order.

Note that if~$W$~is a subset of~$V$,
then the hypergraph~$R(H[W])$~is obtained from~$H$ by a sequence of
safe-deletion operations. Indeed, we can first obtain the
hypergraph~$H[W]$ from~$H$ by a sequence of vertex-deletions in which
the vertices of the set of~$V\setminus W$ are removed one by one;
after this, we can obtain the hypergraph~$R(H[W])$ from~$H[W]$ by a
sequence of covered-edge deletions.

\begin{lemma} \label{lem:characconf} For every hypergraph~$H=(V,E)$
  the following statements hold:
\begin{enumerate} \itemsep=0pt
\item $H$ is not chordal if and only if there
  exists~$W \subseteq V$~with $|W|\geq 4$ 
  and~$R(H[W]) \cong C_{|W|}$.
\item $H$ is not conformal if and only if there
  exists~$W \subseteq V$~with $|W|\geq 3$ 
  and~$R(H[W]) \cong H_{|W|}$.
\end{enumerate}
Moreover, there is a polynomial-time algorithm that, given a
hypergraph~$H$ that is not chordal or not conformal, finds both a
set~$W$ as stated in (1) or (2) and a sequence of safe-deletion
operations that transforms~$H$ to~$R(H[W])$.
\end{lemma}

\begin{proof}
  The proof of (1) is straightforward. For the proof of (2)
  see~\cite{DBLP:journals/csur/Brault-Baron16}. Since there exist
  polynomial-time algorithms that test whether a graph is chordal
  (see, e.g.,~\cite{RoseTarjanLueker1976}), an algorithm to find a~$W$
  as stated in (1), when~$H$ is not chordal, is to iteratively delete
  vertices whose removal leaves a hypergraph with a non-chordal primal
  graph until no more vertices can be removed. Also, since there exist
  polynomial-time algorithms that test whether a hypergraph is
  conformal (see, e.g., Gilmore's Theorem in page~31
  of~\cite{Berge1989book}), an algorithm to find a~$W$ stated in (2),
  when~$H$ is not conformal, is to iteratively delete vertices whose
  removal leaves a non-conformal hypergraph until no more vertices can
  be removed. In both cases, once the set~$W$ is found, a sequence of
  safe-deletion operations that transforms~$H$~to~$R(H[W])$ is
  obtained by first deleting all vertices in~$V\setminus W$, and then
  deleting all covered edges.
\end{proof}

\begin{lemma} \label{lem:cons-preservv} Let~$\mathbb K$ be a positive commutative monoid,
  and let~$H_0$ and~$H_1$ be hypergraphs such that~$H_0$ is obtained
  from~$H_1$ by a sequence of safe-deletion operations. For every
  collection~$D_0$ of~$\mathbb K$-relations over~$H_0$, there exists a
  collection~$D_1$ of~ $\mathbb K$-relations over~$H_1$ such that, for every~$k \geq 1$, it holds that~$D_0$ is~$k$-wise consistent if
  and only if~$D_1$ is~$k$-wise consistent.
\end{lemma}

\begin{proof}
  We define~$D_1$ when~$H_0$ is
  obtained from~$H_1$ by a single safe-deletion operation. The general case follows from iterating the construction. 
  In what follows, suppose
  that~$H_1 = (V_1,E_1)$, where~$V_1 = \{A_1,\ldots,A_n\}$
  and~$E_1 = \{X_1,\ldots,X_m\}$.

  Assume first that~$H_0 = H_1 \setminus X$ where~$X \in E_1$ is such
  that~$X \subseteq X_j$ for some~$j \in [m]$ with~$X \not= X_j$;
  i.e.,~$H_0$ is obtained from~$H_1$ by deleting a covered edge. In
  particular,~$V_0 = V_1$ and~$E_0 = E_1 \setminus \{X\}$. If
  the~$\mathbb K$-relations of~$D_0$ are~$S_i(X_i)$ for~$i \in [m]$
  with~$X_i \not= X$, then~$D_1$ is defined as the collection with
  $\mathbb K$-relations~$R_i(X_i)$ for~$i \in [m]$ defined as follows: For
  each~$i \in [m]$, if~$X_i \not= X$, then 
  let~$R_i := S_i$, else let~$R_i := S_j[X]$.

  Assume next that~$H_0 = H_1 \setminus A$ where~$A \in V_1$;
  i.e.,~$H_0$ is obtained from~$H_1$ by deleting a vertex. In
  particular,~$V_0 = V_1\setminus\{A\}$ and~$E_0 = \{Y_1,\ldots,Y_m\}$
  where~$Y_i = X_i\setminus\{A\}$ for~$i = 1,\ldots,m$. Fix a default
  value~$u_0$ in the domain~$\mathrm{Dom}(A)$ of the attribute~$A$. If
  the~$\mathbb K$-relations of~$D_0$ are~$S_i(Y_i)$ for~$i \in [m]$,
  then~$D_1$ is defined as the collection
  with~$\mathbb K$-relations~$R_i(X_i)$ for~$i \in [m]$ defined as follows:
  For each~$i \in [m]$, if~$A \not\in X_i$, then let~$R_i := S_i$; else
  let~$R_i$ be the~$\mathbb K$-relation of schema~$X_i = Y_i \cup \{A\}$
  defined for every~$X_i$-tuple~$t$ by~$R_i(t) := S_i(t[Y_i])$
  if~$t(A) = u_0$ and~$R_i(t) := 0$ if~$t(A) \not= u_0$. Here,~$0$
  denotes the neutral element of addition in~$\mathbb K$. We note
  that in case~$X_i = \{A\}$, the~$\mathbb K$-relation~$R_i$ has empty
  schema~$Y_i = \emptyset$ and consists of the empty tuple
  with~$\mathbb K$-value~$S_i(u_0)$.

  We prove the main property by cases. Fix an integer $k \geq 1$.

  \begin{claim} \label{claim:removeattribute}
    Assume~$H_0 = H_1 \setminus A$ for some vertex~$A \in V_1$. Then,
    the $\mathbb K$-relations~$S_i(Y_i)$ of~$D_0$ are~$k$-wise consistent if
    and only if the $\mathbb K$-relations~$R_i(X_i)$ of~$D_1$ are~$k$-wise
    consistent.
  \end{claim}

  \begin{proof}
    Fix~$I \subseteq [m]$ with~$|I| \leq k$,
    let~$X = \bigcup_{i \in I} X_i$ and~$Y = \bigcup_{i \in I}
    Y_i$. Observe that~$Y = X \setminus \{A\}$. In particular~$Y = X$
    if~$A$ is not in~$X$.

    (If): Let~$R$ be a $\mathbb K$-relation over~$X$ that witnesses the
    consistency of~$\{ R_i : i \in I \}$, and let~$S := R[Y]$. We
    claim that~$S$ witnesses the consistency of~$\{ S_i : i \in I \}$.
    Indeed, 
    \begin{equation*}
    S[Y_i] = R[Y][Y_i] = R[Y_i] = R_i[Y_i] = S_i,
    \end{equation*}
    where the
    first equality follows from the choice of~$S$, the second equality
    follows from~$Y_i \subseteq Y$, the third equality follows from
    the facts that~$R[X_i] = R_i$ and~$Y_i \subseteq X_i$, and the
    fourth equality follows from the definition of~$R_i$.

    (Only if): Consider the two cases:~$A \not\in X$ or~$A \in X$.
    If~$A \not\in X$, then~$R_i = S_i$ for every~$i \in I$ and
    there is nothing to prove.
    If~$A \in X$, then let~$S$ be a $\mathbb K$-relation over~$Y$ that
    witnesses the consistency of the
    $\mathbb K$-relations~$\{ S_i : i \in I\}$, and let~$R$ be the
    $\mathbb K$-relation over~$X$ defined for every~$X$-tuple~$t$
    by~$R(t) := 0$ if~$t(A) \not= u_0$ and by~$R(t) := S(t[Y])$
    if~$t(A) = u_0$. We claim that~$R$ witnesses the consistency of
    the $\mathbb K$-relations~$R_i$ for~$i \in I$. We show that~$R_i = R[X_i]$
    for~$i \in I$.  Towards this, first we argue
    that~$S[Y_i] = R[Y_i]$.  Indeed, for every~$Y_i$-tuple~$r$ we have
     \begin{align}
      S[Y_i](r) = \sum_{\newatop{s \in S':}{s[Y_i] = r}} S(s) =
      \sum_{\newatop{t \in \tuples(X):}{\newatop{t[Y_i] = r,}{t(A) = u_0}}} S(t[Y]) =
      \sum_{\newatop{t \in S':}{t[Y_i] = r}} R(t) = R[Y_i](r), \label{eqn:lslsggg}
    \end{align}
    where the first equality is the definition of marginal, the
    second equality follows from the fact that the map~$t \mapsto t[Y]$ is a
    bijection between the set of~$X$-tuples~$t$ such that~$t[Y_i]=r$
    and~$t(A) = u_0$ and the set of~$Y$-tuples~$s$ such
    that~$s[Y_i]=r$, the third equality follows from the definition of~$R$, and
    the fourth equality is the definition of marginal.

    In case~$A \not\in X_i$, we have that~$Y_i = X_i$, hence~\eqref{eqn:lslsggg} already shows
    that~$R_i = S_i = S[Y_i] = R[Y_i] = R[X_i]$. In case~$A \in X_i$,
    we use the fact that~$S_i = S[Y_i]$ to show that~$R_i =
    R[X_i]$. For every~$X_i$-tuple~$r$ with~$r(A) \not= u_0$, we
    have~$R_i(r) = 0$ and also~$R[X_i](r) = \sum_{t : t[X_i] = r} R(t) = 0$
    since the conditions that~$t[X_i] = r$ and~$A \in X_i$
    imply that~$t(A) = r(A) \not= u_0$. Thus, $R_i(r) = R[X_i](r) = 0$ 
    in this
    case. For every~$X_i$-tuple~$r$ with~$r(A) = u_0$, we have
    \begin{align}
      R_i(r) =  S_i(r[Y_i]) =  S[Y_i](r[Y_i]) =  R[Y_i](r[Y_i]), \label{lem:interr}
\end{align}
    where the first equality follows from the definition of~$R_i$ and
    the assumption that~$r(A) = u_0$, the second equality follows from
    $S_i = S[Y_i]$, and the third equality follows from~\eqref{eqn:lslsggg}.
    Continuing from the right-hand side of~\eqref{lem:interr}, we have
    \begin{align}
     R[Y_i](r[Y_i]) =  \sum_{\newatop{t \in R':}{t[Y_i] = r[Y_i]}} R(t) =
     \sum_{\newatop{t \in R':}{t[X_i] = r}} R(t) = R[X_i](r),
    \label{lem:retniy}
    \end{align}
    where the first equality is the definition of marginal, the
    second equality follows from the assumption that~$A \in X_i$
    and~$r(A) = u_0$ together with~$R(t) = 0$ in
    case~$t(A) \not= u_0$, and the third equality is the definition
    of marginal. Combining~\eqref{lem:interr}
    with~\eqref{lem:retniy}, we get~$R_i(r) = R[X_i](r)$ also in this
    case. This proves that~$R_i = R[X_i]$.
    \end{proof}

    \begin{claim} \label{claim:removecoverededge}
      Assume~$H_0 = H_1 \setminus X$ for some edge~$X \in E_1$ that is
      covered in~$H_1$. Then, the $\mathbb K$-relations~$S_i(X_i)$ of~$D_0$
      are~$k$-wise consistent if and only if the $\mathbb K$-relations~$R_i(Y_i)$
      of~$D_1$ are~$k$-wise consistent.
    \end{claim}

     \begin{proof}
       Let~$l \in [m]$ be such that~$X = X_l \subseteq X_j$
       for some~$j \in [m]\setminus\{l\}$,
       so~$E_0 = \{ X_i : i \in [m]\setminus \{l\}\}$.

       (If): Fix~$I \subseteq [m]\setminus\{l\}$ with~$|I| \leq k$ and
       let~$X = \bigcup_{i \in I} X_i$. Let~$R$ be a $\mathbb K$-relation
       over~$X$ that witnesses the consistency of~$\{ R_i : i \in I\}$
       and let~$S = R$. Since~$S_i = R_i$ for
       every~$i \in [m]\setminus\{l\}$, it is obvious that~$S$
       witnesses the consistency of~$\{ S_i : i \in I \}$.

       (Only if): Fix~$I \subseteq [m]$ with~$|I| \leq k$ and
       let~$X = \bigcup_{i \in I} X_i$. Let~$S$ be a $\mathbb K$-relation
       over~$X$ that witnesses the consistency
       of~$\{ S_i : i \in I\setminus\{l\} \}$ and let~$R = S$. We
       have~$R_l = S_j[X_l] = S[X_j][X_l] = R[X_j][X_l] = R[X_l]$
       where the first equality follows from the definition of~$R_l$,
       the second equality follows from the fact that~$S_j = S[X_j]$,
       the third equality follows from the choice of~$R$, and the
       fourth equality follows from~$X_l \subseteq X_j$.
     \end{proof}

     The proof of  Lemma \ref{lem:cons-preservv} is now complete.
\end{proof}

Lemma \ref{lem:cons-preservv} implies that the local-to-global
consistency property for~$\mathbb K$-relations is preserved under induced
hypergraphs and under reductions.

\begin{restatable}{corollary}{preservelemma} \label{lem:preserv1here}
  Let~$\mathbb K$ be a positive commutative monoid and let~$H$ be a hypergraph. If~$H$
  has the \ltgc~for~$\mathbb K$-relations, then, for every subset~$W$ of the
  set of vertices of~$H$, the hypergraph~$R(H[W])$ also has the
  \ltgc~for~$\mathbb K$-relations.
\end{restatable}
\begin{proof}
  As discussed earlier, the hypergraph~$R(H[W])$ is obtained from the
  hypergraph~$H$ by a sequence of safe-deletion operations. We will
  apply Lemma~\ref{lem:cons-preservv} with~$H_0= R(H[W])$
  and~$H_1=H$. Let~$m$ be the number of hyperedges of~$R(H[W])$ and
  let~$m'$ be the number of hyperedges of~$H$; clearly, we have
  that~$m\leq m'$. Let~$R_1,\ldots,R_{m}$ be a collection
  of~$\mathbb K$-relations over~$R(H[W])$ that are pairwise consistent. We
  have to show that this collection is globally consistent. By
  Lemma~\ref{lem:cons-preservv}, there is a collection
  of~$\mathbb K$-relations~$S_1,\dots,S_{m'}$ over~$H$ that are pairwise
  consistent. Since~$H$ has the \ltgc~for~$\mathbb K$-relations, it follows
  that the collection~$S_1,\dots,S_{m'}$ is globally consistent, i.e.,
  it is~$m'$-wise consistent. Since~$m\leq m'$, we have that the
  collection~$S_1,\dots,S_{m'}$ is also~$m$-wise
  consistent. By Lemma~\ref{lem:cons-preservv} (but in the reverse
  direction this time), we have that the collection~$R_1,\dots,R_m$ is~$m$-wise consistent, which means
  that it is globally consistent, as it was to be shown.
\end{proof}

We are now ready to give the proof of Theorem~\ref{thm:necessary}.

\begin{proof}[Proof of Theorem~\ref{thm:necessary}] 
  Assume that the hypergraph~$H$ is not acyclic, so in particular~$H$
  is not both chordal and conformal. By Lemma~\ref{lem:characconf},
  there is a subset~$W$ of~$V$ such that~$|W| \geq 3$
  and~$R(H[W])=C_{|W|}$ or there is a
  subset~$W$ of~$V$ such that~$|W| \geq 4$
  and~$R(H[W]) = H_{|W|}$.   
  Now note that for $n \geq 3$ the (hyper)graph $C_n$ is 
  $k$-uniform and $d$-regular for $k = 2 \geq 1$ and $d = 2$, and
  for $n \geq 4$ the hypergraph $H_n$ is $k$-uniform and $d$-regular for
  $k = n-1 \geq 1$ and $d = n-1 \geq 2$. Therefore, Lemma~\ref{lem:newtseitin}
  applies to conclude that $R(H[W])$ does not have the \ltgc~for $\mathbb{K}$-relations, 
  and Corollary~\ref{lem:preserv1here} implies that $H$ does not have it either.
\end{proof}

\subsection{Acyclicity is Not Always Sufficient} \label{sec:not-sufficient}

In this section, we show that there are positive commutative monoids $\mathbb K$ and acyclic schemas $H$ such that $H$ does \emph{not} have the \ltgc~for $\mathbb K$-relations. In other words, the acyclicity of a schema is not a sufficient condition for  the \ltgc~to hold for arbitrary positive commutative monoids.

Let $\mathbb{N}_2=(\{0,1,2\}, \oplus, 0)$ be the structure with the set $\{0,1,2\}$ as its universe and  addition  rounded to $2$ as its operation, i.e., $1\oplus 1 = 2\oplus 1 = 2\oplus 2 = 2$, and $0 \oplus x = x \oplus 0 = x$ for all $x \in \{0,1,2\}$. It is easy to verify that $\mathbb {N}_2$ is a positive commutative monoid.

Let $P_3$ be the \emph{path-of-length-3} hypergraph whose vertices form the set
 $\{A,B,C\}$ and whose edges form the set 
$\{\{A,B\},\{B,C\},\{C,D\}\}$. 
Clearly, $P_3$ is an acyclic hypergraph. 
 
\begin{proposition} \label{prop:path3} 
 The path-of-length-3 hypergraph $P_3$ does not 
 have the \ltgc~for~$\mathbb {N}_2$-relations.
\end{proposition}
\begin{proof}
Consider the following three $\mathbb {N}_2$-relations $R_1(AB),R_2(BC),R_3(CD)$:
  
  \begin{center}
  \begin{tabular}{llllllllllllllll}
  $A$ & $B$ & : & $R_1$ & & $B$ & $C$ & : & $R_2$ & & $C$ & $D$ & : & $R_3$ \\
  \cmidrule(lr){1-4}\cmidrule(lr){6-9}\cmidrule(lr){11-14}
  $a_1$ & $b_1$ & : & $1$ & & $b_1$ & $c_1$ & : & $2$ & & $c_1$ & $d_1$ & : & $1$ \\
  $a_2$ & $b_1$ & : & $1$ & & $b_2$ & $c_2$ & : & $2$ & & $c_1$ & $d_2$ & : & $1$ \\ 
  $a_3$ & $b_2$ & : & $2$ & &       &       &   &     & & $c_1$ & $d_3$ & : & $1$ \\
        &       &   &     & &       &       &   &     & & $c_2$ & $d_4$ & : & $2$
  \end{tabular}
  \end{center}
  
\noindent The $\mathbb{N}_2$-relations $R_{12}(ABC), R_{23}(BCD), R_{13}(ABCD)$ 
that follow witness the pairwise consistency of the $\mathbb{N}_2$-relations $R_1(AB),R_2(BC),R_3(CD)$.

  \begin{center}
  \begin{tabular}{llllllllllllllllll}
  $A$ & $B$ & $C$ & : & $R_{12}$ & & $B$ & $C$ & $D$ & : & $R_{23}$ & & $A$ & $B$ & $C$ & $D$ & : & $R_{13}$ \\
  \cmidrule(lr){1-5}\cmidrule(lr){7-11}\cmidrule(lr){13-18}
  $a_1$ & $b_1$ & $c_1$ & : & $1$ & & $b_1$ & $c_1$ & $d_1$ & : & $1$ & & $a_1$ & $b_1$ & $c_1$ & $d_1$ & : & $1$ \\
  $a_2$ & $b_1$ & $c_1$ & : & $1$ & & $b_1$ & $c_1$ & $d_2$ & : & $1$ & & $a_2$ & $b_1$ & $c_1$ & $d_2$ & : & $1$ \\
  $a_3$ & $b_2$ & $c_2$ & : & $2$ & & $b_1$ & $c_1$ & $d_3$ & : & $1$ & & $a_3$ & $b_2$ & $c_1$ & $d_3$ & : & $1$ \\
        &       &       &   &     & & $b_2$ & $c_2$ & $d_4$ & : & $2$ & & $a_3$ & $b_2$ & $c_2$ & $d_4$ & : & $2$
  \end{tabular}
  \end{center}

We now show that the relations $R_1,R_2,R_3$ are not globally consistent.
Towards a contradiction, assume that there is a $\mathbb{N}_2$-relation $W(ABCD)$ witnessing their global consistency.  For each $i=1,2,3$,  the support of $R_i$ must be equal to the support of  the projection of $W$ on the attributes of $R_i$; thus, $W(ABCD)$ must be of the form:

\begin{center}
\begin{tabular}{lllllllll}
  $A$ & $B$ & $C$ & $D$ & : & $W$ \\
  \cmidrule(lr){1-6}
  $a_1$ & $b_1$ & $c_1$ & $d_1$ & : & $x_1$ \\
  $a_1$ & $b_1$ & $c_1$ & $d_2$ & : & $x_2$  \\
  $a_1$ & $b_1$ & $c_1$ & $d_3$ & : & $x_3$  \\
  $a_2$ & $b_1$ & $c_1$ & $d_1$ & : & $x_4$  \\
  $a_2$ & $b_1$ & $c_1$ & $d_2$ & : & $x_5$  \\
  $a_2$ & $b_1$ & $c_1$ & $d_3$ & : & $x_6$  \\
  $a_3$ & $b_2$ & $c_2$ & $d_4$ & : & $x_7$.
\end{tabular}
\end{center}

For example, the support of $W(ABCD)$ cannot contain the tuple $(a_3,b_2,c_1,d_3)$ because the pair $(b_2,c_1)$ does not belong to the support of $R_2(BC)$.
Since $W$ witnesses the global consistency of $R_1, R_2, R_3$ and since
$R_1(a_1,b_1)  =  R_1(a_2,b_1)= 1$,  we must have
that
\begin{eqnarray}
 x_1 \oplus x_2 \oplus x_3 & = & 1 \label{eq:1}\\
 x_4 \oplus x_5 \oplus x_6 & = & 1. \label{eq:2}
\end{eqnarray}
Similarly and since $R_3(c_1,d_1) = R_3(c_1,d_2) = R_3(c_1,d_3) = 1$, we must have that 
\begin{eqnarray}
 x_1 \oplus x_4 & = &  1 \label{eq:3}\\
  x_2 \oplus x_5 &  = &  1 \label{eq:4}\\
  x_3 \oplus x_6 & = & 1. \label{eq:5}
\end{eqnarray}
By Equation (\ref{eq:3}),  we must have either $x_1 = 1$ and $x_4=0$,  or $x_1=0$ and $x_4 = 1$. 
If $x_1 = 1$ and $x_4=0$, then,  by Equations (\ref{eq:1}) and (\ref{eq:2}), we have that $x_2 = x_3 = 0$ and $x_5 \oplus x_6 = 1$. But then, by Equations (\ref{eq:4}) and (\ref{eq:5}), we have that $x_5 = 1 = x_6$, hence $x_5 \oplus x_6=2$, a contradiction. 
If $x_1 = 0$ and $x_4=1$,
then,  by Equations (\ref{eq:1}) and (\ref{eq:2}), we have that $x_2 \oplus x_3 = 1$ and $x_5 =x_6 = 0$. But then, by Equations (\ref{eq:3}) and (\ref{eq:4}), we have that $x_2 = 1 = x_3$, hence $x_2 \oplus x_3=2$, a contradiction. Therefore, the $\mathbb {N}_2$-relations $R_1, R_2, R_3$ are not globally consistent. 
\end{proof}

\section{Acyclicity and the Transportation Property}  \label{sec:sufficient}

As seen in the previous section, there exist positive commutative monoids $\mathbb{K}$ for which 
 acyclicity of a hypergraph is not a sufficient condition for the hypergraph to have the
\ltgc~for~$\mathbb {K}$-relations. In this
section we ask: under what conditions on the monoid is acyclicity sufficient?  We
introduce a property of commutative monoids, which we call the \emph{transportation property}, and show that it characterizes 
the positive commutative monoids $\mathbb{K}$ for which acyclicity of a hypergraph $H$ is sufficient 
for $H$ to have the \ltgc~for $\mathbb{K}$-relations.
Then, in the next section, we  show that many positive commutative monoids of interest have the 
transportation property. 

\subsection{Transportation Property and Inner Consistency Property}

Let~$\mathbb {K}$ be a positive commutative monoid. Recall that if~$R(X)$ and~$S(Y)$
are~$\mathbb K$-relations, then, by definition,~$R(X)$ and~$S(Y)$ are consistent
if there is a~$\mathbb K$-relation~$T(XY)$ such that~$T[X] = R$
and~$T[Y] = S$. It is not difficult to see that if~$R(X)$ and~$S(Y)$
are consistent, then~$R[X \cap Y] = S[X \cap Y]$, i.e., 
$R(X)$ and $S(Y)$ have the same marginals on the set of their common attributes. 
Motivated by this, we introduce the following two notions.

\begin{definition} \label{defn:inner} 
Let $\mathbb K$ be a positive commutative monoid.
Two~$\mathbb {K}$-relations~$R(X)$ and~$S(Y)$ are \emph{inner
  consistent} if~$R[X \cap Y] = S[X \cap Y]$ holds. 
The \emph{inner consistency
  property holds for $\mathbb K$-relations}
  if whenever two $\mathbb {K}$-relations $R(X)$ and $S(Y)$ are
inner consistent, then $R(X)$ and $S(Y)$ are also consistent.
\end{definition}

The main result of this section asserts that  the inner consistency property holds for $\mathbb K$-relations if and only if every  acyclic hypergraph has the \ltgc~for $\mathbb K$-relations. Rather unexpectedly, it turns out  that this last property is equivalent to just 
 the path-of-length three hypergraph $P_3$  having the \ltgc~for $\mathbb K$-relations.
To prove this result, we will  introduce a combinatorial property of monoids whose definition
involves only elements from the universe of  the monoid, i.e., no relations are involved in the definition of this combinatorial property.

\begin{definition}\label{defn:transportationproblem}
 Let ${\mathbb K}= (K, +,0)$ be a positive commutative monoid.
 The \emph{transportation problem for $\mathbb K$} is the following decision problem: given 
  two positive integers~$m$ and~$n$, a
 column~$m$-vector~$b = (b_1,\ldots,b_m) \in K^m$ with entries in~$K$,  and a
  row~$n$-vector~$c = (c_1,\ldots,c_n) \in K^n$ with entries
 in~$K$, does  there
exist an~$m \times n$
matrix~$D = (d_{ij} : i \in [m], j \in [n]) \in K^{m \times n}$ with
 entries in~$K$ such that~$d_{i1} + \cdots + d_{im} = b_i$ for all~$i \in [m]$
 and~$d_{1j} + \cdots + d_{mj} = c_j$ for all~$j \in [n]$?
 In words, this means that the rows of $D$ sum to $b$ and the columns of $D$ sum to $c$. 
 \end{definition}

An instance $b = (b_1,\ldots,b_m)$ and $c = (c_1,\ldots,c_n)$  of the transportation problem can be viewed as a system of linear equations having $mn$ variables and $m+n$ equations.
 Graphically, we represent the first $m$ equations
horizontally and the next $n$ equations vertically, in accordance with the convention
that $b$ is a column vector and $c$ is a row vector:
\begin{equation}
\begin{array}{ccccccccc}
x_{11} & + & x_{12} & + & \cdots & + & x_{1n} & =      & b_1  \\
+      &   & +      &   &        &   & +      &        &      \\
x_{21} & + & x_{22} & + & \cdots & + & x_{2n} & =      & b_2  \\
+      &   & +      &   &        &   & +      &        & \\
\vdots &   & \vdots &   & \ddots &   & \vdots &        &  \\
+      &   & +      &   &        &   & +      &        & \\
x_{m1} & + & x_{m2} & + & \cdots & + & x_{mn} & =      & b_m  \\
\shortparallel  &   & \shortparallel     &   &        &   & \shortparallel     &        & \\
c_1    &   & c_2    &   &        &   & c_n    &        &
\end{array}
\label{eqn:systemofequations}
\end{equation}

The term ``transportation problem" comes from linear programming, where this problem has the following interpretation. Suppose a product is manufactured in $m$ different factories, where factory $i$ produces $b_i$ units of the product, $i \in [m]$. The units produced have to be transported to $n$ different markets, where   the demand of the product at market $j$ is $c_j$ units, $j\in [n]$. 
The question is whether there is a way to ship every unit produced 
at each factory,  so that the demand at each market is met; thus, the variable $x_{ij}$ represents the number of units produced in factory $i$ that are shipped to market  $j$, where $i \in  [m]$ and $ j\in [n]$.

Suppose that an instance of the transportation problem  has a solution~$(d_{ij} : i \in [m], j \in [n])$ in $\mathbb K$. By summing over all rows of the system  (\ref{eqn:systemofequations}), we have that $\sum_{i=1}^m\sum_{j=1}^n d_{ij}= b_1+ \cdots +b_m$. Similarly, by summing over all columns of the system (\ref{eqn:systemofequations}), we have that $\sum_{j=1}^n\sum_{i=1}^md_{ij}= c_1+ \cdots +c_n$. The commutativity of $\mathbb K$ implies that $\sum_{i=1}^m\sum_{j=1}^n d_{ij}= \sum_{j=1}^n\sum_{i=1}^md_{ij}$, hence $b_1+ \cdots +b_m = c_1+ \cdots +c_n$.
Thus, a necessary condition for an instance of the transportation problem to have a solution is that this instance is \emph{balanced}, i.e., $b_1 + \cdots + b_n = c_1 + \cdots + c_m$. In words, if an instance of the transportation problem has a solution, then the total supply must be equal to the total demand. 

We are now ready to introduce the notion of the transportation property.

\begin{definition}\label{defn:ftp}
 Let ${\mathbb K}= (K, +,0)$ be a positive commutative monoid.
 We say that~$\mathbb K$ has the \emph{\ftp}~if
 for
 every two positive integers~$m$ and~$n$, every
 column~$m$-vector~$b = (b_1,\ldots,b_m) \in K^m$ with entries in~$K$ and
 every row~$n$-vector~$c = (c_1,\ldots,c_n) \in K^n$ with entries
 in~$K$ such that~$b_1 + \cdots + b_m = c_1 + \cdots + c_n$ holds, we have that there
exists an~$m \times n$
matrix~$D = (d_{ij} : i \in [m], j \in [n]) \in K^{m \times n}$ with
entries in~$K$ whose 
rows sum to $b$ and whose columns sum to~$c$,
 i.e.,~$d_{i1} + \cdots + d_{im} = b_i$ for all~$i \in [m]$
 and~$d_{1j} + \cdots + d_{mj} = c_j$ for all~$j \in [n]$.

 In words, $\mathbb K$ has the \emph{\ftp}~if every balanced instance of the transportation problem has a solution in $\mathbb K$.
\end{definition}

The following three examples will turn out to be special cases of more general results that will be established in Section~\ref{sec:examples}, where
many additional examples of positive commutative monoids that have the \ftp~will be provided.

\begin{example}
    The monoid $\mathbb{B} = (\{0,1\},\vee,0)$ of Boolean truth-values with disjunction has the \ftp.
    To see this, consider a system of equations as in~\eqref{eqn:systemofequations} 
    where $b_1 + \cdots + b_m = c_1 + \cdots + c_n$; moreover, here we have that  each $b_i$ or $c_j$ is a truth-value, 
    and $+$ is $\vee$. This means  that either every $b_i$ and every $c_j$ is equal to $0$,
    or at least one $b_i$ is equal to $1$ and at least one $c_j$ is equal to  $1$. To find a solution, set $x_{ij} = b_i \wedge c_j$ for all $i \in [m]$ 
    and $j \in [n]$, where $\wedge$ is the standard Boolean conjunction.
    It is easy to see that this candidate solution satisfies all equations.
\end{example}

\begin{example} \label{ex:nonnegativereals}
    The monoid $\mathbb{R}^{\geq 0} = (R^{\geq 0},+,0)$ of non-negative reals with addition
     has the \ftp. To see this, consider a system of equations as in~\eqref{eqn:systemofequations} 
    and consider the matrices defined by
    $d_{ij} = b_i c_j / \sum_{k=1}^n c_k$ and $e_{ij} = b_i c_j / \sum_{k=1}^m b_k$ for all $i \in [m]$
    and $j \in [n]$, with the convention that $0/0 = 0$. It is straightforward to see
    that the $d_{ij}$ matrix satisfies all horizontal equations and the $e_{ij}$ matrix satisfies all
    vertical equations. Furthermore, if the instance is balanced so that $b_1 + \cdots + b_m = c_1 + \cdots + c_n$ holds,
    then $d_{ij} = e_{ij}$ and then both matrices are equal and satisfy all equations.
\end{example}

\begin{example} \label{ex:bagmonoid}
  The monoid $\mathbb{N} = (Z^{\geq 0},+,0)$ of non-negative integers with addition
  has the \ftp. This will follow from results established in subsequent sections. For now,
  an appealing but indirect way to see this is to notice that if we write the
  system of equations~\eqref{eqn:systemofequations} in the form $Ax = b$, where $A$ is an $(m+n) \times mn$ matrix
  with $0$-$1$ entries and $b$ is an $(m+n)$-vector with non-negative integer entries, 
  then $A$ is the incidence matrix of a bipartite graph and hence 
  a \emph{totally unimodular} matrix (see Example~1 in page~273 of Schrijver's book \cite{Schrijver-book}).
  The main result about totally unimodular matrices implies that if the linear program given 
  by $Ax = b$ and $x \geq 0$ 
  has a solution over $\mathbb{R}$, then it has a  solution with integer entries (see Corollary~19.2a 
  in~\cite{Schrijver-book} and the discussion immediately following its proof).  
  Since the \ftp~holds for $\mathbb{R}^{\geq 0}$,
  the conclusion of this is that the \ftp~for $\mathbb{N}$ follows from the \ftp~for $\mathbb{R}^{\geq 0}$
  from Example~\ref{ex:nonnegativereals}.
\end{example}

\subsection{Transportation Property and Acyclicity}

With all definitions in place, we are ready to state and prove the main result of this section.

\begin{theorem} \label{thm:allequivalent}
Let $\mathbb{K}$ be a positive commutative monoid. Then, the following statements are equivalent:
\begin{enumerate} \itemsep=0pt
\item[(1)] $\mathbb{K}$ has the \ftp.
\item[(2)] The inner consistency property holds for $\mathbb{K}$-relations.
\item[(3)] Every acyclic hypergraph has
the  \ltgc~for $\mathbb K$-relations. 
\item[(4)] The hypergraph $P_3$ has 
 the \ltgc~for $\mathbb K$-relations.
\end{enumerate}
\end{theorem}

\begin{proof} We close a cycle of implications: (1) $\Longrightarrow$ (2) $\Longrightarrow$ (3)
$\Longrightarrow$ (4) $\Longrightarrow$ (1).

(1) $\Longrightarrow$ (2).
 Suppose that $\mathbb{K}$ has the \ftp.
   Let~$R(X)$ and~$S(Y)$ be two  inner consistent $\mathbb K$-relations and
   let~$Z = X \cap Y$. For each~$Z$-tuple~$w$ in the support   of~$R[Z] = S[Z]$, let~$u_1,\ldots,u_{m_w}$ be an enumeration of
   the~$X$-tuples that are in the support~$R'$ of~$R$ and extend~$w$,
  and let~$v_1,\ldots,v_{n_w}$ be an enumeration of the~$Y$-tuples that
  are in the support~$S'$ of~$S$ and extend~$w$.
  Let~$b_w = (b_{w,1},\ldots,b_{w,{m_w}})$ be the column vector defined
  by~$b_{w,j} := R(u_j)$ for~$j \in [m_w]$, and
  let~$c_w = (c_{w,1},\ldots,c_{w,{n_w}})$ be the row vector defined
  by~$c_{w,i} := S(v_i)$ for~$i \in [n_w]$. Since $R$ and $S$ are inner consistent, we have that ~$R(w) = S(w)$, hence
  \begin{equation}
  b_{w,1} + \cdots + b_{w,m_w} = c_{w,1} + \cdots + c_{w,n_w}.
  \end{equation}
  By
  the \ftp~of~$\mathbb K$, there exists an~$m_w \times n_w$
   matrix~$M_w = (d_w(i,j) : i \in [m_w], j \in [n_w])$ that has~$b_w$ as column sum and~$c_w$ as row sum.
  Let~$T(XY)$ be the~$\mathbb K$-relation defined for every~$XY$-tuple~$t$
  by~$T(t) := d_w(i,j)$ where~$w = t[Z]$ and~$i$ and~$j$ are such
  that~$t[X] = u_i$ and~$t[Y] = v_j$ in the enumerations of the tuples
  in~$R'$ and~$S'$ that are used in  defining~$b_w$ and~$c_w$.
  For any other~$XY$-tuple~$t$, set~$T(t) := 0$.  It follows from the   definitions that~$T$ is a~$\mathbb K$-relation that witnesses the
  consistency of~$R$ and~$S$.

(2) $\Longrightarrow$ (3).
  Assume that the hypergraph~$H$ is acyclic and therefore it has the
  running intersection property. Hence, there is a
  listing~$X_1,\ldots,X_m$ of its hyperedges such that for
  every~$i \in [m]$ with~$i \geq 2$, there is a~$j \in [i-1]$ such
  that~$X_i \cap (X_1 \cup \cdots \cup X_{i-1}) \subseteq
  X_j$. Let~$R_1(X_1),\ldots,R_m(X_m)$ be a collection
  of~$\mathbb K$-relations that is pairwise consistent.  By induction
  on~$i = 1,\ldots,m$, we show that there is a~$\mathbb K$-relation~$T_i$
  over~$X_1 \cup \cdots \cup X_i$ that witnesses the global
  consistency of the~$\mathbb K$-relations~$R_1,\ldots,R_i$.  For~$i = 1$ the
  claim is obvious by taking~$T_1 = R_1$. Assume then that~$i \geq 2$
  and that the claim is true for all smaller indices.
  Let~$X := X_1 \cup \cdots \cup X_{i-1}$. By the running intersection
  property, let~$j \in [i-1]$ be such that~$X_i \cap X \subseteq X_j$.
  By induction hypothesis, there is a~$\mathbb K$-relation~$T_{i-1}(X)$ that
  witnesses the global consistency of~$R_1,\ldots,R_{i-1}$.  First, we
  show that~$T_{i-1}$ and~$R_i$ are consistent. Since, by assumption,
  the inner consistency property for~$\mathbb K$-relations holds, it suffices
  to show that~$T_{i-1}$ and~$R_i$ are inner consistent, i.e.,
  that~$T_{i-1}[X \cap X_i] = R_i[X \cap X_i]$.  Let~$Z = X \cap X_i$,
  so~$Z \subseteq X_j$ by the choice of~$j$, and
  indeed~$Z = X_j \cap X_i$.  Since~$j \leq i-1$, we
  have~$R_j = T_{i-1}[X_j]$. Since~$Z\subseteq X_j$, we
  have
  \begin{equation}
      R_j[Z] = T_{i-1}[X_j][Z] = T_{i-1}[Z]. \label{eqn:tr1}
  \end{equation}    
  By assumption,
  also~$R_j$ and~$R_i$ are consistent, and if~$W$ is any~$\mathbb K$-relation
  that witnesses their consistency and~$Z = X_j \cap X_i$,
  then
  \begin{equation}
      R_j[Z] = W[X_j][Z] = W[Z] = W[X_i][Z] = R_i[Z]. \label{eqn:tr2}
  \end{equation}
  By transitivity,~\eqref{eqn:tr1} and~\eqref{eqn:tr2} give~$T_{i-1}[Z] = R_i[Z]$, as was to be proved to
  show that~$T_{i-1}$ and~$R_i$ are consistent. Now, let~$T_i$ be
  a~$\mathbb K$-relation that witnesses the consistency of~$T_{i-1}$
  and~$R_i$.  We show that~$T_i$ witnesses the global consistency
  of~$R_1,\ldots,R_i$. Since~$T_{i-1}$ and~$R_i$ are consistent
  and~$T_i$ is a witness, we have~$T_{i-1} = T_i[X]$
  and~$R_i = T_i[X_i]$. Now fix~$k \leq i-1$ and note that
  \begin{equation*}
      R_k = T_{i-1}[X_k] = T_i[X][X_k] = T_i[X_k],
  \end{equation*}
  where the first
  equality follows from the fact that~$T_{i-1}$ witnesses the
  consistency of~$R_1,\ldots,R_{i-1}$ and~$k \leq i-1$, and the other
  two equalities follow from~$T_{i-1} = T_i[X]$ and the fact
  that~$X_k\subseteq X$.  Thus,~$T_i$ witnesses the consistency
  of~$R_1,\ldots,R_i$, which was to be shown.   

(3) $\Longrightarrow$ (4).
This statement is obvious.

(4) $\Longrightarrow$ (1).
Assume that the path-of-length-$3$ hypergraph $P_3$ has the \ltgc~for $\mathbb{K}$-relations. 
Let $(b_1,\ldots,b_m)$ and $(c_1,\ldots,c_n)$ be the two vectors of a balanced instance of the transportation problem for $\mathbb{K}$. Consider
the associated system of equations as in~\eqref{eqn:systemofequations}. Let $a = b_1 + \cdots + b_m = c_1 + \cdots + c_n$. If $a = 0$, then $b_1 = \cdots = b_m = c_1 = \cdots = c_n = 0$ by the positivity of $\mathbb{K}$, and then setting $x_{ij} = 0$ for all $i$ and $j$ we get a solution
to~\eqref{eqn:systemofequations}. Assume then that $a \not= 0$. Based on this instance, we first build three $\mathbb{K}$-relations $R(AB),S(BC),T(CD)$, then we show that they are pairwise consistent, and finally we show how to use any witness of their global consistency to build a solution to the given balanced instance of the transportation problem.
The three $\mathbb{K}$-relations are given by the following tables, where the third column is the annotation value from $\mathbb{K}$ for the tuple on its left:
  \begin{center}
  \begin{tabular}{cccccccccccccccccccccccc}
  $A$ & $B$ & : & $R$ & & $B$ & $C$ & : & $S$ & & $C$ & $D$ & : & $T$ \\
  \cmidrule(lr){1-4}\cmidrule(lr){6-9}\cmidrule(lr){11-14}
  $u_1$ & $0$ & : & $b_1$ & \hspace{1cm} & $0$ & $0$ & : & $a$ & \hspace{1cm} & $1$ & $u_1$ & : & $b_1$ \\
  $\vdots$ & $\vdots$ & & $\vdots$ & & $1$ & $1$ & : & $a$ & & $\vdots$ & $\vdots$ & & $\vdots$ \\
  $u_m$ & $0$ & : & $b_m$ & & & & & & & $1$ & $u_m$ & : & $b_m$ \\
  $v_1$ & $1$ & : & $c_1$ & & & & & & & $0$ & $v_1$ & : & $c_1$ \\
  $\vdots$ & $\vdots$ & & $\vdots$ & & & & & & & $\vdots$ & $\vdots$ & & $\vdots$ \\
  $v_n$ & $1$ & : & $c_m$ & & & & & & & $0$ & $v_n$ & : & $c_n$ \\
  \end{tabular}
  \end{center}
As witnesses to the pairwise consistency of these three $\mathbb K$-relations, consider the following $\mathbb{K}$-relations:
  \begin{center}
  \begin{tabular}{cccccccccccccccccccccccccc}
  $A$ & $B$ & $C$ & : & $U$ & \hspace{1cm} & $B$ & $C$ & $D$ & : & $V$ & \hspace{1cm} & $A$ & $B$ & $C$ & $D$ & : & $W$ \\
  \cmidrule(lr){1-5}\cmidrule(lr){7-11}\cmidrule(lr){13-18}
  $u_1$ & $0$ & $0$   & : & $b_1$ & & 
  $1$   & $1$     & $u_1$ & : & $b_1$ & &
  $u_1$ & $0$ & $1$   & $u_1$ & : & $b_1$ \\
  $\vdots$ & $\vdots$ & $\vdots$ & & $\vdots$ & & $\vdots$ & $\vdots$ & $\vdots$ & & $\vdots$ & & $\vdots$ & $\vdots$ & $\vdots$ & $\vdots$ & & $\vdots$ \\
  $u_m$ & $0$ & $0$   & : & $b_m$ & & 
  $1$   & $1$ & $u_m$ & : & $b_m$ & &
  $u_m$ & $0$ & $1$   & $u_m$ & : & $b_m$ \\
  $v_1$ & $1$ & $1$   & : & $c_1$ & & 
  $0$   & $0$ & $v_1$ & : & $c_1$ & &
  $v_1$ & $1$ & $0$   & $v_1$ & : & $c_1$ \\
  $\vdots$ & $\vdots$ & $\vdots$ & & $\vdots$ & & $\vdots$ & $\vdots$ & $\vdots$ & & $\vdots$ & & $\vdots$ & $\vdots$ & $\vdots$ & $\vdots$ & & $\vdots$ \\
  $v_m$ & $1$ & $1$ & : & $c_m$ & & 
  $0$ & $0$ & $v_m$ & : & $c_m$ & &
  $v_m$ & $1$ & $0$ & $v_m$ & : & $c_m$ \\
  \end{tabular}
  \end{center}
By construction, we have $U[AB] = R$ and $U[BC] = S$, also $V[BC] = S$ and $V[CD] = T$, and $W[AB] = R$ and $W[CD] = T$.
By the assumption that the hypergraph $P_3$ has the \ltgc~for $\mathbb{K}$-relations, there is a $\mathbb{K}$-relation $Y(ABCD)$ that witnesses the global consistency of $R,S,T$. Since $Y[BC] = S$, for every tuple $(a,b,c,d)$ in the support $Y'$ of $Y$, we have $b=c=0$ or $b=c=1$. Similarly, since $Y[AB] = R$, we have that if $b = 0$ then $a = u_i$ for some $i \in [m]$, and since $Y[CD] = T$, we have that if $c = 0$ then $d = v_j$ for some $j \in [n]$. Now, set $d_{ij} := Y(u_i,0,0,v_j)$ for every $i \in [m]$ and $j \in [n]$. For every $i \in [m]$ we have
\begin{equation*}
    \sum_{j \in [n]} d_{ij} = \sum_{j \in [n]} Y(u_i,0,0,v_j) =
    \sum_{(u_i,0,c,d) \in Y'} Y(u_i,0,c,d) = R(u_i,0) = b_i,
\end{equation*}
where the first equality follows from the choice of $d_{ij}$, the
second follows from the above-mentioned properties of the tuples $(a,b,c,d)$ in the support $Y'$ of $Y$, the third follows from $Y[AB] = R$, and the last follows from the choice of $R$. Similarly,
for every $j \in [n]$ we have
\begin{equation*}
    \sum_{i \in [m]} d_{ij} = \sum_{i \in [m]} Y(u_i,0,0,v_j) =
    \sum_{(a,b,0,v_j) \in Y'} Y(a,b,0,v_j) = T(0,v_j) = c_j,
\end{equation*}
with very similar justifications for each step. This proves that 
$D = (d_{ij} : i \in [m], j \in [n])$ is a solution to the balanced instance of the transportation property of $\mathbb{K}$ given by the vectors $(b_1,\ldots,b_m)$ and $(c_1,\ldots,c_n)$, which completes the proof.
\end{proof}

By combining Theorems~\ref{thm:necessary} and~\ref{thm:allequivalent}, we obtain the
following result.

\begin{corollary} \label{cor:generalization} Let~$\mathbb K$ be a positive commutative monoid that has the \ftp. For every hypergraph
$H$, 
  the following statements are
  equivalent:
\begin{enumerate} \itemsep=0pt
\item $H$ is an acyclic hypergraph.
\item $H$ has the \ltgc~for $\mathbb K$-relations.
\end{enumerate}
\end{corollary}

Since  the \ftp~holds for $\mathbb{B}$ and since the $\mathbb{B}$-relations are 
the ordinary relations, Corollary~\ref{cor:generalization} contains
the Beeri-Fagin-Maier-Yannakakis Theorem~\ref{thm:BFMY} as a special
case. In the next section, we identify
several different classes of positive commutative monoids that have 
the \ftp; therefore, 
Corollary~\ref{cor:generalization} applies to all such monoids.

\section{Monoids with the Transportation Property} \label{sec:examples}
                    
 We now turn to the question of identifying broad classes of positive commutative monoids
that do have the transportation property. We give five different types  of such monoids:
\begin{enumerate} \itemsep=0pt
\item[--] monoids that can be expanded to a semiring with the standard join; 
\item[--] monoids that can be expanded to a semifield with the Vorob'ev join; 
\item[--] monoids to which the \emph{Northwest Corner Method} applies;
\item[--] power monoids; 
\item[--] free commutative monoids. 
\end{enumerate}
For the first two types of monoids, the solution to the system of equations of a balanced instance of
the transportation problem can be obtained using an operation that, when  interpreted on $\mathbb{K}$-relations, 
generalizes the relational join of  ordinary relations (i.e., ${\mathbb B}$-relations) in the first case and the  Vorob'ev join of probability distributions
in the second. For the third type of monoids, the solution is not
obtained using an operation but via a procedural method that we call the \emph{Northwest Corner Method} and comes
inspired by the theory of linear programming.

\subsection{Expansion to a Semiring and the Standard Join} \label{sec:semiring-expansion}

To motivate the concepts and results in this section, let us first consider  ordinary  relations.  
As discussed earlier, the ordinary relations coincide with the $\mathbb B$-relations, where ${\mathbb B} = (\{0,1\},\lor, 0)$ is the Boolean commutative monoid.  Also,  $\mathbb B$ has the inner consistency property and, moreover,  there is a natural witness to the consistency of two consistent $\mathbb B$-relations. Specifically, 
if~$R$ and~$S$ are ordinary relations, then the \emph{relational join}
of~$R$ and~$S$, denoted by~$R \Join S$, is the ordinary relation that
consists of all~$XY$-tuples~$t$ such that~$t[X]$ is in~$R$ and~$t[Y]$
is in~$S$. It is well known and easy to see that if~$R$
and~$S$ are consistent ordinary relations, then~$R \Join S$ is a witness to
their consistency. Note, however, that the relational join is defined using the conjunction $\land$ of two Boolean values, since
\begin{equation}
(R\Join S)(t) = R(t[X]) \land S(t[Y]). 
\end{equation}
This suggests that for some positive commutative monoids ${\mathbb K}=(K, +, 0)$, 
witnesses to the consistency of two $\mathbb K$-relations
may be  explicitly constructed 
 using operations other than the operation $+$ of $\mathbb K$.  As we will see in this section,  certain positive commutative monoids can be shown to have the inner consistency property via an expansion to 
 \emph{semirings}
 with additional properties, where  witnesses to the consistency of two $\mathbb K$-relations can be  explicitly constructed   using the operations in the expansion.
 
\paragraph{Additively Positive Semirings} 
A \emph{semiring}
is a structure ${\mathbb K}=(K,+, \times, 0,1)$ with the following properties:
 \begin{itemize} \itemsep=0pt
     \item $(K,+, 0)$ and $(K,\times,1)$ are commutative monoids;
     \item $\times$ distributes over $+$, i.e., 
     $p\times (q + r) = p\times q + p \times r$, for all $p, q, r \in K$. .
     \item $0$ annihilates, i.e., $0\times p = p\times 0 = 0$, for all $p\in K$. 
 \end{itemize}
An \emph{additively positive semiring} is a semiring ${\mathbb K}=(K,+, \times, 0,1)$ whose additive reduct~$(K,+,0)$ is a positive monoid, i.e., 
$p+q  = 0$ implies that  $p=0$ and $q=0$. 
 
The Boolean semiring ${\mathbb B}= (\{0,1\}, \lor, \land, 0, 1)$, the \emph{bag} semiring
${\mathbb N}=(Z^{\geq 0}, +, \times, 0,1)$ of the non-negative integers, and the  semiring ${\mathbb R}^{\geq 0}= (R^{\geq 0}, +, \times, 0,1)$ of the non-negative real numbers, where $+$ and $\times$ are the standard arithmetic operations, are examples of additively positive semirings. Note that, to keep the notation simple, we used the same symbol ($\mathbb B$, $\mathbb N$, ${\mathbb R}^{\geq 0}$) to denote  both
the original positive commutative monoid and its expansion to a semiring. We will use a similar convention in the sequel.

\paragraph{The Standard Join}
Let~${\mathbb K}=(K,+,\times, 0,1)$ be an additively positive semiring. If~$R(X)$ and~$S(Y)$ are two~$\mathbb K$-relations, then the
\emph{standard~$\mathbb K$-join} of~$R$ and~$S$, denoted by~$R \standardjoin{\mathbb K} S$, is
the~$\mathbb K$-relation~$W(XY)$ defined for every~$XY$-tuple~$t$ by
the equation
\begin{equation}
W(t) = R(t[X]) \times S(t[Y]). \label{eqn:standard-k-join}
\end{equation} 
Clearly, if~$\mathbb K$ is the Boolean semiring $\mathbb B$, then the
standard~$\mathbb K$-join coincides with the relational join. Unfortunately, if $\mathbb K$ is an arbitrary positive semiring, then the standard~$\mathbb K$-join need not always be a 
witness to consistency of two consistent $\mathbb K$-relations. 
For example,  consider the
positive commutative monoid ${\mathbb N}=(Z^{\geq 0}, +, 0)$ of the non-negative integers with addition and its expansion to the semiring ${\mathbb N} =(Z^{\geq 0}, +,\times, 0,1)$, where $+$ and $\times$ are the standard arithmetic operations. 
As pointed out in~\cite{DBLP:conf/pods/AtseriasK21}, the standard $\mathbb N$-join  need not witness the consistency of two consistent $\mathbb N$-relations.
To see this,
 consider the
${\mathbb N}$-relations
\begin{align*}
R(AB) & = \{ (1,2):{1},\; (2,2):{1} \}, \\
S(BC) & = \{ (2,1):{1},\; (2,2):{1} \}.
\end{align*}
Their standard $\mathbb N$-join is
$(R\standardjoin{\mathbb N} S)(ABC) = 
\{(1,2,1):{1}, (1,2,2):{1}, 
(2,2,1):{1}, (2,2,2):{1}\}$,
which clearly does not witness the consistency of  $R$ and $S$. In fact, it is easy to verify that 
the only ${\mathbb N}$-relations
that witness the consistency of $R$ and $S$ are
\begin{align*}
T_1(ABC) = \{ (1,2,2):{1},\; (2,2,1):{1} \}, \\
T_2(ABC) = \{ (1,2,1):{1},\; (2,2,2):{1} \}.
\end{align*}

In what follows, we will pinpoint
the class of additively positive semirings 
for which the inner consistency property holds for $\mathbb K$-relations with the 
standard $\mathbb K$-join 
witnessing the consistency of two consistent $\mathbb{K}$-relations. In such a case, we
say that \emph{the \icp~holds for~$\mathbb{K}$-relations via the
standard~$\mathbb{K}$-join}. 

\paragraph{Characterization} 
Our aim is to characterize the additively positive
semirings $\mathbb{K}$ for which
the inner consistency property holds for $\mathbb{K}$-relations via
the standard $\mathbb{K}$-join. For this we need two definitions.
Let~${\mathbb K} = (K,+,\times,0,1)$ be  a semiring. We say that~$\mathbb K$ is
\emph{additively absorptive} if for all~$p,q \in K$ it holds
that~$p + p \times q = p$. We say that~$\mathbb K$ is \emph{multiplicatively
  idempotent} if for all~$p \in K$ it holds that~$p \times p = p$. Being additively absorptive has three immediate consequences that we now discuss.  First, being additively absorptive is equivalent to having
that~$1+ q=1$ holds, for all~$q\in K$. Second, if~$\mathbb K$ is additively
absorptive, then~$\mathbb K$ is~\emph{additively idempotent},
i.e.,~$p + p = p$, for all~$p \in K$ (take $q=1$ in the identity $p+ p\times q = p$). Third, if $\mathbb K$ is additively absorptive, then $\mathbb K$ is additively positive. Indeed, suppose that $p$ and $q$ are two elements of $K$ such that
$p+q = 0$. Then $p = p + (p + q) = (p + p) + q = p + q = 0$, where the first and last equalities follow from the assumption
that $p+q = 0$, and the second and third equalities follow from associativity and additive idempotence, respectively. 
In a similar manner we get $q = (p+q)+q = p+(q+q) = p+q = 0$, hence $p = q = 0$.

\begin{proposition} \label{prop:absorptiveandstaridempotent}
  Let ${\mathbb K}$ be a
semiring.
  Then the following statements are equivalent.
\begin{enumerate} \itemsep=0pt
\item[(1)] $\mathbb K$ is additively absorptive and multiplicatively idempotent.
\item[(2)] $\mathbb K$ is additively positive and the \icp~holds for $\mathbb{K}$-relations
via the standard $\mathbb K$-join.
\end{enumerate}
\end{proposition}

\begin{proof} We prove the implications (1)~$\Longrightarrow$ (2) 
and (2)~$\Longrightarrow$~(1).

(1)~$\Longrightarrow$~(2). We argued already that the assumption that $\mathbb{K}$ is additively 
  absorptive implies that $\mathbb{K}$ is additively positive. 
  For the second part, for notational simplicity,
  consider two~$\mathbb K$-relations~$R(AB)$ and~$S(BC)$ such
  that~$R[B] = S[B]$. We will show that the
  standard~$\mathbb K$-join~$R\standardjoin{\mathbb K} S$ witnesses their
  consistency. Setting~$W := R \standardjoin{\mathbb K} S$, we will show
  that~$W[AB] = R$; the proof that~$W[BC] = S$ is similar.  We may
  assume that~$R$ and~$S$ have non-empty support or else, since~$\mathbb K$ is
  additively positive, the assumption~$R[B] = S[B]$ implies that both have empty
  support and then the claim is trivial. Let~$(a,b)$ be a tuple in
  the support of~$R$ and let~$p = R(a,b)$.  Then there are
  elements~$u$ and~$w$ in~$K$ such that~$R(b) = w = S(b)$
  and~$w = p + u$.  Let~$(b,c_1),\ldots,(b,c_m)$ be a list of the tuples
  in the support of~$S$ that join with~$(a,b)$, and
  let~$q_i = S(b,c_i)$ for~$i = 1,\ldots,m$.
  Then~$W(a,b) = \sum_{i=1}^m p \times q_i = p \times \sum_{i=1}^m q_i
  = p \times w$, where the last equality follows from the fact
  that~$R[b] = w = S[b]$. Therefore, we have
  that~$W(a,b) = p \times w = p \times (p + u)= p \times p + p \times u
  = p + p \times u = p$, where the last two equalities follow from the
  assumption that~$\mathbb K$ is both multiplicatively idempotent and
  additively absorptive.

(2)~$\Longrightarrow$ (1). The assumption that $\mathbb{K}$ is additively
  positive makes the definition of the~\icp~apply to $\mathbb{K}$-relations. 
  Assume it holds via the standard $\mathbb{K}$-join.
  We first show that~$\mathbb K$
  is~multiplicatively idempotent.  For this, take an arbitrary
  element~$p$ of~$K$ and consider the~$K$-relations~$R(AB)$
  and~$S(BC)$ given by~$R(a,b) = S(b,c) = p$, where~$a,b,c$ are three
  fixed values in the domains of the attributes~$A,B,C$,
  and~$R(r) = S(s) = 0$ for any other tuples~$r$
  and~$s$. Clearly,~$R[B] = S[B]$. By the hypothesis about~$\mathbb K$, the
  relations~$R$ and~$S$ are consistent and their consistency is
  witnessed by~$R \standardjoin{\mathbb K} S$.  Since~$R \standardjoin{\mathbb K} S$ takes
  value~$p \times p$ on the tuple~$(a,b,c)$ and~$0$ everywhere else,
  we conclude that~$p = p \times p$. Hence, since~$p$ was an 
  arbitrary element of $K$,~it follows that $\mathbb K$ is multiplicatively idempotent.  To show that~$\mathbb K$ is
  additively absorptive, consider two arbitrary elements~$p$ and~$q$
  of~$\mathbb K$ and the~$\mathbb K$-relations~$R(AB)$ and~$S(BC)$ given
  by~$R(a,b) = S(b,c) = p$ and~$R(a',b) = S(b,c') = q$, where~$b$ is a
  fixed value in the domain of~$B$, and~$a,a'$ and~$c,c'$ are fixed
  values in the domains of~$A$ and~$C$, respectively,
  and~$R(r) = S(s) = 0$ for any other tuples~$r$ and~$s$.
  Clearly~$R(b) = p + q = S(b)$ and hence~$R[B] = S[B]$. By the
  hypothesis about~$\mathbb K$, the relations~$R$ and~$S$ are consistent and
  their consistency is witnessed by~$R \standardjoin{{\mathbb K}} S$. 
  Since~$R \standardjoin{{\mathbb K}} S$
  takes value~$p \times p$ on the tuple~$(a,b,c)$, value~$p \times q$
  on the tuple~$(a,b,c')$, and value~$0$ on any other tuple that
  projects to~$(a,b)$, we conclude that~$p = p \times p + p \times q$.
  Since~$\mathbb K$ is multiplicatively idempotent, it follows
  that~$p = p + p\times q$. Hence, since~$p$ and~$q$ were
  arbitrary elements of $K$, it follows that~$\mathbb K$ is additively absorptive.
\end{proof}

Every semiring~${\mathbb K}= (K, +, \times,0,1)$ that is additively
absorptive and multiplicatively idempotent is a bounded distributive
lattice. To see this, recall that a lattice is an algebraic
structure~${\mathbb M} = (M, \vee, \wedge)$ such that the \emph{join} and \emph{meet} operations
$\vee$ and $\wedge$ are
binary, commutative and associative, and satisfy the
absorption laws~$x \vee (x\wedge y) = x$ and~$x \wedge (x \vee y) = x$.
Recall also that a lattice is bounded if it has a least element $0$ and a greatest element $1$ with respect to the partial order $\leq$ defined
by $a\leq b$ if $a\vee b = b$ (equivalently, if $a\wedge b= a$), for all $a,b \in M$.
The first absorption law in the language of~$\mathbb K$
reads~$x + x \times y = x$, which holds for~$\mathbb K$ because $\mathbb K$ is additively
absorptive. 
For the second absorption law, we have
that~$x \times (x + y) = x \times x + x \times y = x + x \times y= x$
where the first equality holds by the distibutivity property for~$\mathbb K$,
the second equality holds by the multiplicative idempotence of~$\mathbb K$,
and the third one holds by the additive absorptiveness of~$\mathbb K$.  We also have that $0$ is the least element of $\mathbb K$ (viewed as a lattice) and $1$ is its greatest element, since $0+ q = q$ and $q+1=1$, for all $q\in K$. Furthermore, it is easy to verify that the converse is true, i.e., every bounded distributive lattice is an additively absorptive and multiplicatively idempotent  semiring. Thus,
the additively absorptive and
multiplicatively idempotent  semirings are precisely the bounded
distributive lattices. 

\begin{example} \label{ex:lattices}
Examples of bounded distributive lattices include the Boolean
semiring~$\mathbb{B} = (\{0,1\},\vee,\wedge,0,1)$, the powerset
semiring~$\mathbb{P}_A = (\mathcal{P}(A),\cup,\cap,\emptyset,A)$ for 
an arbitrary set~$A$, and every  max/min semiring~$\mathbb{M}_A = (A,\max,\min,a,b)$, where
$(A,\leq)$ is a totally ordered set with smallest element $a$  and greatest element $b$. Note that the
max/min semirings contain as special cases the fuzzy
semiring~$\mathbb{F} = ([0,1],\max,\min,0,1)$ and the  \emph{access
control} semirings, which are max/min semirings based on finite linear orders with each element indicating a different level of access control (``confidential'', ``secret'', and so on). Another example is the semiring $\mathbb{PB}(X)=(\textsc{PosBool}(X),\lor, \land, 0, 1)$, where $X$ is a set of variables and $\textsc{PosBool}(X)$ is the set  all \emph{Boolean positive expressions} (i.e., Boolean formulas over $X$ built from $0$, $1$, and variables from $X$ using $\lor$ and $\land$) and where two such  expressions are identified if they are logically equivalent. This semiring has been studied in the context of provenance for database queries (e.g., see \cite{DBLP:journals/mst/Green11}).
\end{example}

For each semiring ${\mathbb K}=(K,+,\times, 0,1)$ considered in Example~\ref{ex:lattices},
the underlying commutative monoid ${\mathbb K}=(K,+, 0)$ is positive, the inner consistency property holds for $\mathbb K$-relations, and 
the standard $\mathbb K$-join witnesses the consistency of two consistent $\mathbb K$-relations.

\subsection {Expansion to a Semifield and the Vorob'ev Join} \label{sec:semifield-expansion}

If the standard~$\mathbb K$-join does not always witness the consistency 
of two consistent~$\mathbb K$-relations, then a natural alternative to consider is what
we call the \emph{Vorob'ev~$\mathbb K$-join}. 
This, however, requires an expansion of the positive commutative monoid to a semifield.
By definition, 
a \emph{semifield} is a structure ${\mathbb K}=(K, +,\times, 0,1)$ with the following properties:
\begin{itemize} \itemsep=0pt
    \item ${\mathbb K}=(K, +,\times, 0,1)$ is a semiring. 
    \item For every  element $p\neq 0$ in $K$, there exists an element $q$ in $K$ such that
    $p\times q = 1 = q\times p$.
\end{itemize}
In other words, a semifield is a semiring such that  $(K\setminus \{0\}, \times, 1)$ is a group. 
Note that if $\mathbb K$ is a semifield, then for every $p\neq 0$, there is exactly one element $q$ such that
$p\times q = 1 = q\times p$ (if there were two such elements $q$ and $q'$, then $p\times q = 1$ implies that $q'\times p \times q = q'$, which implies that $q=q'$). This unique element $q$ is called the \emph{multiplicative inverse} of $p$ and is denoted by $1/p$.  As usual if $q\neq 0$ and $p$ is an arbitrary element of $K$, we will write $p/q$, or~$\frac{p}{q}$, for the element $p\times (1/q)$. 

An \emph{additively positive semifield} is a semifield ${\mathbb K}=(K, +,\times, 0,1)$ in which the underlying additive monoid $(K,+,0)$ is  positive.
Two well known examples of positive semifields are the 
semiring~$\mathbb{R}^{\geq 0} = (R^{\geq 0}, +, \times, 0, 1)$ of
non-negative real numbers and its rational substructure~$\mathbb{Q}^{\geq 0} = (Q^{\geq 0}, +, \times,
0, 1)$.

\paragraph{The Vorob'ev Join} Let ${\mathbb K}=(K, +,\times, 0,1)$ be a semifield.
If~$R(X)$ and~$S(Y)$ are two inner consistent~$\mathbb K$-relations (i.e., they 
satisfy~$R[X \cap Y] = S[X \cap Y]$),  then the \emph{Vorob'ev~$\mathbb K$-join}
of $R$ and $S$, denoted by~$R \vorobyevjoin{{\mathbb K}} S$,
is the~$\mathbb K$-relation~$W(XY)$ defined for every~$XY$-tuple~$t$ by the
equation
\begin{equation*}
W(t) = \frac{R(t[X]) \times S(t[Y])}{R(t[X \cap Y])} =
  \frac{R(t[X]) \times S(t[Y])}{S(t[X \cap Y])} \label{eqn:Vorobevjoin}
\end{equation*}
if~$R(t[X \cap Y]) = S(t[X \cap Y]) \not= 0$, and by~$W(t) = 0$
otherwise.  Note that the Vorob'ev~$\mathbb K$-join of two~$\mathbb K$-relations is
well-defined because  the two~$\mathbb K$-relations $R(X)$ and $S(Y)$ were assumed to be  inner consistent.

We say that \emph{the inner consistency property holds for~$\mathbb{K}$-relations~via the Vorob'ev $\mathbb K$-join} if the inner consistency property holds for $\mathbb K$-relations and, moreover,  the Vorob'ev $\mathbb K$-join witnesses the consistency of two consistent $\mathbb K$-relations. 

\begin{proposition} \label{prop:formalfractionscase}
  If~${\mathbb K}$ is an additively  positive semifield, 
  then the inner consistency property holds for $\mathbb K$-relations
via the Vorob'ev $\mathbb K$-join.
\end{proposition}

\begin{proof} Suppose that~$R$ and~$S$ are two inner consistent  $\mathbb K$-relations and
  let~$Z = X \cap Y$; i.e.,~$R[Z] = S[Z]$.  Therefore, their Vorob'ev
  $\mathbb K$-join~$W := R \vorobyevjoin{{\mathbb K}} S$ is a well-defined~${\mathbb K}$-relation.  We
  now check that for each~$X$-tuple~$r$, we have~$W[X](r) = R(r)$.  If~$r$
  is not in the support of~$R$, then~$W(t) = 0$ for
  every~$XY$-tuple $t$ with~$t[X] = r$ and
  hence~$W[X](r) = \sum_{t : t[X]=r} 0 = 0 =
  R(r)$. Suppose then that~$r$ is in the support of~$R$; in
  particular, by the assumption that~$R[Z] = S[Z]$ and the hypothesis that $\mathbb K$ is additively positive, we have~$S(t[Z]) = R(t[Z]) \not= 0$ for every~$XY$-tuple~$t$
  such that~$t[X] = r$. Therefore, we
  have 
  \begin{equation*}
  W[X](r) = \sum_{\newatop{t\in W':}{t[X] = r}} R(t[X]) \times
  S(t[Y]) / S(t[Z]) = R(r) \times \sum_{\newatop{t\in W':}{t[X] = r}}
  S(t[Y]) / S(t[Z]). 
  \end{equation*}
  Now note that~$t[Z] = t[X][Z] = r[Z]$
  whenever~$t[X] = r$, and that there is a bijection between the set
  of~$XY$-tuples~$t$ such that~$t[X] = r$ and the set
  of~$Y$-tuples~$s$ such that~$s[Z] = r[Z]$. Therefore, this last
  expression can be rewritten
  as
  \begin{equation*}
  R(r) \times \sum_{\newatop{s \in S':}{s[Z] = r[Z]}} S(s)/S(r[Z]) =
  R(r) \times S(r[Z])/S(r[Z]) = R(r),
  \end{equation*}
  where the last equality follows from the already argued fact that~$S(r[Z]) = S(t[Z]) \not= 0$. 
  This proves $W[X](r) = R(r)$. A symmetric argument shows that for
  each~$Y$-tuple~$s$ we have that~$W[Y](s) = S(s)$, and the proposition is
  proved.
\end{proof}    

\begin{example} \label{exam:tropical}
The semiring~$\mathbb{R}^{\geq 0} = (R^{\geq 0}, +, \times, 0, 1)$ of
non-negative real numbers and its rational
substructure~$\mathbb{Q}^{\geq 0} = (Q^{\geq 0}, +, \times,
0, 1)$ where mentioned before as examples of additively positive semifields. Other well-known examples 
include the tropical semirings, and their smooth variants, the log semirings:
\begin{align}
    & \mathbb{T}_{\min} = ((-\infty,+\infty],\min,+,+\infty,0)  &  & \mathbb{T}_{\max} = ([-\infty,+\infty),\max,+,-\infty,0) 
    \label{eqn:trop} \\
    & \mathbb{L}_{\min} = ((-\infty,+\infty],\oplus_{\min},+,+\infty,0) &  & \mathbb{L}_{\max} = ([-\infty,-\infty),\oplus_{\max},+,-\infty,0)
    \label{eqn:smoothtrop}
\end{align}
where~$x \oplus_{\min} y = -\log(e^{-x} + e^{-y})$ and~$x \oplus_{\max} y = \log(e^{x} + e^{y})$, with the conventions that $e^{-\infty} = 0$ and $\log(0) = -\infty$.  
In all four cases the multiplicative inverse of the semifield is the standard inverse of addition over $(-\infty,+\infty)$.
It is obvious that $\mathbb{T}_{\min}$ is additively positive; furthermore,  $\mathbb{L}_{\min}$ is additively positive because~$-\log(e^{-x} + e^{-y}) = +\infty$ if and only if~$e^{-x} +
e^{-y} = 0$  if and only if~$x = y = +\infty$. Dually, the semirings $\mathbb{T}_{\max}$ and $\mathbb{L}_{\max}$ are  additively positive.
\end{example}

For each semiring ${\mathbb K}=(K,+,\times, 0,1)$ considered in Example \ref{exam:tropical},
the underlying positive commutative monoid ${\mathbb K}=(K,+, 0)$ is positive,
the inner consistency property for $\mathbb K$-relations holds, and the Vorob'ev $\mathbb K$-join witnesses the consistency of two consistent $\mathbb K$-relations.

\subsection{Northwest Corner Method} \label{sec:northwest}

In the previous two sections, we established the inner consistency property for different classes of positive commutative monoids by expanding them to richer algebraic structures. In this section, we will establish the inner consistency property for certain positive commutative monoids without expanding them. There will be a  trade-off, however, in the sense that the witnesses to the consistency of two consistent relations will be obtained via an algorithm, instead of an explicit construction such as the standard join or the Vorob'ev join. In return, the witnessing relations will be \emph{sparse} in that their supports consist of relatively few tuples. This is in contrast to the standard join and the Vorob'ev joins whose supports, in general, consist of a large number of tuples. We will quantify these notions later in this section.

\paragraph{Canonical Order and Cancellativity}
Let ${\mathbb K}=(K,+,0)$ be a positive commutative monoid.
Consider the
binary relation~$\sqsubseteq$ on~$K$ defined, for all~$b,c \in K$,
by~$b \sqsubseteq c$ if and only if there exists some~$a\in K$ such that~$b
+ a = c$. The binary relation~$\sqsubseteq$ is  reflexive and
transitive, and is hence a pre-order, called the \emph{canonical pre-order} of $\mathbb K$. 

\begin{itemize} \itemsep=0pt
\item $\mathbb{K}$ is \emph{cancellative} if $a+b=a+c$ implies $b=c$,
for all $a,b,c \in K$,
\item $\mathbb K$ is \emph{weakly cancellative} if $a+b=a+c$ implies $b=c$ or $b=0$ or $c=0$, for all $a,b,c \in K$,
\item $\mathbb K$ is \emph{totally canonically pre-ordered}
if $b \sqsubseteq c$ or $c \sqsubseteq b$, for all $b,c \in K$.
\end{itemize}

Let us consider some examples before proceeding. The positive commutative monoid ${\mathbb N}=(Z^{\geq 0}, +, 0)$ of the non-negative integers is cancellative and totally canonically preordered; in fact, its canonical pre-order is a total order.  These  properties are also shared by the positive commutative monoids
${\mathbb Q}^{\geq 0}=(Q^{\geq 0}, +,0)$ and
${\mathbb R}^{\geq 0}=(R^{\geq 0}, +,0)$ of the non-negative rational numbers and the non-negative real numbers.

Consider the positive commutative monoid ${\mathbb R}_1=(\{0\}\cup [1,\infty),+,0)$ where the universe
is the set of non-negative reals with a gap in the interval $[0,1]$ as only the endpoints of that interval are maintained..
The operation is the standard addition of the real numbers. This monoid is cancellative, but it is not totally canonically pre-ordered because if $b$ and $c$ are different elements between $1$ and $2$, then neither $b\sqsubseteq c$ nor $c\sqsubseteq b$ holds.  The 3-element positive commutative monoid ${\mathbb N}_2=(\{0,1,2\}, \oplus, 0)$ discussed in Section~\ref{sec:not-sufficient} is totally canonically pre-ordered because $1\oplus 1 = 2$, but it is not weakly cancellative 
because $2\oplus 1 = 2 = 2\oplus 2$ but $1\neq 2$, $2\neq 0$, $1\neq 0$.

\paragraph{Northwest Corner Method}
We will show that if a positive commutative monoid $\mathbb K$ is weakly cancellative and totally canonically pre-ordered, then the inner consistency property for $\mathbb K$-relations holds. In fact, we will establish that every such monoid  has the \ftp~introduced in Section \ref{sec:sufficient}. This will be achieved by using  the \emph{northwest corner method} of linear programming for finding solutions for the 
transportation problem.

Intuitively, the northwest corner method  starts by assigning a value to the variable in the northwest corner of the system of equations, eliminating at least one equation, and iterating this process by considering next the variable in the northwest corner of the resulting system.
Unlike  the case of linear programming,
 here we cannot subtract values; instead, we have to carefully use the assumption  that the monoid is weakly cancellative and totally canonically pre-ordered.

\begin{proposition} \label{prop:northwest}
If $\mathbb K$ is positive commutative monoid that 
is weakly cancellative and totally canonically pre-ordered, 
then $\mathbb K$ has the \ftp.
\end{proposition}
\begin{proof}
Let $\mathbb{K} = (K,+,0)$ be a monoid that satisfies the hypothesis of the proposition at hand.
We need to show that for every two positive integers $m$ and $n$, every $m$-vector $(b_1,\ldots,b_m) \in K^m$ and every $n$-vector $(c_1,\ldots,c_n) \in K^n$ with $b_1+\cdots+b_m =c_1+\cdots+c_n$, the following system of $m + n$ equations on $mn$ variables has a solution in $\mathbb{K}$. The first $m$ equations are written horizontally, and the next $n$ are written vertically:
\begin{equation*}
\begin{array}{ccccccccc}
x_{11} & + & x_{12} & + & \cdots & + & x_{1n} & =      & b_1  \\
+      &   & +      &   &        &   & +      &        &      \\
x_{21} & + & x_{22} & + & \cdots & + & x_{2n} & =      & b_2  \\
+      &   & +      &   &        &   & +      &        & \\
\vdots &   & \vdots &   & \ddots &   & \vdots &        &  \\
+      &   & +      &   &        &   & +      &        & \\
x_{m1} & + & x_{m2} & + & \cdots & + & x_{mn} & =      & b_m  \\
\shortparallel  &   & \shortparallel     &   &        &   & \shortparallel     &        & \\
c_1    &   & c_2    &   &        &   & c_n    &        &
\end{array}
\end{equation*}
Note that, by the positivity of $\mathbb K$, we may assume that $b_i\neq 0$ and $c_j\neq 0$ for all $i \in [m]$ and $j \in [n]$. Indeed, if, say,  $b_i=0$, then each variable $x_{ij}$ in the $i$-row of the system must take value $0$, hence the equation in that row and all variables appearing in that row can be eliminated. 

We proceed by induction on the sum $m+n$, which is the total number of equations in the system. 
We take the pairs $(m,n)$ with $m = 1$ or $n = 1$ as the base cases of induction.
If $m = 1$, then we can set $x_{1j} = c_j$ for $j = 1,\ldots,n$ and we get a solution since $c_1 + \cdots + c_n = b_1$. 
Similarly, if $n = 1$,
then we can set $x_{i1} = b_i$ for $i = 1,\ldots,m$ and we get a solution since $b_1 + \cdots + b_m = c_1$.
Let then the pair $(m,n)$ be such that $m \geq 2$ and $n \geq 2$, so $k := m + n \geq 4$, and assume that the induction hypothesis holds for all  systems with $m+n<k$.  Let us consider $b_1$ and $c_1$. Since $\mathbb K$
is totally canonically pre-ordered, we have that  $b_1=c_1$ holds or $b_1\sqsubseteq c_1$ holds or
$c_1\sqsubseteq b_1$ holds (more than one of these conditions may hold at the same time).

If $b_1=c_1$,  we set $x_{11}=b_1$, we set  $x_{1j}=0$ for $j = 2,\ldots,n$, and we set $x_{i1}=0$ for $i = 2,\ldots,m$. 
This assignment satisfies the equations 
\begin{equation*}
\begin{array}{ccccccccc}
   x_{11} & + & x_{12} & + & \cdots & + & x_{1n} & = & b_1 \\
   x_{11} & + & x_{21} & + & \cdots & + & x_{m1} & = & c_1. 
\end{array}
\end{equation*}
After eliminating from the other equations the variables that are set to $0$ in these two equations, we are left with the following system
of $m+n-2$ equations on $(m-1)(n-1)$ variables. Again the first $m-1$
equations are written horizontally, and the next $n-1$ are written
vertically:
\begin{equation*}
\begin{array}{ccccccccc}
x_{22} & + & \cdots & + & x_{2n} & =      & b_2  \\
+      &   &        &   & +      &        & \\
\vdots &   & \ddots &   & \vdots &        &  \\
+      &   &        &   & +      &        & \\
x_{m2} & + & \cdots & + & x_{mn} & =      & b_m  \\
\shortparallel     &   &        &   & \shortparallel     &        & \\
c_2    &   &        &   & c_n    &        &
\end{array}
\end{equation*}

We claim that this system is  a 
balanced 
instance of the transportation problem, i.e., $b_2+\cdots +b_m= c_2+\cdots +c_n$. Indeed, we have that $b_1+b_2+\cdots +b_m = c_1+c_2+\cdots +c_n$ and  
$b_1=c_1$, which means that
$b_1+b_2+\cdots +b_m = b_1+c_2+\cdots +c_n$. Since all the $b_i$'s and the $c_j$'s are different from $0$, the positivity of $\mathbb K$ implies that $b_2+\cdots +b_m \neq 0$ and $c_2+\cdots+c_n\neq 0$. Since $\mathbb K$ is weakly cancellative, we conclude that 
 $b_2+\cdots +b_m= c_2+\cdots +c_n$. By induction hypothesis, the preceding system has a solution in $\mathbb K$, hence the original system also has a solution in $\mathbb K$.

 Next assume that $b_1\neq c_1$ and $b_1\sqsubseteq c_1$. This means that there is an element $a\in K$ such that $b_1 +a =c_1$. Moreover, $a\neq 0$ because $b_1\neq c_1$. We now set  $x_{11}=b_1$ and $x_{1j}=0$ for $j = 2,\ldots,n$. This assignment satisfies the equation 
\begin{equation*}
\begin{array}{ccccccccc}
   x_{11} & + & x_{12} & + & \cdots & + & x_{1n} & = & b_1. \\
\end{array}
\end{equation*}
 We eliminate from the other equations the variables that are set to $0$ in this equation, eliminate also $x_{11}$ from the equation of $c_1$, and replace $c_1$ by $a$. This results into the following system of $m+n-1$ equations on $n(m-1)$ variables
\begin{equation*}
\begin{array}{ccccccccc}
x_{21} & + & x_{22} & + & \cdots & + & x_{2n} & =      & b_2  \\
+      &   & +      &   &        &   & +      &        & \\
\vdots &   & \vdots &   & \ddots &   & \vdots &        &  \\
+      &   & +      &   &        &   & +      &        & \\
x_{m1} & + & x_{m2} & + & \cdots & + & x_{mn} & =      & b_m  \\
\shortparallel  &   & \shortparallel     &   &        &   & \shortparallel     &        & \\
a    &   & c_2    &   &        &   & c_n    &        &
\end{array}
\end{equation*}
We claim that this system is  a balanced instance of the  
transportation problem, i.e., we claim that $b_2+\cdots +b_m = a + c_2+\cdots + c_n$. Indeed, we have that $b_1+b_2+\cdots +b_m = c_1+c_2+\cdots +c_n$ and $b_1+a =c_1$, which means that
$b_1+b_2+\cdots +b_m = b_1+ a +c_2+\cdots+c_n$. Since $a\neq 0$ and since all the $b_i$'s and the $c_j$'s are different from $0$, the positivity of $\mathbb K$ implies that
$b_2+\cdots +b_m\neq 0$ and $c_2+\cdots+c_n\neq 0$. Since $\mathbb K$ is weakly cancellative, we conclude that
$b_2+\cdots +b_m = a + c_2+\cdots + c_n$.  By induction hypothesis, the preceding system has a solution in $\mathbb K$, hence the original system also has a solution in $\mathbb K$.

The remaining case $b_1\neq  c_1$ and $c_1\sqsubseteq b_1$ is similar to the previous one with the roles of $b_1$ and $c_1$ exchanged.
\end{proof}

\paragraph{Northwest Corner Joins} By combining the proof of the implication (1) $\Longrightarrow $ (2) in Theorem~\ref{thm:allequivalent} with 
the northwest corner method  described in the proof of Proposition~\ref{prop:northwest},  we   obtain a procedure that computes a witness of the consistency of two consistent $\mathbb{K}$-relations, provided the monoid  $\mathbb{K}$ meets the conditions
of Proposition~\ref{prop:northwest}. We make this procedure explicit in what follows.
Mirroring the earlier state of affairs with the standard join and the Vorob'ev join, here we 
say that \emph{the inner consistency property holds for $\mathbb{K}$-relations
via the northwest corner method}. To be clear, though, it should be noted that in contrast to the standard join 
and the Vorob'ev join considered earlier, 
the witnesses of consistency that will be produced by the northwest corner method will 
not be canonical. In other words, their construction involves some arbitrary choices during 
the execution of the procedure, 
and while any choices will lead to a correct witness of consistency, different
choices may lead to different witnesses. To reflect this multitude of witnesses,
we refer to them as \emph{northwest corner joins}; in plural.

To describe the procedure that computes a witness of the consistency of two inner consistent
$\mathbb{K}$-relations $R(X)$ and $S(Y)$, 
let us assume that the monoid $\mathbb{K} = (K,+,0)$ is fixed at the outset and that it 
is positive, commutative, weakly cancellative,  and totally canonically
pre-ordered. Our goal is to produce a $\mathbb{K}$-relation
$W(XY)$ that witnesses the consistency of $R(X)$ and $S(Y)$,  i.e., $W(XY)$ is 
such that $W[X] = R$ and $W[Y] = S$. Write $X = AB$ and $Y = AC$, where $A,B,C$ are disjoint sets 
of attributes. First we enumerate
the tuples $a_1,\ldots,a_r$ in the supports $R[A]' = S[A]'$ of the marginals on the common attributes, where the equality
between the supports follows
from Lemma~\ref{lem:easyfacts1} and the assumption that $R$ and $S$ are inner consistent,
and $\mathbb{K}$ is positive. 
For each $k = 1,\ldots,r$, we enumerate the $B$-tuples $b_{k1},\ldots,b_{km_k}$
such that $R(a_k,b_{kj}) \not= 0$ for $j = 1,\ldots,m_k$, and the $C$-tuples $c_{k1},\ldots,c_{kn_k}$ such 
that $S(a_k,c_{kj}) \not= 0$ for $j = 1,\ldots,n_k$. Since $R[A] = S[A]$ holds by inner consistency, we have that for each
$k = 1,\ldots,r$ the equality 
\begin{equation}
R(a_k,b_{k1}) + \cdots + R(a_k,b_{km_k}) = S(a_k,c_{k1}) + \cdots + S(a_k,c_{kn_k})
\end{equation}
holds, so we are dealing with a different 
balanced 
instance
of the transportation problem for each $k = 1,\ldots,r$. By applying the northwest corner method as described in the proof of Proposition~\ref{prop:northwest}
to each such instance with $k = 1,\ldots,r$, we find a values $x_{k,ij}$ in $K$ that solve the corresponding system of
equations. From those, we build the $\mathbb{K}$-relation $W(ABC)$ by setting
\begin{equation*}
W(a_k, b_{kj}, c_{kj}) := x_{k,ij}
\end{equation*}
for all $k=1,\ldots,r$,
all $j = 1,\ldots,m_k$, and all $i = 1,\ldots,n_k$, and $W(a,b,c) = 0$ for any other $ABC$-tuple $(a,b,c)$. 
It is a matter of unfolding the definitions to check that this $\mathbb{K}$-relation $W(ABC)$
satisfies $W[AB] = R$ and $W[AC] = S$, hence it witnesses the consistency of $R$ and $S$. We say that $W$
is  \emph{a northwest corner join for $R$ and $S$}. 

As an immediate corollary of Proposition~\ref{prop:northwest} and the description of the 
procedure for computing a northwest corner join, 
we obtain the following proposition.

\begin{proposition} \label{prop:northwestcornerjoinmethod}
If $\mathbb{K}$ is a positive commutative monoid that is weakly cancellative and totally canonically pre-ordered,
then the inner consistency property holds for $\mathbb{K}$-relations via the northwest corner method.
\end{proposition}

As indicated earlier, 
the witness $W$ that is obtained from applying the northwest corner method to $R$ and $S$ is not canonically defined
in the sense that its definition depends on the choice of the 
orders in the enumerations $b_{k1},\ldots,b_{km_k}$ and $c_{k1},\ldots,c_{kn_k}$ featuring above. One of the advantages of 
the northwest corner method, however, 
is that it always produces a \emph{sparse} $\mathbb{K}$-relation in the sense of the following
proposition.

\begin{proposition} \label{prop:sparse}
  Let $\mathbb{K}$ be a positive commutative monoid such that
  the inner consistency property for~$\mathbb{K}$-relations
  via the northwest corner method.
  Let $R(X)$ and $S(Y)$ be two inner consistent $\mathbb{K}$-relations,
  and let $W$ be a northwest corner join for $R$ and $S$. Then
  the support size $|W'|$ of $W$ is bounded by the sum of
the support sizes $|R'|$ and $|S'|$, i.e., 
  \begin{equation}
|W'| \leq |R'|+|S'|. \label{eqn:uppderounb}
  \end{equation}
\end{proposition}

\begin{proof}
    Consider the procedure that computes $W$ from $R$ and $S$ as described above. Write $X = AC$ and $Y = BC$,
    where $A,B,C$ are disjoint sets of attributes. In the proof of Proposition~\ref{prop:northwest} applied to the system 
corresponding to $a_k$, where $a_1,\ldots,a_r$ is the enumeration of $R[A]'=S[A]'$, at each iteration at least one row or column (or both) is eliminated
while adding exactly one tuple in the support of $W$. At the base cases, 
either the single remaining row is eliminated while adding
one tuple in the support of $W$ for each remaining column, or
the single remaining column is eliminated while adding one tuple in the support
of $W$ for each remaning row. Thus, for each separate $k$ at most one tuple for each row or column is added,
which gives the bound in~\eqref{eqn:uppderounb}.
\end{proof}

The sparsity of the support size $|W'|$ of any northwest corner join $W$ for $R$ and $S$
contrasts with the standard join, and with the Vorob'ev join, whose support sizes
could grow multiplicatively as in~$|R'|\cdot|S'|$.

Finally, we point out that for most examples of positive monoids, the operations that are involved in the
computation of a northwest corner join $W$ for $R$ and $S$ can be performed efficiently. In particular, this the case
for the monoid $\mathbb{N} = (Z^{\geq 0},+,0)$ of the natural numbers with addition when the numbers
are represented in binary notation. This is the prime example of a positive commutative monoid that has
the transportation property via the northwest corner method. We discuss this example along with several 
others next.

\begin{example} \label{ex:bagmonoidnw}
Since the positive commutative monoid ${\mathbb N}=(Z^{\geq 0},+,0)$ of the non-negative integers is cancellative and totally canonically ordered, Proposition \ref{prop:northwest} implies that $\mathbb N$ has the inner consistency property via the northwest corner method, 
hence every acyclic hypergraph has the \ltgc~for $\mathbb N$-relations; the latter property was established via a different argument in \cite{DBLP:conf/pods/AtseriasK21}.  
An example of similar flavor to $\mathbb{N}$ is the positive monoid~$\mathbb{N}/b^{\mathbb{N}} =
  (\{m/b^n : m,n \in \mathbb{N}\}, +, 0)$ of terminating
  fractions in base~$b$, where~$b \geq 2$ is a natural number. This monoid is additively cancellative and totally canonically ordered; in fact, its canonical order is the natural order of the rational
  numbers restricted to the terminating fractions.  The non-negative reals $\mathbb{R}^{\geq 0} = (R^{\geq 0},+,0)$
  and the  non-negative rationals $\mathbb{Q}^{\geq 0} = (Q^{\geq 0},+,0)$ are also positive, totally
  ordered, and additively cancellative commutative monoids.
\end{example}

\begin{example}
For an application of a different flavor, consider the positive commutative monoid ${\mathbb M}_2=(\{0,1,2\},\oplus', 0)$, where
$1\oplus' 1=2$, $1\oplus' 2 =1= 2\oplus' 1$, and $2\oplus' 2=2$.
It is easy to see that ${\mathbb M}_2$ is weakly cancellative (but not cancellative) and totally canonically pre-ordered. Thus,
${\mathbb M}_2$ has the inner consistency property and every acyclic hypergraph has the \ltgc~for ${\mathbb M}_2$-relations, 
unlike the positive commutative monoid ${\mathbb N}_2=(\{0,1,2\},\oplus,0)$.
\end{example}

\begin{example}
The additive monoids of the tropical
  semirings~$\mathbb{T}_{\min}$ and~$\mathbb{T}_{\max}$ from~\eqref{eqn:trop} are 
  non-examples since they are not weakly cancellative: 
  if $a,b,c \in (-\infty,+\infty]$ are such that $b \not= c$ and $a < b < c < +\infty$, then $\min(a,b) = \min(a,c)$,
  yet $b \not= c$ and $b \not= +\infty$ and $c \not= +\infty$. The~$\max$
  case is dual.
  In contrast, the additive monoids of the 
log semirings~$\mathbb{L}_{\min}$ and~$\mathbb{L}_{\max}$ from~\eqref{eqn:smoothtrop}, seen
  as smooth approximations of~$\mathbb{T}_{\min}$
  and~$\mathbb{T}_{\max}$, are totally canonically ordered
  and additively cancellative.
  For~$\mathbb{L}_{\min}$, the canonical
order~$\sqsubseteq$ is the reverse order~$\geq$
on~$(-\infty,+\infty]$, which is total. To see this, observe that for
all~$x,y \in (-\infty,+\infty]$ we have that~$x \sqsubseteq y$ if and
only if there exists~$z \in (-\infty,+\infty]$ such
that~$-\log(e^{-x}+e^{-z}) = y$, which happens if and only if there
exists~$z \in (-\infty,+\infty]$ such that~$e^{-y}-e^{-x} = e^{-z}$,
which is the case if and only if~$e^{-y}-e^{-x} \geq 0$, and hence if
and only if~$x \geq y$. The equivalence in which~$z$ drops out from
the equation holds by the combination of the following three facts:
first,~$e^{-z}$ is a non-negative real for every~$z \in
(-\infty,+\infty]$; second,~$e^{-y}-e^{-x}$ is a finite non-negative
real whenever~$x \geq y$; and, third, each finite non-negative real
number~$r$ can be put in the form~$e^{-z}$ for~$z = \log(1/r)$, which
is a value in~$(-\infty,+\infty]$, if we use the convention that~$\log(1/0) =
+\infty$. Further,~$\mathbb{L}_{\min}$ 
is additively cancellative since~$-\log(e^{-x} + e^{-z}) = -\log(e^{-y} +
e^{-z})$ if and only if~$e^{-x} + e^{-z} = e^{-y} + e^{-z}$, and hence
if and only if~$x = y$ because~$e^{-z}$ is finite for every~$z \in
(-\infty,+\infty]$. As usual, the cases of~$\mathbb{T}_{\max}$ and~$\mathbb{L}_{\max}$ 
are dual.
\end{example}

\begin{example}
Finally, consider next the non-negative version~$\mathbb{L}_{\min}^{\geq 0} =
([0,+\infty],\oplus_{\min},+,+\infty,0)$ of~$\mathbb{L}_{\min}$, and its
 dual, the non-positive version~$\mathbb{L}_{\max}^{\leq 0} =
 ([-\infty,0],\oplus_{\max},+,-\infty,0)$ of~$\mathbb{L}_{\max}$. The additive
   monoids of these are positive,
  canonically totally ordered, and additively
   cancellative. For~$\mathbb{L}_{\min}^{\geq 0}$, the canonical order
   is also the reverse natural order on~$[0,+\infty]$. To see this,
   follow the same argument as in the proof for its version over all
   reals noting that, if~$x,y \in [0,+\infty]$,
   then~$|e^{-y}-e^{-x}|\leq 1$. Since each real number~$r$ in the
   interval~$[0,1]$ can be put in the form~$e^{-z}$ for~$z =
   \log(1/r)$, which is in~$[0,+\infty]$ since~$r \in [0,1]$, the claim
 follows.
It should be pointed out that, unlike its version over all
reals~$\mathbb{L}_{\min}$, the non-negative log
semiring~$\mathbb{L}_{\min}^{\geq 0}$ is not a semifield 
because its multiplicative part, the addition of the real
numbers restricted to~$[0,+\infty]$, is not a group
on~$[0,+\infty]$. Furthermore, its additive part, the
operation~$\oplus_{\min}$ restricted to~$[0,+\infty]$, is not
absorptive. This means that $\mathbb{L}_{\min}^{\geq 0}$ is an example
of a semiring that is not covered by the cases considered in earlier sections. 
\end{example}

\subsection{Products and Powers} \label{sec:powers}

The purpose of this section is to show that the standard product composition of positive commutative
monoids inherits the \ftp~from its factors. This will give a way to produce new examples of monoids with the transportation property from
old ones.

Recall from Section~\ref{sec:prelims} the definition of the product monoid $\prod_{i \in I} \mathbb{K}_i$
for a finite or infinite indexed sequence of monoids 
$(\mathbb{K}_i : i \in I)$.
It is easy to check that if each $\mathbb{K}_i$ is a positive commutative monoid, then their product $\prod_{i \in I} \mathbb{K}_i$
is also 
a positive commutative monoid. Actually, many properties of the factors
are preserved in the product, except an important one: the canonical order of the product is not total in general,
even if that of each factor is. Because of this, the product construction 
will constitute a different source of monoids for which the transportation property cannot be derived from the constructions seen so far.

\paragraph{Powers and Finite Support Powers}
Recall from Section~\ref{sec:prelims} the definition of the power construction~$\mathbb{K}^I$.
We will need a variant $\mathbb{K}^I_\mathrm{fin}$ of $\mathbb{K}^I$, which we call  the \emph{finite support power} of $\mathbb{K} = (K,+,0)$.
Its elements are the \emph{finite support maps}
from the index set $I$ to the base set $K$. 
More precisely, the finite support
power $\mathbb{K}^I_\mathrm{fin}$ is the monoid whose
base set is the set of all maps $f : I \to K$ of finite support, i.e., the maps
for which $f^{-1}(0)$ is co-finite, with addition $f+g$ of two maps $f$ and $g$
defined also pointwise as in~\eqref{eqn:pointwise}.
Observe that if $f$ and $g$ have finite support, then $f+g$
also has finite support and, therefore, the operation is well defined. The neutral element of the power ${\mathbb K}^I$ is the
constant $0$ map, which of course has finite support. In the sequel, we treat maps $f : I \to K$ and indexed
sequences $f = (f(i) : i \in I) \in K^I$ interchangeably.

\begin{proposition} \label{prop:product} 
Let $I$ be a finite or infinite non-empty index set
and let $\mathbb{K}$ be a positive commutative monoid. 
The following statements
are equivalent:
\begin{enumerate} \itemsep=0pt
\item[(1)] $\mathbb{K}$ has the \ftp,
\item[(2)] $\mathbb{K}^I$ has the \ftp,
\item[(3)] $\mathbb{K}^I_\mathrm{fin}$ has the \ftp.
\end{enumerate}
\end{proposition}

\begin{proof} We close a cycle of implications (3) $\Longrightarrow$ (1) $\Longrightarrow$ (2) $\Longrightarrow$ (3).

(3) $\Longrightarrow$ (1).
First observe that $\mathbb{K}$ is isomorphic to a substructure of $\mathbb{K}^I_\mathrm{fin}$: consider the embedding $a \mapsto \hat{a}$ that sends $a \in K$ to the map $\hat{a} : I \to K$ defined by $\hat{a}(k_0) := a$ for some fixed index $k_0 \in I$
and $\hat{a}(k) := 0$ for every other index $k \in I \setminus \{k_0\}$. With this embedding, every balanced instance $b= (b_1,\ldots,b_m)$ and $c = (c_1,\ldots,c_n)$ of the transportation problem for $\mathbb{K}$  lifts to a balanced instance $\hat{b} = (\hat{b}_1,\ldots,\hat{b}_m)$ and $\hat{c} = (\hat{c}_1,\ldots,\hat{c}_n)$
of the transportation problem for $\mathbb{K}^I_\mathrm{fin}$. By~(3), this instance has
a solution, say $u = (u_{ij} : i \in [m], j \in [n])$, where each $u_{ij}$ is
an indexed sequence of finite support, say $u_{ij} = (u_{ij}(k) : k \in I)$. Furthermore, since  $\hat{b}_i(k) = \hat{c}_j(k) = 0$ for all $k \in I \setminus \{k_0\}$ and $\mathbb{K}$ is positive, we must  have that $u_{ij}(k) = 0$ for all $k \in I \setminus \{k_0\}$,
since  $u$ is a solution.
This means that~$u$ is indeed of the form $(\hat{d}_{ij} : i \in [m], j \in [n])$ where $d_{ij} := u_{ij}(k_0)$. Setting $D := (d_{ij} : i \in [m], j \in [n])$ we get
a solution to the balanced  instance
of the transportation problem for $\mathbb{K}$ given by $b$ and $c$, which proves that~(1) holds.

(1) $\Longrightarrow$ (2).
Let $b = (b_1,\ldots,b_m)$ and $c = (c_1,\ldots,c_n)$ 
be a balanced instance of the transportation problem for $\mathbb{K}^I$, where each $b_i$ and $c_j$ is an indexed sequence, say $b_i = (b_i(k) : k \in I)$ and $c_j = (c_j(k) : k \in I)$. We proceed by defining a solution component by component.
For each $k \in I$, the pair of vectors $b(k) := (b_1(k),\ldots,b_m(k))$ and $c(k) := (c_1(k),\ldots,c_n(k))$ is a balanced instance of the  transportation problem for
$\mathbb{K}$. By~(1), each such instance has a solution, say $d(k) = (d_{ij}(k) : i \in [m], j \in [n])$. It follows that the collection of indexed sequences 
$d := (d_{ij} : i \in [m], j \in [n])$, where $d_{ij} := (d_{ij}(k) : k \in I)$, is a solution to
the balanced instance of the transportation problem for  $\mathbb{K}^I$ given by $b$ and $c$,
which proves that~(2) holds.

(2) $\Longrightarrow$ (3).
Let $b = (b_1,\ldots,b_m)$ and $c = (c_1,\ldots,c_n)$ 
be a balanced instance of the transportation problem for $\mathbb{K}^I_{\mathrm{fin}}$, i.e., $b_i = (b_i(k) : k \in I)$ and 
$c_j = (c_j(k) : k \in I)$ have finite support and the balance condition holds. View this as a balanced 
instance of the transportation problem for $\mathbb{K}^I$ and, by~(2), let $d = (d_{ij} : i \in [m], j \in [n])$ be 
a solution over $\mathbb{K}^I$. Then, by the finite support condition on the $b_i$ and $c_j$ we have 
$d_{ij}(k) = 0$ for all but finitely many $k \in I$ because $\mathbb{K}$ is positive. This means that $d$ is then 
also a solution over $\mathbb{K}^I_{\mathrm{fin}}$, which proves that~(3) holds.
\end{proof}

\paragraph{Component-Based Join and its Sparsity}
Let $\mathbb{K}$ be a positive commutative monoid for which
the \icp~holds for $\mathbb{K}$-relations, and 
let $\Join_{\mathbb{K}}$ be a join operation that produces a witness of the consistency
of any two inner consistent $\mathbb{K}$-relations, i.e., if $R$ and $S$ are $\mathbb{K}$-relations
that are inner consistent, then $R$ and $S$ are consistent and $R \Join_{\mathbb{K}} S$ witnesses
their consistency. We say that the \icp~holds for $\mathbb{K}$-relations
\emph{via the join operation $\Join_{\mathbb{K}}$}.

Consider now the power monoids $\mathbb{K}^I$ and $\mathbb{K}^I_{\mathrm{fin}}$ for an index set $I$. 
The proof of the implications (1)~$\Longrightarrow$~(2)~$\Longrightarrow$~(3) 
in Proposition~\ref{prop:product} proceeds 
component by component.
In turn, by inspecting the proof of the implication (1)~$\Longrightarrow$~(2)
in Theorem~\ref{thm:allequivalent},
this means that if the~\icp~holds for $\mathbb{K}$-relations via a join operation
$\Join_{\mathbb{K}}$, then
the same join operation can be applied component by component 
to witness the consistency of any two inner consistent
$\mathbb{K}^I$-relations $R$ and $S$, or two inner consistent 
$\mathbb{K}^I_{\mathrm{fin}}$-relations~$R$ and~$S$. 
The result will
be denoted by $R \componentwisejoin{I}{\mathbb{K}} S$ and will 
be described more explicitly in the proof of Proposition~\ref{prop:sparseproduct} below,
where it is called the \emph{component-wise join} of $R$ and $S$.
In the terminology above, we say
that the \icp~holds for $\mathbb{K}^I$-relations, or $\mathbb{K}^I_{\mathrm{fin}}$-relations respectively, 
via the component-wise join~$\componentwisejoin{I}{\mathbb{K}}$. Furthermore, as we will see,
the sparsity of the witnesses of consistency of the factors
may be preserved in the following sense.

For a positive real number $c$, two consistent $\mathbb{K}$-relations $R(X)$ and $S(Y)$,
and a $\mathbb{K}$-relation $W(XY)$ that witnesses their consistency, we
say that $W$ is an \emph{$c$-sparse witness} if 
\begin{equation}
    |W'| \leq (|R'|+|S'|)c. \label{eqn:sparse}
\end{equation}
In Example~\ref{ex:bagmonoid}, we have seen that the bag monoid $\mathbb{N}$ has the inner 
consistency property via the northwest corner method and, hence, by Proposition~\ref{prop:sparse}, any
two inner consistent bags have a $1$-sparse witness of consistency. 

We say that the inner consistency property for $\mathbb{K}$-relations holds \emph{with sparse witnesses} 
if there exists a positive real number
$c$ such that for any two inner consistent $\mathbb{K}$-relations $R(X)$ and $S(Y)$ there is a $\mathbb{K}$-relation
$W(XY)$ that is an $c$-sparse witness of consistency of $R(X)$ and $S(Y)$.
If the $c$-sparse witness $W$ can be chosen as $R \Join_{\mathbb{K}} S$ for a join operation $\Join_{\mathbb{K}}$, then
we say that the join operation 
$\Join_{\mathbb{K}}$ \emph{produces sparse witnesses}, or that it \emph{produces $c$-sparse witnesses}, when the
$c$-factor is important.

\begin{proposition} \label{prop:sparseproduct}
Let $I$ be a finite or infinite non-empty index set
and let $\mathbb{K}$ be a positive commutative monoid such that
the inner consistency property holds for $\mathbb{K}$-relations via a join operation $\Join_{\mathbb{K}}$.
Then, the inner consistency property holds for $\mathbb{K}^I_{\mathrm{fin}}$-relations 
via the component-wise join operation $\componentwisejoin{I}{\mathbb{K}}$.
Furthermore, if the join operation $\Join_{\mathbb{K}}$ produces $c$-sparse witnesses for some positive
real $c$, then the component-wise join operation $\componentwisejoin{I}{\mathbb{K}}$ produces $cd$-sparse witnesses 
$R \componentwisejoin{I}{\mathbb{K}} S$ where $d$ is any 
bound on the maximum number of non-zero components 
in the annotation of any tuple in the (finite) supports of the $\mathbb{K}^I_{\mathrm{fin}}$-relations
$R$ or $S$. In particular, if $I$ is finite and the~\icp~holds for $\mathbb{K}$-relations with sparse
witnesses, then the~\icp~holds for $\mathbb{K}^I$-relations with sparse witnesses.
\end{proposition}

\begin{proof}
Suppose that $\mathbb{K}$ is a positive commutative monoid
such that the inner consistency property holds for $\mathbb{K}$-relations
via a join operation $\Join_{\mathbb{K}}$. Let $I$ be a  
finite or infinite non-empty 
index set and consider the finite support power monoid
$\mathbb{K}^I_\mathrm{fin}$. Let $R(X)$ and $S(Y)$
be two 
$\mathbb{K}^I_\mathrm{fin}$-relations that are inner consistent. 
In this proof we offer a more explicit description of the component-wise join
$R \componentwisejoin{I}{\mathbb{K}} S$ of $R$ and $S$,
and then use this more explicit description to analyze its sparsity.

First we
define two new $\mathbb{K}$-relations $R_0(X,C)$ and $S_0(Y,C)$, where
$C$ is a new attribute that does not appear in $XY$ and has the
index set $I$ as its domain of values, i.e., $\mathrm{Dom}(C) = I$.
These new $\mathbb{K}$-relations are populated by setting
\begin{align}
    R_0(r,i) := R(r)(i) \;\;\;\text{ and }\;\;\; S_0(s,i) := S(s)(i)
    \label{eqn:defr0s0}
\end{align}
for every $X$-tuple $r$, every $Y$-tuple $s$, and every index $i \in I$.
Observe that $R_0$ and $S_0$ are proper $\mathbb{K}$-relations, i.e.,
their supports are finite because the supports of $R$ and $S$ are finite, and each element $f$ 
in $\mathbb{K}^I_\mathrm{fin}$ has, by definition, 
finite support as a function that maps each index $i \in I$ to an element 
$f(i)$ of $\mathbb{K}$.
We claim that, since $R$ and $S$ are inner consistent, so are $R_0$ and $S_0$;
indeed, by setting $Z = X \cap Y$, we have 
\begin{equation}
    R_0[Z](t,i) = \sum_{r : r[Z] = t} R(r)(i) =
    R[Z](t)(i) = S[Z](t)(i) = \sum_{s : s[Z] = t} S(s)(i) =
    S_0[Z](t,i) \label{eqn:calculationinn}
\end{equation}
for every $Z$-tuple $t$ 
and every index $i \in I$.
The point of the definition of $R_0$ and $S_0$ is that they encode the 
$\mathbb{K}^I_\mathrm{fin}$-relations
$R$ and $S$ as $\mathbb{K}$-relations in a way that from a
$\mathbb{K}$-relation $W_0$ that witnesses the consistency of $R_0$
and $S_0$, it is possible to produce a 
$\mathbb{K}^I_\mathrm{fin}$-relation
$W$ that
witnesses the consistency of $R$ and $S$. Concretely, 
if we take the join $W_0 = R_0 \Join_{\mathbb{K}} S_0$ that
we assumed to exist as witness of the consistency of $R_0$ and $S_0$, 
then the 
$\mathbb{K}^I_\mathrm{fin}$-relation
$W$ that works is the one defined by the equation
\begin{equation}
    W(t)(i) = W_0(t,i) = (R_0 \Join_{\mathbb{K}} S_0)(t,i). \label{eqn:productjoin}
\end{equation}
for every $XY$-tuple $t$ and every index $i \in I$.
It is easy
to see that this agrees with
what we earlier described as applying the join operation $\Join_{\mathbb{K}}$
component by component; i.e., the component-wise join $R \componentwisejoin{I}{\mathbb{K}} S$
of $R$ and $S$.

For the sparsity analysis
first note that, by the choice of $d$, 
the $\mathbb{K}$-relations $R_0$ and $S_0$
have support sizes bounded by $|R'|d$ and $|S'|d$, respectively.
It follows that $R_0 \Join_{\mathbb{K}} S_0$ is a $c$-sparse witness of their consistency,
which means that its support size is at most $(|R'|+|S'|)cd$.
The $cd$ bound on the sparsity of $R \componentwisejoin{I}{\mathbb{K}} S$ 
now follows from the definition of the component-based join in~\eqref{eqn:productjoin}.
\end{proof}

Next we discuss examples of monoids for which the inner consistency
property can be derived using the product construction. We start with various collections of monoids of polynomials with coefficients over a monoid $\mathbb{K}$ and variables from a set of indeterminates $X$.

\begin{example} \label{ex:polynomials} \textit{Monoids of Polynomials}. Let $\mathbb{K}[x]$ be the monoid of
formal univariate polynomials with coefficients in the 
monoid $\mathbb{K}$ and a single indeterminate variable $x$.
More broadly, let $\mathbb{K}[X]$ be the monoid of
formal multivariate polynomials $\mathbb{K}[X]$ with
coefficients in the monoid and indeterminates in the set $X$. Here, $X$ is
a finite or infinite indexed set of commuting
variables, or indeterminates.
To view $\mathbb{K}[x]$ and $\mathbb{K}[X]$ as product monoids of the form $\mathbb{K}^I_\mathrm{fin}$, 
in both cases the indexed set $I$ is taken
as the collection of all monomials; that is to say, $I$ is $1,x,x^2,x^3,
\ldots$ in the univariate case, and $I$ is  the collection of monomials 
$X^\alpha$ in the multivariate case, where $\alpha : X \to \mathbb{N}$ is a map 
that takes each indeterminate to its degree with the condition that the \emph{total degree} $\sum_{x \in \alpha'} \alpha(x)$
is finite, where $\alpha'$ is the support of $\alpha$. The notation $X^\alpha$ is then a shorthand for 
the formal monomial $\prod_{x \in \alpha'} x^{\alpha(x)}$, where $\prod$ is a formal product operation for indexed sets. 
With this
notation, the polynomials in $\mathbb{K}[X]$ take
the form of formal sums
\begin{equation*}
\sum_{m \in c'} c(m) m,
\end{equation*}
where $c : I \to K$ is a coefficient map of 
finite support $c'$, where $I$ is the set of monomials. In this monoid, 
addition is defined component-wise on the coefficients:
\begin{equation*}
\sum_{m\in c'} c(m) m + \sum_{m \in d'} d(m) m = \sum_{m \in c' \cup d'} (c(m) + d(m)) m.
\end{equation*}
The same idea can be applied to polynomials of restricted types by restricting the indexed set $I$ of monomials. For example, the collection $\mathbb{K}[X]_{\mathrm{m}}$ of \emph{multilinear polynomials} with coefficients in $\mathbb{K}$ can be obtained by restricting $I$ to the set of monomials $X^{\alpha}$ that have $\alpha(x) \in \{0,1\} $ for each $x \in X$. Similarly, for an integer $d$, the collection $\mathbb{K}[X]_{\leq d}$ of \emph{total degree-$d$ polynomials} is obtained by restricting $I$ to the set of monomials $X^{\alpha}$ that have $\sum_{x \in \alpha'} \alpha(x) \leq d$. The collection $\mathbb{K}[X]_d$ of \emph{degree-$d$ forms} is obtained by restricting $I$ to the set of monomials $X^{\alpha}$ with $\sum_{x \in \alpha'} \alpha(x) = d$. The special case $\mathbb{K}[X]_1$ is the collection of \emph{linear forms} on the variables $X$ with coefficients in $\mathbb{K}$. The monoid $\mathbb{N}[X]_1$ will feature prominently in Section~\ref{sec:freemonoid}. Note that the elements in $\mathbb{N}[X]_1$ can be identified with the finite support maps $c : X \to Z^{\geq 0}$ that assign a non-negative integer to each indeterminate.

For all these examples, if $\mathbb K$ has the transportation property, so do the various monoids of polynomials $\mathbb{K}[x]$, $\mathbb{K}[X]$, $\mathbb{K}[X]_m$, etc., by
Proposition~\ref{prop:product}. Similarly, if the~\icp~holds for $\mathbb{K}$-relations with sparse witnesses, then the sparsity of witnesses is inherited for $\mathbb{K}[X]$-relations annotated by polynomials with few non-zero coefficients, by Proposition~\ref{prop:sparseproduct}.
\end{example}

\begin{example} \textit{Powersets revisited}. 
An example of a different flavour is the powerset monoid
$\mathbb{P}(A) = (\mathscr{P}(A),\cup,\emptyset)$ of a finite 
set $A$. This monoid is isomorphic to 
the product $\mathbb{B}^A$, 
where $\mathbb{B} = (\{0,1\},\vee,0)$ is the Boolean monoid, and
$A$ is viewed as a finite index set.
Note that it in this case it makes no difference whether we consider
$\mathbb{B}^A$ or $\mathbb{B}^A_{\mathrm{fin}}$ because the index set is 
finite and, therefore, any indexed sequence has finite support.

Similarly, the monoid
$\mathbb{P}_{\mathrm{fin}}(A) = (\mathscr{P}_{\mathrm{fin}}(A),\cup,\emptyset)$
of finite subsets of a countably infinite set $A$ is isomorphic to $\mathbb{B}^A_\mathrm{fin}$. It is also isomorphic
to the monoid $\mathbb{B}[X]$ of formal multivariate
polynomials with coefficients in $\mathbb{B}$ from the previous paragraph.
\end{example}

\begin{example} \textit{Additive monoids of provenance semirings}.
The semiring $\mathbb{N}[X]$ of formal multivariate polynomials with coefficients in $\mathbb{N}$ is the \emph{most informative} member of a well-studied hierarchy of \emph{provenance semirings} in database theory - see  Figure~\ref{fig:diagram}.

\begin{figure}
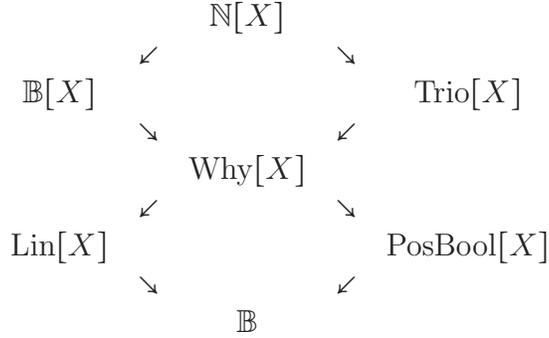

\begin{equation*}
\begin{array}{ccccc}
                  &          & \mathbb{N}[X]   &                     &                     \\
                  & \swarrow &                 & \searrow            &                     \\
 \mathbb{B}[X]    &          &                 &                     & \mathrm{Trio}[X]    \\
                  & \searrow &                 & \swarrow            &                     \\
                  &          & \mathrm{Why}[X] &                     &                     \\
                  & \swarrow &                 & \searrow            &                     \\
 \mathrm{Lin}[X] &          &                 &                     & \mathrm{PosBool}[X] \\
                  & \searrow &                 & \swarrow            &                     \\
                  &          & \mathbb{B}      &                     &  
\end{array}
\end{equation*}
\caption{The provenance semirings from \cite{DBLP:conf/pods/GreenKT07}. In this diagram, an arrow $\mathbb{K}_1 \rightarrow \mathbb{K}_2$ means that there is a surjective semiring homomorphism from $\mathbb{K}_1$ to $\mathbb{K}_2$.}
\label{fig:diagram}
\end{figure}

The $\mathrm{Trio}[X]$ semiring has a technical definition (see \cite{DBLP:journals/mst/Green11}) but it is easily seen to be equivalently defined
as $\mathbb{N}[X]_{\mathrm{m}}$, the semiring of multilinear multivariate polynomials with
coefficients in $\mathbb{N}$. The $\mathrm{Why}[X]$ semiring is
equivalently defined as $\mathbb{B}[X]_{\mathrm{m}}$, the semiring of multilinear multivariate polynomials with coefficients in $\mathbb{B}$. The $\mathrm{Lin}[X]$ semiring is defined to have $P_{\mathrm{fin}}(X) \cup \{\bot\}$ as its base set, where $P_{\mathrm{fin}}(X)$ denotes the collection of finite subsets of $X$ and $\bot$ is a fresh element, with addition and multiplication both defined as the union of sets, except for $\bot$ which is treated as the neutral element of addition and as the absorptive element of multiplication. Finally, the $\mathrm{PosBool}[X]$ semiring has as base set the collection of positive Boolean formulas  with variables in $X$ and constants $1$ and $0$ for true and false, identified up to logical equivalence. Its operations are the standard disjunction and conjunction of formulas for addition and multiplication, respectively.

For the questions of interest in this paper, only the additive monoid structure of these semirings matters. It should be clear that $\mathbb{N}[X]$ and $\mathrm{Trio}[X]$ have the additive structure of $\mathbb{N}^I_\mathrm{fin}$ for an appropriate index set $I$, and, likewise, $\mathbb{B}[X]$ and $\mathrm{Why}[X]$ have the additive structure of $\mathbb{B}^I_\mathrm{fin}$ again for appropriate index set $I$. Thus, the additive monoids of these four cases are covered by Proposition~\ref{prop:product}, which means that these monoids have the transportation  property. The additive structure of $\mathrm{Lin}[X]$ is somewhat peculiar, but it is not hard to check that if it is alternatively expanded with the intersection of sets for its multiplicative structure, viewing $\bot$ as a second copy of the empty set, then we get an additively absorptive and multiplicatively idempotent semiring, which is then covered by Proposition~\ref{prop:absorptiveandstaridempotent}. Similarly, $\mathrm{PosBool}[X]$ is covered in the same way and therefore the additive monoids of these two semirings also have the transportation property. Finally, we argued already that the Boolean semiring $\mathbb{B}$ has the transportation property, which completes all cases in the diagram of Figure~\ref{fig:diagram}.
\end{example}

\subsection{The Free Commutative Monoid} \label{sec:freemonoid}

For this section, recall the basic definitions of universal algebra
concerning homomorphisms, subalgebras, products and varieties of monoids
as they were presented in Section~\ref{sec:prelims}. An important result of 
universal algebra states that varieties 
have \emph{universal objects}, referred to as \emph{free algebras}. 
We state this in the special case of
monoids, but first we need two definitions. 

Let $\mathcal{C}$
be a class of monoids. Note that so far we do not require $\mathcal{C}$ to be a variety. 
Let $\mathbb{K}(X) = (K,+,0)$ be a monoid which is generated by a finite or infinite set 
$X \subseteq K$
of generators; this means that each $a \in K$ can be written in the form $t(a_1,\ldots,a_n)$ for some~$n \geq 0$ and~$a_1,\ldots,a_n \in X$, 
where $t(a_1,\ldots,a_n)$ denotes the result of evaluating an
expression formed by composing the constants~$0$ and~$a_1,\ldots,a_n$ with 
the binary operation $+$. We say that $\mathbb{K}(X)$ has the \emph{universal mapping property for~$\mathcal{C}$ 
over $X$} if for every~$\mathbb{M} = (M,+,0)$ in~$\mathcal{C}$ and every map $g : X \to M$
there is a homomorphism $h : K \to M$ which extends $g$ (see Definition~10.5 in \cite{DBLP:books/daglib/0067494}).

With these definitions, now we can state the result that we need from universal algebra. The general
theorem is due to Birkhoff and here we state only its specialization to varieties of monoids: 
For every finite or infinite set $X$ of \emph{indeterminates} (also called \emph{variables} or \emph{free generators}), 
and for every variety~$\mathcal{C}$ of monoids, there is a monoid $\mathbb{F}_{\mathcal{C}}(X)$ in $\mathcal{C}$ 
that is generated by $X$ and has the universal mapping property for $\mathcal{C}$ over $X$ 
(see Theorems~10.10 and~10.12 in~\cite{DBLP:books/daglib/0067494}). Furthermore, $\mathbb{F}_{\mathcal{C}}(X)$ 
is, up to isomorphism, the unique monoid $\mathbb{K}(Y)$ in $\mathcal{C}$ that is generated by a set $Y$ of
generators of cardinality~$|Y|=|X|$ and
has the universal mapping property for $\mathcal{C}$
over $Y$ (see Exercise~6 in Chapter~II.10 in~\cite{DBLP:books/daglib/0067494}).
Since we care only for commutative monoids, which form a variety of monoids, 
we write $\free{X}$ for $\mathbb{F}_{\mathcal{C}}(X)$, when
$\mathcal{C}$ is the variety of commutative monoids, and we refer to it as \emph{the free commutative
monoid generated by $X$}.

It turns out that, as we argue below, the free commutative monoid generated by $X$ has an explicit description:
it is precisely the monoid that we called $\mathbb{N}[X]_1$ in Section~\ref{sec:powers}, i.e., the monoid of linear forms on the indeterminates $X$ with non-negative integer coefficients. 
One consequence of this is that the free commutative monoid $\free{X}$ is always positive. Another consequence is that it has the \ftp.
A third consequence that is inherited from this 
is that any two $\free{X}$-relations that are inner consistent
have a sparse witness of consistency, when the set $X$ of generators is finite. We collect the first
two properties in the following proposition.

\begin{proposition} \label{prop:freeftp}
    For every set $X$ of indeterminates, the free commutative monoid generated by $X$ is isomorphic to $\mathbb{N}[X]_1$, i.e., $\free{X} \cong \mathbb{N}[X]_1$, and is a positive commutative monoid
    that has the \ftp. 
\end{proposition}

\begin{proof} 
Since $\mathbb{N}[X]_1$ is positive and has the transportation property by Example~\ref{ex:polynomials},
it suffices to show that $\free{X} \cong \mathbb{N}[X]_1$.
For this proof, let $\mathcal{C}$ denote the variety of commutative monoids, so that $\mathbb{F}_{\mathcal C}(X) = \free{X}$.
Since $\mathbb{N}[X]_1$ is generated by $X$, by the uniqueness of $\mathbb{F}_{\mathcal C}(X)$ 
it suffices to show that $\mathbb{N}[X]_1$ has the universal mapping property for $\mathcal{C}$ over $X$. Before we do this, let us recall from Example~\ref{ex:polynomials} that every element in $\mathbb{N}[X]_1$ is identified with a finite-support map $c : X \to Z^{\geq 0}$, with each indeterminate $x \in X$ being identified with the finite-support map $c_x : X  \to Z^{\geq 0}$ defined by $c_x(x) = 1$ and $c_x(y) = 0$ for all $y \in X  \setminus \{x\}$. 

Now, to prove the universal mapping property for $\mathbb{N}[X]_1$, fix a commutative monoid $\mathbb{M} = (M,+,0)$ and 
let $g : X \to M$ be any map. Define the required homomorphism $h$ as the \emph{evaluation map} 
\begin{equation}
    c \mapsto \sum_{\newatop{x \in X:}{c(x) \not= 0}} c(x) g(x), \label{eqn:summation}
\end{equation}
where $c$ is an element in $\mathbb{N}[X]_1$ identified with a 
finite-support map $c : X \to Z^{\geq 0}$. The external sum on the right-hand side of Equation~\eqref{eqn:summation} is 
in $\mathbb{M}$, and 
the notation $na$ for a positive integer $n$ and 
an element $a \in M$ stands for the sum $a + \cdots + a$ in $\mathbb{M}$ with
$n$ occurrences of $a$ in the sum if $n \geq 1$, and the neutral element $0$
of $M$ if $n = 0$. Note that the summation sign in~\eqref{eqn:summation} has finite extension because $c$ has finite support. 
Using the choice of $c_x$ for $x \in X$ defined above, 
it is straightforward to prove that $h$ is a homomorphism from $\mathbb{N}[X]_1$ to $\mathbb{M}$ that extends $g$.
\end{proof}

The additional claim we made that any two $\free{X}$-relations that are inner consistency have a sparse
witness of consistency when the set $X$ of generators is finite follows from combining
the fact that $\free{X} \cong \mathbb{N}[X]_1$ with the correspondence 
$\mathbb{N}[X]_1 \cong \mathbb{N}^X_\mathrm{fin}$ discussed in Example~\ref{ex:polynomials},
together with Example~\ref{ex:bagmonoidnw}, Proposition~\ref{prop:northwestcornerjoinmethod}, Proposition~\ref{prop:sparse}, 
and Proposition~\ref{prop:sparseproduct}.

\subsection{Some Important Non-Examples}

As we have seen, many important positive commutative monoids have the \ftp.
Unfortunately there are positive commutative monoids of different character that  fail to have the transportation property. Here we present  a few examples of
such monoids.

\begin{example} \textit{Natural numbers with addition truncated to $2$}.
Recall the positive commutative monoid $\mathbb{N}_2$ from Section~\ref{sec:not-sufficient}: the natural numbers $\{0,1,2\}$ with addition truncated to~$2$. In
that section we showed that 
the
path-of-length-$3$ hypergraph $P_3$
does not have the \ltgc~for ${\mathbb N}_2$-relations.
 From the implications (1) $ \Longrightarrow$ (3) and (2) $\Longrightarrow$ (3) in Theorem~\ref{thm:allequivalent}, it follows  that
$\mathbb{N}_2$ does not have the \ftp~ and, furthermore, the inner consistency property for ${\mathbb N}_2$-relations fails.
Here, we give a simple example showing that the inner consistency property for ${\mathbb N}_2$-relations fails. Combined with Theorem~\ref{thm:allequivalent}, this gives a different proof that ${\mathbb N}_2$ does not have the transportation property.

Let $R(AC)$ and $S(BC)$ be the
$\mathbb {N}_2$-relations given by
$R(a_1,c)=R(a_2,c)=S(b_1,c)=1$ and
$S(b_2,c)=2$, and no other tuples in their support.
These two $\mathbb {N}_2$-relations are inner consistent because $R(c) = S(c) = 2$. However, they are
not consistent. To prove this and towards a contradiction, assume  that $W(ABC)$ is an $\mathbb{N}_2$-relation
such that $W[AB] = R$ and $W[BC] = S$. Let us say $W(a_i,b_j,c) = x_{ij}$ for $i = 1,2,3$ and $j = 1,2$, where
each $x_{ij}$ is a value in $\{0,1,2\}$. This assumption gives rise to a system of five equations:
  \begin{center}
  \begin{tabular}{lllll}
  $x_{i1} \oplus x_{i2}$ & $= 1$ & \;\;\; & for $i= 1, 2,3$ \\
  $x_{1j} \oplus x_{2j} \oplus x_{3j}$ & $= 2$ & & for $j = 1,2$.
  \end{tabular}
  \end{center}
We reach a contradiction by double-counting the number of $x_{ij}$'s that are assigned the value~$1$.
The first type of equation implies that, for all $i=1,2,3$, either $x_{i1} = 0$ and $x_{i2} = 1$, or $x_{i1} = 1$ and $x_{i2} = 0$. In particular, exactly three among all $x_{ij}$ with $i=1,2,3$ and $j=1,2$ are assigned the value $1$ and the rest are assigned the value $0$. 
Therefore, for at least one among $j=1,2$ there is at most one among $i=1,2,3$ such that $x_{ij}$ is assigned the value $1$ and the rest are assigned 
the value $0$, which is against the second type of equation for this $j$.
\end{example}

\begin{example} 
\textit{Non-negative real numbers with addition and a gap}.
Let~$\mathbb{R}_1 = (\{0\} \cup [1,+\infty), +,  0)$ be 
the structure with  $0$ and  all real numbers bigger or equal than~$1$ as its universe, and with the
standard addition as its operation. It is obvious that  $\mathbb{R}_1$ is a positive commutative monoid. We show that  the \icp~for $\mathbb{R}_1$~fails,  hence $\mathbb{R}_1$ does not have the \ftp.

Let~$R(AC)$ and~$S(BC)$ be the~$\mathbb{R}_1$-relations given
  by~$R(a_i,c) = 1$ for~$i= 1,2,3$ and~$S(b_j,c) = 1.5$
  for~$j= 1, 2$, and no other tuples in their supports.
  These two~$\mathbb{R}_1$-relations are inner consistent,
  since~$R(c) = S(c) = 3$. We claim that they are not
  consistent. Indeed, assume that there is
  an~$\mathbb{R}_1$-relation~$W(ABC)$ witnessing the consistency
  of~$R(AB)$ and~$S(BC)$. Let us say that~$W(a_i,b_j,c) = x_{ij}$
  for~$i=1, 2,3$ and~$j= 1,2$, where each~$x_{ij}$ is a value
  in~$\{0\} \cup [1,+\infty)$. This assumption gives rise
  to a system of five equations:
  \begin{center}
  \begin{tabular}{llll}
  $x_{i1} + x_{i2}$ & $= 1$ & \;\;\; & for $i= 1, 2,3$ \\
  $x_{1j} + x_{2j} + x_{3j}$ & $= 1.5$ & & for $j = 1,2$.
  \end{tabular}
  \end{center}
  The first type of equation with~$i = 1$ implies that
  either~$x_{11} = 0$ and~$x_{12} = 1$, or that~$x_{11} = 1$
  and~$x_{12} = 0$. If~$x_{11} = 0$, then the second type of
  equation with~$j = 1$ implies that~$x_{21} + x_{31} = 0.5$, which
  is impossible. If~$x_{12} = 0$, then the second type of equation
  with~$j = 2$ implies that~$x_{22} + x_{32} = 0.5$, which is
  impossible. Since the system has no solution in~$\mathbb{R}_1$, we
  conclude that the relations~$R(AC)$ and~$S(BC)$ are not
  consistent. Note that in this proof we used only three of the six
  equations. However, the other two are forced by the inner
  consistency condition (i.e., there is no choice but to
  have~$x_{21} + x_{22} = 1$ and~$x_{31} + x_{32} = 1$).  
\end{example}

\begin{example} \textit{Truncated powersets}.
For each natural number $k$, let $\mathbb{P}_k = (\{0,\ldots,k+1\},+,0)$ be
the monoid with neutral element $0$, absorbing element $k+1$, and
such that $i+i=i$ for all $i \in [k]$, and $i+j = k+1$ for all $i,j \in [k]$ with $i \not= j$.
An alternative presentation of $\mathbb{P}_k$ is as the substructure of the
powerset monoid $\mathbb{P}([k+1]) = (\mathcal{P}([k+1]),\cup,\emptyset)$ induced
by the empty set $\emptyset$, the full set $[k+1]$, and the $(k-1)$-element subsets $[k]\setminus\{i\}$ for
$i = 1,\ldots,k$. This explains  the name \emph{truncated powersets}.
For example, this alternative presentation of
$\mathbb{P}_3$ is the structure $(\{\emptyset,\{1,2\},\{1,3\},\{2,3\},\{1,2,3\}\},\cup,\emptyset)$.

Clearly each $\mathbb{P}_k$ is positive and commutative. We show that $\mathbb{P}_k$
does not have the transportation property unless $k = 0$ or $k = 1$ or $k = 2$. For $k = 0$
we have that $\mathbb{P}_k$ is isomorphic to Boolean monoid $\mathbb{B} 
= (\{0,1\},\vee,0)$. For $k = 1$ we
have that $\mathbb{P}_k$ is isomorphic to $(\{0,1,2\},\max,0)$. For $k = 2$ we
have that $\mathbb{P}_k$ is isomorphic to $(\mathcal{P}(\{1,2\}),\cup,\emptyset)$.
These three cases are covered by the lattice case in Example~\ref{ex:lattices} and
have then the transportation property. For $k \geq 3$ we show that $\mathbb{P}_k$ does 
not have the transportation property.

Let $k \geq 3$, and let $R(AC)$ and $S(BC)$ be the $\mathbb{P}_k$-relations with $R(a_1,c) = 1$ and $R(a_2,c) = 3$
and $S(b_1,c) = 2$ and $S(b_2,c) = 3$, and no other tuples in their supports. These are
inner consistent since, in the structure $\mathbb{P}_k$ with $k \geq 3$, we have $R[C](c) = 1+3 = k+1 = 2+3 = S[C](c)$. 
We show that $R$ and $S$ are not consistent.
Indeed, assume that there is a~$\mathbb{P}_k$-relation~$W(ABC)$ witnessing the consistency
  of~$R(AB)$ and~$S(BC)$. Let us say that~$W(a_i,b_j,c) = x_{ij}$
  for~$i=1,2$ and~$j= 1,2$, where each~$x_{ij}$ is a value
  in~$\{0,\ldots,k+1\}$. This assumption gives rise
  to a system of four equations:
  \begin{center}
  \begin{tabular}{llll}
  $x_{11} + x_{12} = 1$ \\
  $x_{21} + x_{22} = 3$ \\
  $x_{11} + x_{21} = 2$ \\
  $x_{12} + x_{22} = 3$.
  \end{tabular}
  \end{center}
The first equation interpreted in $\mathbb{P}_k$
implies that $x_{11} = 1$ or $x_{12} = 1$.
If $x_{11} = 1$, then the third equation cannot be satisfied since there
is no $j$ such that $1+j=2$ in $\mathbb{P}_k$, while if $x_{12} = 1$,
then the fourth equation cannot be satisfied since there is no $j$ such
that $1+j = 3$ in $\mathbb{P}_k$.
\end{example}

Our last example involves  a natural positive commutative monoid for which the failure
of the transportation property is conceptually significant as it corresponds
to the deep fact of quantum mechanics that there exist pairs of 
binary observables that cannot be jointly measured. This is a manifestation of the
celebrated \emph{Heisenberg uncertainty principle} for 
positive-operator-valued measures~\cite{MiyaderaImai2008}; we do not elaborate on
this here and refer the interested reader to the introduction of the cited article
for an extensive survey of related literature.

\begin{example} \textit{Positive semidefinite matrices under addition}.
Let $n \geq 1$ be a positive integer and let $\mathbb{PSD}_n$ be the set of positive semidefinite matrices
in $\mathbb{R}^{n \times n}$, i.e., the $n\times n$ symmetric 
real matrices $A$ for which $z^\transpose A z \geq 0$ holds for all $z \in \mathbb{R}^n$. Equivalently, $A$ is positive semidefinite if and only if 
it is symmetric and all its eigenvalues are non-negative. 
This is a commutative monoid under componentwise addition; commutativity is obvious and
the sum of positive semidefinite matrices is positive semidefinite since $z^\transpose (A+B)z = z^\transpose A z + z^\transpose B z \geq 0$
for all $z \in \mathbb{R}^n$, where the inequality follows from the positive semidefiniteness of $A$ and $B$.
The monoid is also positive. To see this, first note that
if $A + B = 0$, then $z^\transpose A z + z^\transpose B z = z^\transpose(A + B)z = 0$, so $z^\transpose A z = z^\transpose B z = 0$ 
for all vectors $z \in \mathbb{R}^n$ by the positive semidefiniteness of $A$ and $B$. By applying this to the standard basis vectors $e_i = (0,\ldots,0,1,0,\ldots,0)\in \mathbb{R}^n$ with $i=1,\ldots,n$ we see that the diagonals of $A$ and $B$ vanish, so the traces of $A$ and $B$ vanish, which means that the sums of their eigenvalues vanish, so all their eigenvalues vanish since positive semidefinite matrices have non-negative eigenvalues. From this
it follows that $A$ and $B$ are the zero matrix by 
considering their spectral decompositions $A = PDP^\transpose$ and $B = QEQ^\transpose$, where $D$ and $E$ are the diagonal matrices that collect their eigenvalues.

Next we show that $\mathbb{PSD}_n$ does not have the \ftp, provided $n > 1$. 
For $n = 1$, we have that $\mathbb{PSD}_n$ is isomorphic to the monoid  $\mathbb{R}^{\geq 0}$
of the non-negative reals with addition, and this has been shown to have the \ftp~in Example~\ref{ex:bagmonoid}. Next we argue that $\mathbb{PSD}_2$ does not have the \ftp. From this, the claim follows for
$\mathbb{PSD}_n$ with $n > 2$ by padding the matrices with zeros. Our proof for $n = 2$
is an adaptation of a more general statement that can be found in~\cite{KunjwalHeunenFritz2014}.

Consider the classical Pauli matrices:
\begin{equation*}
    X = \left({\begin{array}{cc} 0 & 1 \\ 1 & 0 \end{array}}\right) \;\;\;\;\; 
    Y = \left({\begin{array}{cc} 0 & -i \\ i & 0 \end{array}}\right) \;\;\;\;\;
    Z = \left({\begin{array}{cc} 1 & 0 \\ 0 & -1 \end{array}}\right).
\end{equation*}
Observe that $Y$ has complex entries, but $X$ and $Z$ are $2\times 2$ real matrices. Consider the
instance of the transportation problem given by the four matrices 
\begin{equation*}
\begin{array}{ccccccc}
    B_1 & = & (\Id+X)/2 & \;\;\;\; & B_2 & = & (\Id-X)/2 \\
    C_1 & = & (\Id+Z)/2 & &  C_2 & = & (\Id-Z)/2,
\end{array}
\end{equation*}
where $\Id$ is the $2\times 2$ identity matrix. These are positive semidefinite matrices since their eigenvalues are in $\{0,1\}$,
and the vectors $(B_1,B_2)$ and $(C_1,C_2)$ form a balanced instance of the transportation problem since $B_1 + B_2 = C_1 + C_2 = \Id$. This gives rise to a system
of four matrix equations
\begin{align*}
    X_{11} + X_{12} = B_1 \\
    X_{21} + X_{22} = B_2 \\
    X_{11} + X_{21} = C_1 \\
    X_{12} + X_{22} = C_2 
\end{align*}
We claim that this system is infeasible in $2 \times 2$ positive semidefinite matrices. Suppose otherwise. 
Left-multiply the first equation by $X$, the second equation by $-X$, the third equation by $Z$, the fourth equation by $-Z$,
and add up everything. Using the fact that $X^2 = Z^2 = \Id$ and hence $XB_1 - XB_2 = ZC_1 - ZC_2 = \Id$, 
this gives the identity
\begin{equation}
A_{11} X_{11} + A_{12} X_{12} + A_{21} X_{21} + A_{22} X_{22} = 2\Id \label{eqn:midentity}
\end{equation}
where
\begin{equation*}
\begin{array}{ccccccccc}
    A_{11} & = & X+Z & \;\;\;\; & A_{12} & = & X-Z \\
    A_{21} & = &-X+Z & &  A_{22} & = & -X-Z.
\end{array}
\end{equation*}
The trace of the matrix on the right-hand side in~\eqref{eqn:midentity}
is~$4$. In contrast, by Hölder's inequality for the Schatten 
norm with $p = 1$ and $q = \infty$,
the trace of the matrix on the left-hand side in \eqref{eqn:midentity} is bounded by
\begin{equation}
    \norm{A_{11}}\trace(X_{11}) + \norm{A_{12}}\trace(X_{12}) + \norm{A_{21}}\trace(X_{21}) + \norm{A_{22}}\trace(X_{22}), \label{eqn:bound}
\end{equation}
where
$\norm{A}$ denotes the spectral norm of $A$, i.e., the largest eigenvalue of the matrix $A$, in absolute value.
It can be checked by direct computation that each of the matrices $A_{ij}$ 
has eigenvalues $\pm \sqrt{2}$,
so their spectral norm is $\sqrt{2}$. Furthermore, each $X_{ij}$ is a positive semidefinite matrix
by assumption; hence 
 its trace, which is the sum of the eigenvalues, 
which are non-negative for positive semidefinite matrices, is non-negative. It follows that~\eqref{eqn:bound} is bounded by
\begin{equation}
    (\trace(X_{11}) + \trace(X_{12}) + \trace(X_{21}) + \trace(X_{22})) \sqrt{2} = 
    \trace(X_{11} + X_{12} + X_{21} + X_{22}) \sqrt{2} = 2\sqrt 2,
\end{equation}
where the first equality follows from the linearity of the trace, and the second follows from the fact that the sum of the $X_{ij}$ 
is $B_1 + B_2 = C_1 + C_2 = \Id$, which has trace~$2$. 
The conclusion is that the trace of 
the left-hand side in \eqref{eqn:midentity} is at most $2\sqrt{2} < 4$,
which is against the fact that the trace of the right-hand side in \eqref{eqn:midentity} is~$4$.
\end{example}

\section{Local Consistency up to a Cover} \label{sec:uptocovers}

In the previous sections, we characterized the class of
positive commutative
monoids $\mathbb{K}$ for which the standard local consistency of $\mathbb{K}$-relations 
agrees with their global consistency for precisely the acyclic hypergraphs. 
The goal of this section is to investigate whether there 
is a suitably modified
notion of local consistency of $\mathbb{K}$-relations
that has the same effect of capturing the global consistency of $\mathbb{K}$-relations for
precisely
the acyclic hypergraphs, but that applies to
\emph{every} positive commutative monoid.

We achieve this by strengthening the requirement
of locality: in addition to requiring that the relations are pairwise consistent as $\mathbb{K}$-relations,
we will also require that they are pairwise consistent when they are appropriately viewed as $\free{X}$-relations, where $\free{X}$ 
is the free commutative monoid with a large
enough set $X$ of generators. 
We refer to this new notion of local consistency of $\mathbb{K}$-relations as \emph{pairwise consistency up to the free cover} of $\mathbb{K}$. 
Surprisingly, we show that this abstract notion of local
consistency of $\mathbb{K}$-relations characterizes global consistency of $\mathbb{K}$-relations
for precisely  
the acyclic hypergraphs, and for \emph{every} positive commutative monoid~$\mathbb{K}$.

\subsection{Consistency up to a Cover}

Let~$\mathbb{K}$ be a positive commutative monoid.  A \emph{cover
of~$\mathbb{K}$} is a positive commutative monoid~$\mathbb{K}^*$ such that there is a surjective homomorphism $h$ from ${\mathbb K}^*$  
onto~$\mathbb{K}$. 
The \emph{identity cover} is the
cover where $\mathbb{K}^*$ is~$\mathbb{K}$ itself and $h$ is the identity map.
A cover of $\mathbb{K}$ is given by the pair $(\mathbb{K}^*,h)$ of both objects;  we use the notation~$h : \mathbb{K}^* \onto \mathbb{K}$ to say 
that the pair $(\mathbb{K}^*,h)$ is a cover
of~$\mathbb{K}$. 
For the definitions of the next paragraph, fix such a cover.

For a~$\mathbb{K}$-relation~$R(Y)$, an \emph{$h$-lift of~$R$} is
a~$\mathbb{K}^*$-relation~$R^*(Y)$ such that $h(R^*(t)) = R(t)$ holds
for every $Y$-tuple $t$, i.e.,~$h \circ R^* = R$
holds. In most of the cases that follow, the cover
will be clear from the context, and we simply say that~$R^*$ is a
lift of~$R$, without any reference to~$h$. Note that, since the
homomorphism~$h$ is surjective onto~$\mathbb{K}$,
every~$\mathbb{K}$-relation~$R$ has at least
one~$h$-lift~$R^*$. Consider the special case where $h : \mathbb{K}^* \onto \mathbb{K}$
is a \emph{retraction}, meaning that $K \subseteq K^*$ and $h$ is the identity on $K$, 
where $K$ and $K^*$ are  the universes of $\mathbb{K}$ and $\mathbb{K}^*$, respectively; in this case,  
the \emph{direct $h$-lift of $R$} is the $\mathbb{K^*}$-relation $R^*$ defined
by $R^*(t) = R(t)$, for every $Y$-tuple $t$.

\begin{definition}\label{defn:consistencycovers}
Let $\mathbb{K}$ be a positive commutative monoid, 
let~$h:\mathbb{K}^* \onto \mathbb{K}$ be a cover of~$\mathbb{K}$, 
let~$X_1,\ldots,X_m$ be a schema, 
let~$R_1(X_1),\ldots,R_m(X_m)$ be a collection
of~$\mathbb{K}$-relations over the schema~$X_1,\ldots,X_m$, 
and let~$k$ be a positive integer.
We say that the collection $R_1,\ldots,R_m$ is~\emph{$k$-wise consistent up to
the cover $h : \mathbb{K}^* \onto \mathbb{K}$} if
there exists a collection $R^*_1,\ldots,R^*_m$ of~$h$-lifts of
$R_1,\ldots,R_m$ that is $k$-wise consistent (as a collection of $\mathbb{K}^*$-relations).
If~$k = 2$, then we
say that the collection~$R_1,\ldots,R_m$ is \emph{pairwise consistent up to the cover}. 
If~$k = m$,
then we say that the collection~$R_1,\ldots,R_m$ is \emph{globally consistent up to the cover}.
When $k = m = 2$ we just say that $R_1$ and $R_2$ are \emph{consistent up to the cover}.
\end{definition}

Before we go on, it is important to point out
that in the definition of consistency up to a cover,
not only the choice of the cover $h : \mathbb{K}^* \onto \mathbb{K}$
potentially matters, but also the choice of $h$-lifts $R_1^*,\ldots,R_m^*$ matters.
We illustrate this with an example.

\begin{example}
    Consider the $\mathbb{N}_2$-relations $R_2(BC)$ and $R_3(CD)$ in 
    Proposition~\ref{prop:path3}. Consider the cover $h : \free{x,y}
    \onto \mathbb{N}_2$ given by the canonical homomorphism $h$ from the free commutative
    monoid $\free{x,y}$ with two generators
    $x$ and $y$ for the non-zero elements $1$ and $2$ of $\mathbb{N}_2$, which is of course a 
    surjective homomorphism.    
    As shown in Proposition~\ref{prop:path3}, the $\mathbb{N}_2$-relations $R_2(BC)$ and $R_3(CD)$ are 
    consistent as $\mathbb{N}_2$-relations. 
    However, when viewed as $\free{x,y}$-relations 
    $R_2^*$ and $R_3^*$ through the direct $h$-lift with the retraction that
    identifies $x$ with $1$ and $y$ with $2$,
    the two $\free{x,y}$-relations $R^*_2$ and $R^*_3$ are $h$-lifts of $R_2$ and $R_3$ that
    are not consistent because they are not even 
    inner consistent, since  we
    have that $R^*_2[C](c_1) = y \not= x+x+x = R^*_3[C](c_1)$. Nonetheless,
    if we take the $\mathbb{N}_2$-relation $R_{23}(BCD)$ that witnesses the consistency
    of $R_2$ and $R_3$ as $\mathbb{N}_2$-relations, 
    then we can view $R_{23}$ as an $\free{x,y}$-relation $W^*$ that is an $h$-lift
    of $R_{23}$, and we can now take $R_2^* := W^*[BC]$ and $R_3^*= W^*[CD]$, and these 
    are obviously both consistent
    $\free{x, y}$-relations 
    and $h$-lifts of $R_2$ and $R_3$, though not direct $h$-lifts.
\end{example}

\paragraph{Global Consistency up to Covers and its Absoluteness} 
The first technical result of this section is the following simple but
important observation stating that, as regards to  global consistency, the choice of the cover does
not really matter. While this independence of the cover will not be shared by 
the notion of pairwise consistency up to a cover that we will introduce later on, 
the fact that it holds for global consistency is key for our purposes.

\begin{proposition}[Absoluteness of Global Consistency] \label{prop:absoluteness}
Let~$\mathbb{K}$ be a positive commutative monoid and
let~$R_1,\ldots,R_m$ be a collection
of~$\mathbb{K}$-relations. The following statements are equivalent:
\begin{enumerate} \itemsep=0pt
\item the collection $R_1,\ldots,R_m$ is globally consistent,
\item the collection $R_1,\ldots,R_m$ is globally consistent 
up to every cover of $\mathbb{K}$,
\item the collection $R_1,\ldots,R_m$ is globally consistent 
up to some cover of $\mathbb{K}$.
\end{enumerate}
\end{proposition}

\begin{proof}
Let~$Y_i$ be the set of attributes of  $R_i$ for $i = 1,\ldots,m$. 

(1) $\Longrightarrow$ (2): Let~$W$ be a~$\mathbb{K}$-relation such that~$W[Y_i]
= R_i$ holds for all~$i \in [m]$.  Fix an arbitrary cover~$h :
\mathbb{K}^* \onto \mathbb{K}$ and let~$W^*$ be any $h$-lift of~$W$. 
Such a lift exists because~$h$ is surjective onto $\mathbb{K}$. For
each~$i \in [m]$, choose~$R_i^* := W^*[Y_i]$.  We claim
that~$R_1^*,\ldots,R_m^*$ are lifts of~$R_1,\ldots,R_m$, and also
that~$W^*$ witnesses their global consistency. Indeed, for each~$i \in
[m]$ and each~$Y_i$-tuple~$t$ we have
\begin{align}
h(R^*_i(t)) & = h(W^*[Y_i](t)) = h\Big({\sum_{r : r[Y_i] = t} W^*(r)}\Big) = 
\label{eqn:hofrstarone} \\
& = {\sum_{r : r[Y_i] = t} h(W^*(r))} 
= {\sum_{r : r[Y_i] = t} W(r)}  = W[Y_i](t) = R_i(t), \label{eqn:hofrstartwo}
\end{align}
where the first equality follows from the choice of~$R^*_i$, the
second follows from the definition of marginal, the third follows from
the fact that~$h$ is a homomorphism, the fourth follows from the fact
that~$W^*$ is a lift of~$W$, the fifth follows from the definition of
marginal, and the sixth follows from~$W[Y_i] = R_i$. This shows
that~$h \circ R_i^* = R_i$, so~$R_1^*,\ldots,R_m^*$ are lifts
of~$R_1,\ldots,R_m$. Finally, the fact that~$W^*$ witnesses the global
consistency of~$R_1^*,\ldots,R_m^*$ is obvious by construction.

(2) $\Longrightarrow$ (3): This is obvious by choosing the identity cover.

(3) $\Longrightarrow$ (1): Let~$h : \mathbb{K}^* \onto \mathbb{K}$ be a cover
up to which the collection~$R_1,\ldots,R_m$ is globally
consistent. Let then~$R_1^*,\ldots,R_m^*$ be a collection of lifts
of~$R_1,\ldots,R_m$ that is globally consistent. Let~$W^*$ be
the~$\mathbb{K}^*$-relation that witnesses its global consistency and
define~$W := h \circ W^*$. We claim that~$W$ witnesses the global
consistency of~$R_1,\ldots,R_m$. Indeed, for each $i \in [m]$
and each~$Y_i$-tuple~$t$ it
holds that
\begin{align}
W[Y_i](t) & = {\sum_{r:r[Y_i]=t} W(r)} =
{\sum_{r:r[Y_i]=t} h(W^*(r))} = \label{eqn:broken1} \\
& =
h\Big({\sum_{r:r[Y_i]=t} W^*(r)}\Big) = h(W^*[Y_i](t)) = h(R^*_i(t)) =
R_i(t), \label{eqn:broken2}
\end{align}
where the first equality follows from the definition of marginal, the
second follows from the choice of~$W$, the third follows from the fact
that $h$ is a homomorphism, the fourth follows from the definition of
marginal, the fifth follows from the fact that~$W^*[Y_i] = R^*_i$, and
the sixth follows from the fact that~$R^*_i$ is an $h$-lift
of~$R_i$. This shows that~$W[Y_i] = R_i$, hence the
collection~$R_1,\ldots,R_m$ is globally consistent.
\end{proof}

In view of Proposition~\ref{prop:absoluteness}, we say that the notion
of global consistency up to covers is \emph{absolute} as if it holds
for some cover, then it holds for all covers. Next we localize this
notion. Unlike the global notion, the local notion will
\emph{not} be absolute in the sense that a collection
of~$\mathbb{K}$-relations may be locally consistent up to some cover
but not up to every cover. 

\paragraph{Local Consistency up to Covers} We show that, up to covers, two thirds of
Proposition~\ref{prop:absoluteness} descend from
global consistency to local consistency. Concretely,
we show in Proposition~\ref{prop:descends} below that a collection is $k$-wise consistent in the standard
sense if and only if it is $k$-wise consistent up to some cover of $\mathbb{K}$. 
In contrast, we also show in Example~\ref{ex:notabsolute} below that a third statement
quantifying over all covers of $\mathbb{K}$ would not be equivalent. This state
of affairs notwithstanding, 
two additional refined notions of local consistency up to a cover
make sense and those are indeed
equivalent to the one we defined. While these refined notions will not play a role
in later sections, we spell them out next to clarify the choices 
that were involved in the original definition of local consistency up to a cover.

We say that the collection $R_1,\ldots,R_m$ is \emph{weakly $k$-wise consistent
up to the cover $h : \mathbb{K}^* \onto \mathbb{K}$} if for every $t \in [k]$,
every $i_1,\ldots,i_t \in [m]$, and every $j \in [t]$,
there exists an $h$-lift $R_{i_j}^*$ of $R_{i_j}$ such that the collection
$R_{i_1}^*,\ldots,R_{i_t}^*$ is globally consistent. Finally,
we say that the collection $R_1,\ldots,R_m$ is \emph{very weakly $k$-wise consistent
up to some covers of $\mathbb{K}$} if for every $t \in [k]$
and
every $i_1,\ldots,i_t \in [m]$ 
there exists a cover $h : \mathbb{K}^*_t \onto \mathbb{K}$ such
that for and every $j \in [t]$, there exist an 
$h$-lift $R^*_{i_j}$ of $R_{i_j}$ such that the collection~$R_{i_1}^*,\ldots,R_{i_t}^*$
is globally consistent. Note the difference with the earlier definition:
in the \emph{weak} case, the choices of lifts for each $R_i$ may depend on the subcollection, and
in the \emph{very weak} case even the cover up to which
consistency is defined may depend on the subcollection.

\begin{proposition} \label{prop:descends}
Let~$\mathbb{K}$ be a positive commutative monoid,
let~$R_1,\ldots,R_m$ be a collection of~$\mathbb{K}$-relations, and
let~$k$ be a positive integer. The following statements are
equivalent:
\begin{enumerate} \itemsep=0pt
\item[(1)] the collection $R_1,\ldots,R_m$ is $k$-wise consistent,
\item[(2)] the collection $R_1,\ldots,R_m$ is $k$-wise consistent
up to some cover of $\mathbb{K}$,
\item[(3)] the collection $R_1,\ldots,R_m$ is weakly $k$-wise consistent
up to some cover of $\mathbb{K}$,
\item[(4)] the collection $R_1,\ldots,R_m$ is very weakly $k$-wise consistent
up to some covers of~$\mathbb{K}$.
\end{enumerate}
\end{proposition}

\begin{proof} 
Let $Y_i$ be the set of attributes in $R_i$ for $i = 1,\ldots,m$.

(1) $\Longrightarrow$~(2): This is obvious by choosing the identity cover.  

(2) $\Longrightarrow$~(3): This is obvious by choosing the same lifts.

(3) $\Longrightarrow$~(4): This is obvious by choosing the same cover and
the same lifts.

(4) $\Longrightarrow$~(1): Fix $t \in [k]$ and $i_1,\ldots,i_t \in [m]$,
and let $h : \mathbb{K}^* \onto \mathbb{K}$ and 
$R^*_{i_1},\ldots,R^*_{i_t}$ be as given by the definition of very weakly $k$-wise consistency 
up to some covers for this $t$ and these $i_1,\ldots,i_t$.
In particular the collection $R^*_{i_1},\ldots,R^*_{i_t}$ is globally consistent.
Therefore, the subcollection $R_{i_1},\ldots,R_{i_t}$ of the original $\mathbb{K}$-relations
is globally consistent up to  some cover, i.e., namely
$h : \mathbb{K}^* \to \mathbb{K}$, 
and hence globally consistent by Proposition~\ref{prop:absoluteness}.
\end{proof}

It should be pointed out that the equivalence of the items
in Proposition~\ref{prop:descends} would not go through if the same cover of $\mathbb{K}$
were fixed at the outset for all items. This will follow from the fact that,
as we argue below, absoluteness fails
for local consistency. For the main result of this section,
what really matters from Proposition~\ref{prop:descends} is
the equivalence between items~(1) and~(2), which states that local consistency up to a cover
is a conservative generalization of the classical notion of local consistency.

Finally we give the promised example showing that, in general,~$k$-wise consistency up 
to a cover is not absolute in the sense of
Proposition~\ref{prop:absoluteness}. Concretely, the example will show
that for the positive commutative monoid $\mathbb{N}_2$ in 
Proposition~\ref{prop:path3} and for the values~$k = 2$ and $m=3$, 
one cannot
add   to Proposition~\ref{prop:descends} 
a condition analogous to the second condition 
  in Proposition~\ref{prop:absoluteness} stating that
the collection~$R_1,\ldots,R_m$ is~$k$-wise consistent up to every
cover of~$\mathbb{N}_2$. In other words, there are collections
of~$\mathbb{N}_2$-relations that are pairwise consistent but are not pairwise
consistent up to every cover of~$\mathbb{N}_2$.

\begin{example} \label{ex:notabsolute}
Consider the
collection~$R(AB),S(BC),T(CD)$ of the three~$\mathbb{N}_2$-relations from
Proposition~\ref{prop:path3}.  These relations are pairwise consistent but are not globally consistent
as~$\mathbb{N}_2$-relations. Consider the cover $h : \mathbb{N} \onto \mathbb{N}_2$, where 
$\mathbb{N}$ is the bag monoid and $h$
maps~$n$ to~$n$ if $n = 0$
or $n = 1$, and maps $n$ to $2$ if~$n \geq 2$, i.e., $h$ truncates addition to $2$. 
We claim that the collection~$R,S,T$ cannot be lifted to a
collection of pairwise
consistent~$\mathbb{N}$-relations~$R^*,S^*,T^*$.  For, if they
could, then~$R^*,S^*,T^*$ would be a collection of pairwise consistent
bags, hence they would also be globally consistent by the local-to-global consistency
property for bags on acyclic schemas, since  the schema~$AB,BC,CD$
is acyclic. But then~$R(AB),S(BC),T(CD)$ would be also globally
consistent as~$\mathbb{N}_2$-relations by truncating to~$2$ every
natural number bigger than~$2$ in the bag~$W^*$ that witnesses the
global consistency of~$R^*,S^*,T^*$. This contradicts 
 Proposition~\ref{prop:path3} and completes the example.
\end{example}

\subsection{Local-to-Global Consistency up to Covers}

The \ltgc\ up to a cover is
defined to generalize Definition~\ref{defn:local-to-global} as follows:

\begin{definition}
Let~$\mathbb{K}$ be a positive commutative monoid,
let~$h : \mathbb{K}^* \onto \mathbb{K}$ be a cover of~$\mathbb{K}$,
and let $X_1,\ldots,X_m$ be a listing of all the hyperedges of a hypergraph $H$.
We say that~$H$ has the \emph{\ltgc\
for~$\mathbb{K}$-relations up to the cover~$h:\mathbb{K}^* \onto \mathbb{K}$} if every collection
$R_1(X_1),\ldots,R(X_m)$ of~$\mathbb{K}$-relations that is pairwise
consistent up to the cover is globally consistent.
\end{definition}

Recall from Section~\ref{sec:freemonoid} the definition of the free commutative monoid $\free{X}$ for a finite
or finite set of indeterminates $X$.
In the statement of the following theorem, 
let $\free{K^+}$ denote the free commutative monoid generated by the
set $K^+$ of non-zero elements in~$K$ seen
as indeterminates. 
Note that $\free{K^+}$ is positive by Proposition~\ref{prop:freeftp}.
The \emph{free cover
of~$\mathbb{K}$} refers to the cover 
$h : \free{K^+} \onto \mathbb{K}$
provided by the homomorphism $h$ from~$\free{K^+}$ 
to~$\mathbb{K}$ given
by the universal mapping property of $\free{K^+}$ applied to the identity
map $g : K^+ \to K^+$ defined by $g(x) = x$ for all $x \in K^+$. Clearly, $h$ is surjective onto $K$
as it extends $g$ and any homomorphism between monoids maps the neutral element of the first monoid
to the neural element of the second. Hence, $h : \free{K^+} \onto \mathbb{K}$
is indeed a cover.

\begin{theorem} \label{thm:mainuptocovers}
Let~$\mathbb{K}$ be a positive commutative monoid and let~$H$ be a
hypergraph. Then, the following statements are equivalent:
\begin{enumerate} \itemsep=0pt
\item $H$ is acyclic,
\item $H$ has the local-to-global consistency property up to 
  the free cover of $\mathbb{K}$,
\item $H$ has the local-to-global consistency property up to 
  some cover of $\mathbb{K}$.
\end{enumerate}
\end{theorem}

\begin{proof}
Let~$Y_1,\ldots,Y_m$ be a listing of the hyperedges of~$H$.  

(1) $\Longrightarrow$ (2).
We need to show that if~$H$ is acyclic, then pairwise
consistency up to the free cover of~$\mathbb{K}$ is a sufficient
condition for global consistency. This proof uses as a black box the
previously established fact that, for any
non-empty set~$X$ of indeterminates, the free commutative
monoid~$\free{X}$ has the \ftp, hence every acyclic hypergraph
has the (standard) local-to-global consistency property
for~$\free{X}$-relations - see Proposition~\ref{prop:freeftp} in
Section~\ref{sec:freemonoid}, and Theorem~\ref{thm:allequivalent} in Section~\ref{sec:sufficient}.

Let~$R_1(Y_1),\ldots,R_m(Y_m)$ be a collection
of~$\mathbb{K}$-relations and assume that it is pairwise consistent
up to the free cover~$h : \free{K^+} \onto \mathbb{K}$. Accordingly, 
let~$R^*_1,\ldots,R^*_m$ be a collection
of~$\free{K^+}$-relations that are~$h$-lifts of~$R_1,\ldots,R_m$,
respectively, and assume that the collection~$R^*_1,\ldots,R^*_m$ is
pairwise consistent. Since $H$ is acyclic, it has the~\ltgc~ for $\free{K^+}$-relations,
so the
collection~$R^*_1,\ldots,R^*_m$ of~$\free{K^+}$-relations is globally
consistent as~$\free{K^+}$-relations. But, then, the
collection~$R_1,\ldots,R_m$ of~$\mathbb{K}$-relations itself is
globally consistent up to the free cover of~$\mathbb{K}$, so it is
globally consistent by the absoluteness property stated in
Proposition~\ref{prop:absoluteness}.

(2) $\Longrightarrow$ (3).
This is obvious because the free cover of $\mathbb{K}$ is a cover of $\mathbb{K}$.

(3) $\Longrightarrow$ (1).
First we adapt the proof of Lemma~\ref{lem:newtseitin}
to show that there is no cover up to which the minimal
non-acyclic hypergraphs~$C_n$ and~$H_n$ with $n \geq 3$ have the
local-to-global consistency property.
As in Lemma~\ref{lem:newtseitin}, we prove this more generally for any non-trivial
uniform and regular hypergraph in Lemma~\ref{lem:acyclicnecessaryforminimalcover} below. After this is proved, we show
that the reduction that transfers the \ltgc~from any non-acyclic
hypergraph to the minimal cases also works up to covers. This
is done by adapting Lemma~\ref{lem:cons-preservv} to
the new context in Lemma~\ref{lem:cons-preserv-robust} below.
\end{proof}

The statement of the following lemma is almost identical
to its predecessor Lemma~\ref{lem:newtseitin},
the only difference being that the pairwise consistency
of the collection of $\mathbb{K}$-relations
is claimed up to every cover. We prove it by indicating how
the original arguments need to be adjusted.

\begin{lemma} \label{lem:acyclicnecessaryforminimalcover}
Let~$\mathbb{K}$ be a positive commutative monoid and
let~$X_1,\ldots,X_m$ be a schema that is~$k$-uniform and~$d$-regular
with~$k \geq 1$ and~$d \geq 2$. Then, there exists a collection
of~$\mathbb{K}$-relations of schema~$X_1,\ldots,X_m$ that is pairwise
consistent up to every cover of~$\mathbb{K}$ but not globally
consistent.
\end{lemma}

\begin{proof}
The construction of the $\mathbb{K}$-relations
$R_1,\ldots,R_m$ proceeds exactly as in Lemma~\ref{lem:newtseitin} 
until the point where it is argued that it is pairwise consistent. 
Here we need to show that it is pairwise consistent up to every cover of $\mathbb{K}$. Fix such a cover~$h : \mathbb{K}^* \to
  \mathbb{K}$ and argue as follows. 
  
  By the surjectivity of~$h$, there exists an
  element~$c^*$ of~$\mathbb{K}^*$ such that~$h(c^*) = c$. Since~$h$ is
  a homomorphism and~$c \not= 0$ in~$\mathbb{K}$, we have that also~$c^* \not= 0$
  in~$\mathbb{K}^*$. Let~$a^* := c^* + \cdots + c^*$ with $c^*$ 
  appearing $d^k$ times in the sum, which is computed in $\mathbb{K}^*$. 
  Using the notation $nx$ for $x + \cdots + x$ with $x$ appearing $n \geq 1$ times in the sum, which is computed in $\mathbb{K}$ or $\mathbb{K}^*$ depending on whether $x$ is an element of $\mathbb{K}$ or of $\mathbb{K}^*$, we have
  \begin{equation}
  h(a^*) = h(d^k c^*) = d^k h(c^*) = d^k c = a, \label{eqn:lifteq}
  \end{equation} 
  and~$a^* \not= 0$
  in~$\mathbb{K}^*$, again because~$h$ is a homomorphism and~$a \not=
  0$ in~$\mathbb{K}$. Next we define a collection~$R^*_1,\ldots,R^*_m$
  of $h$-lifts of~$R_1,\ldots,R_m$ by setting $R^*_i(t) = a^*$ for 
  every $X_i$-tuple $t$ such that $t \in R_i'$, and $R^*_i(t) = 0$ for every other $X_i$-tuple. By~\eqref{eqn:lifteq} we
  have~$h \circ R^*_i = R_i$, so $R^*_i$ is an $h$-lift of $R_i$. 
  The proof that the
  collection~$R^*_1,\ldots,R^*_m$ is
  pairwise consistent as a collection~$\mathbb{K}^*$-relations is
  identical to that in Lemma~\ref{lem:newtseitin} 
  for $R_1,\ldots,R_m$ but arguing with $c^*$ and $a^*$ in 
  $\mathbb{K}^*$ instead of arguing with $c$ and $a$ in $\mathbb{K}$.

  The proof that the collection~$R_1,\ldots,R_m$
  of~$\mathbb{K}$-relations is not globally consistent stays the same, which completes the proof.
\end{proof}

Next we argue that the two operations that transform an arbitrary  non-acyclic
hypergraph to a minimal one of the form $C_n$ with $n \geq 3$, or 
$H_n$ with $n \geq 4$,
preserve the same levels of
consistency up to a cover. The statement of the following lemma
is almost identical to that of Lemma~\ref{lem:cons-preservv}. To
prove it we will only indicate the differences in the arguments.

\begin{lemma} \label{lem:cons-preserv-robust} 
Let~$\mathbb{K}$ be a positive commutative monoid and let~$h :
\mathbb{K}^* \onto \mathbb{K}$ be a cover of~$\mathbb{K}$.  Let~$H_0$
and~$H_1$ be hypergraphs such that~$H_0$ is obtained from~$H_1$ by a
sequence of safe-deletion operations. For every collection~$D_0$
of~$\mathbb{K}$-relations over~$H_0$, there exists a collection~$D_1$
of~$\mathbb{K}$-relations over~$H_1$ such that, for every positive
integer~$k$, it holds that~$D_0$ is~$k$-wise consistent up to the cover 
if and only if~$D_1$ is~$k$-wise consistent up to the cover.
\end{lemma}

\begin{proof}
  The construction is the same as in Lemma~\ref{lem:cons-preservv}
  just that besides the $\mathbb{K}$-relations $R_i$ we also need
  to construct their $h$-lifts $R^*_i$ from the $h$-lifts $S^*_i$ of the $S_i$. Concretely, the argument is as follows.
  In an edge-deletion operation with the notation as in the proof of Lemma~\ref{lem:cons-preservv}, the lift $R^*_i$ associated to an $R_i$ with $X_i \not= X$ is taken as $R^*_i := S^*_i$, and that associated to the $R_i$ with $X_i = X$ is taken as~$R^*_i := S^*_j[X]$. In a vertex-deletion operation with the notation as in the proof of Lemma~\ref{lem:cons-preservv}, the lift~$R^*_i$ associated to an~$R_i$ with~$A \not\in X_i$ 
  is taken as~$R^*_i := S^*_i$, and that associated to an $R_i$ with $A \in X_i$ is defined by~$R^*_i(t) := S^*_i(t[X_i])$ if~$t(A) = u_0$ and~$R^*_i(t) := 0$ if~$t(A) \not= u_0$.  
  Observe that, in both cases, since~$h \circ S^*_i = S_i$ holds for all indices $i \in [m]$ for which $S_i$ and $S^*_i$ exist, 
  also~$h \circ R^*_i = R_i$ holds for all $i \in [m]$, so $R^*_1,\ldots,R^*_m$ are $h$-lifts of $R_1,\ldots,R_m$.

  With these definitions, the proof follows
  from Claims~\ref{claim:removeattribute} and~\ref{claim:removecoverededge} applied to the positive
  commutative monoid~$\mathbb{K}^*$ instead of~$\mathbb{K}$, and to
  the~collections of~$\mathbb{K}^*$-relations~$R^*_i$ and~$S^*_i$
  instead of the collections of~$\mathbb{K}$-relations~$R_i$
  and~$S_i$.
  \end{proof}

\section{Concluding Remarks} \label{sec:conclusions}

In this paper, we carried out a systematic investigation of the  interplay between local consistency
and global consistency for $\mathbb K$-relations, where ${\mathbb K}$ is a positive commutative monoid. In particular, we characterized the positive commutative monoids $\mathbb K$ for which a schema $H$ is acyclic if and only if
$H$ has the \ltgc~for $\mathbb K$-relations; this characterization was in terms of the inner consistency property, which is a semantic notion, and also in terms of the transportation property, which is a combinatorial notion. Furthermore, we  showed that, by strengthening the notion of pairwise consistency to pairwise consistency up to the free cover of $\mathbb K$, we can characterize the \ltgc~for collections of $\mathbb K$-relations on acyclic schemas for arbitrary positive commutative monoids. 

We conclude by describing a few open problems  motivated by
the work reported here.

As seen earlier, there are finite positive commutative monoids that have the transportation property (e.g., $\mathbb B$) and others that do not (e.g., ${\mathbb N}_2$). How difficult is it to decide whether or not a given finite positive commutative monoid $\mathbb K$ has the transportation property? Is this problem decidable or undecidable? The same question can be asked when the given monoid is \emph{finitely presentable}.
Note that the transportation property is defined using an infinite set of first-order axioms in the language of monoids. Thus, a related question is whether or not  the transportation property is finitely axiomatizable.

We  exhibited several classes of monoids that have the transportation property. In each case, we gave  an explicit construction or a procedure for finding a witness to the consistency of two consistent $\mathbb K$-relations. In some cases (e.g., when the monoid has an expansion to a semifield), there is a suitable join operation that yields a \emph{canonical} such witness. However, in some other cases (e.g., when the northwest corner method is used), no \emph{canonical} such witness seems to exist.  Is there a way to compare the different witnesses to consistency and classify them according to some desirable property, such as maximizing some
carefully chosen objective function?

Beeri et al.\ \cite{BeeriFaginMaierYannakakis1983} showed that
hypergraph acyclicity is  also equivalent to  semantic
conditions other than the \ltgc~for ordinary relations, such as the existence
of a \emph{full reducer}, which is a sequence of \emph{semi-join} operations for computing a witness to global 
consistency. 
Does an analogous result hold  for positive commutative monoids $\mathbb K$ that have the transportation property? The main difficulty is that it is not clear if a suitable semi-join operation on $\mathbb K$-relations can be defined for such monoids.

Finally, the work presented here expands the study of relations with annotations over semirings to relations with annotations over monoids. As explained in the Introduction, consistency notions only require the use of an addition operation (and not a multiplication operation). What other fundamental problems in databases can be studied in this  broader framework of relations with annotations over monoids?

\section*{Acknowledgments} 
The research of Albert Atserias was 
partially supported by grants PID2019-109137GB-C22 (PROOFS) and 
PID2022-138506NB-C22 (PROOFS BEYOND), and Severo Ochoa and María 
de Maeztu Program for Centers and Units of Excellence in 
R\&D (CEX2020-001084-M) of the AEI. The research of Phokion Kolaitis
was partially supported by NSF Grant IIS-1814152.

\bibliographystyle{alpha}
\bibliography{biblio}

\begin{thebibliography}{BFMY83}

\bibitem[ABU79]{DBLP:journals/tods/AhoBU79}
Alfred~V. Aho, Catriel Beeri, and Jeffrey~D. Ullman.
\newblock The theory of joins in relational databases.
\newblock {\em {ACM} Trans. Database Syst.}, 4(3):297--314, 1979.

\bibitem[AK21]{DBLP:conf/pods/AtseriasK21}
Albert Atserias and Phokion~G. Kolaitis.
\newblock Structure and complexity of bag consistency.
\newblock In Leonid Libkin, Reinhard Pichler, and Paolo Guagliardo, editors,
  {\em PODS'21: Proceedings of the 40th {ACM} {SIGMOD-SIGACT-SIGAI} Symposium
  on Principles of Database Systems, Virtual Event, China, June 20-25, 2021},
  pages 247--259. {ACM}, 2021.

\bibitem[AK23]{AtseriasK23}
Albert Atserias and Phokion~G.\ Kolaitis.
\newblock Acyclicity, consistency, and positive semirings.
\newblock In Alessandra Palmigiano and Mehrnoosh Sadrzadeh, editors, {\em
  Samson Abramsky on Logic and Structure in Computer Science and Beyond},
  Outstanding Contributions to Logic, pages 623--668. Springer, 2023.

\bibitem[Bel64]{bell1964einstein}
John~S Bell.
\newblock On the {E}instein-{P}odolsky-{R}osen paradox.
\newblock {\em Physics Physique Fizika}, 1(3):195, 1964.

\bibitem[Ber89]{Berge1989book}
Claude Berge.
\newblock {\em Hypergraphs - combinatorics of finite sets}, volume~45 of {\em
  North-Holland mathematical library}.
\newblock North-Holland, 1989.

\bibitem[BFMY83]{BeeriFaginMaierYannakakis1983}
Catriel Beeri, Ronald Fagin, David Maier, and Mihalis Yannakakis.
\newblock On the desirability of acyclic database schemes.
\newblock {\em J. ACM}, 30(3):479--513, July 1983.

\bibitem[Bir35]{birkhoff1935structure}
Garrett Birkhoff.
\newblock On the structure of abstract algebras.
\newblock In {\em Mathematical proceedings of the Cambridge philosophical
  society}, volume 31(4), pages 433--454. Cambridge University Press, 1935.

\bibitem[Bra16]{DBLP:journals/csur/Brault-Baron16}
Johann Brault{-}Baron.
\newblock Hypergraph acyclicity revisited.
\newblock {\em {ACM} Comput. Surv.}, 49(3):54:1--54:26, 2016.

\bibitem[BS81]{DBLP:books/daglib/0067494}
Stanley Burris and Hanamantagouda~P. Sankappanavar.
\newblock {\em A course in universal algebra}, volume~78 of {\em Graduate texts
  in mathematics}.
\newblock Springer, 1981.

\bibitem[GKT07]{DBLP:conf/pods/GreenKT07}
Todd~J. Green, Gregory Karvounarakis, and Val Tannen.
\newblock Provenance semirings.
\newblock In Leonid Libkin, editor, {\em Proceedings of the Twenty-Sixth {ACM}
  {SIGACT-SIGMOD-SIGART} Symposium on Principles of Database Systems, June
  11-13, 2007, Beijing, China}, pages 31--40. {ACM}, 2007.

\bibitem[Gre11]{DBLP:journals/mst/Green11}
Todd~J. Green.
\newblock Containment of conjunctive queries on annotated relations.
\newblock {\em Theory Comput. Syst.}, 49(2):429--459, 2011.

\bibitem[HLY80]{DBLP:journals/ipl/HoneymanLY80}
Peter Honeyman, Richard~E. Ladner, and Mihalis Yannakakis.
\newblock Testing the universal instance assumption.
\newblock {\em Inf. Process. Lett.}, 10(1):14--19, 1980.

\bibitem[KHF14]{KunjwalHeunenFritz2014}
Ravi Kunjwal, Chris Heunen, and Tobias Fritz.
\newblock Quantum realization of arbitrary joint measurability structures.
\newblock {\em Phys. Rev. A}, 89:052126, May 2014.

\bibitem[KRS14]{DBLP:journals/tods/KostylevRS14}
Egor~V. Kostylev, Juan~L. Reutter, and Andr{\'{a}}s~Z. Salamon.
\newblock Classification of annotation semirings over containment of
  conjunctive queries.
\newblock {\em {ACM} Trans. Database Syst.}, 39(1):1:1--1:39, 2014.

\bibitem[MI08]{MiyaderaImai2008}
Takayuki Miyadera and Hideki Imai.
\newblock Heisenberg's uncertainty principle for simultaneous measurement of
  positive-operator-valued measures.
\newblock {\em Phys. Rev. A}, 78:052119, Nov 2008.

\bibitem[RTL76]{RoseTarjanLueker1976}
Donald~J. Rose, Robert~Endre Tarjan, and George~S. Lueker.
\newblock Algorithmic aspects of vertex elimination on graphs.
\newblock {\em {SIAM} J. Comput.}, 5(2):266--283, 1976.

\bibitem[Sch86]{Schrijver-book}
Alexander Schrijver.
\newblock {\em Theory of Linear and Integer Programming}.
\newblock John Wiley \& Sons, Inc., USA, 1986.

\bibitem[Tse68]{Tseitin1968}
G.~S. Tseitin.
\newblock On the complexity of derivation in propositional calculus.
\newblock {\em Structures in Constructive Mathematics and Mathematical Logic},
  pages 115--125, 1968.

\bibitem[Ull82]{DBLP:conf/pods/Ullman82}
Jeffrey~D. Ullman.
\newblock The u. r. strikes back.
\newblock In Jeffrey~D. Ullman and Alfred~V. Aho, editors, {\em Proceedings of
  the {ACM} Symposium on Principles of Database Systems, March 29-31, 1982, Los
  Angeles, California, {USA}}, pages 10--22. {ACM}, 1982.

\bibitem[Vor62]{vorob1962consistent}
Nikolai~Nikolaevich Vorob’ev.
\newblock Consistent families of measures and their extensions.
\newblock {\em Theory of Probability \& Its Applications}, 7(2):147--163, 1962.

\end{thebibliography}

\end{document}